\documentclass[12pt,notitlepage]{article}%
\usepackage[margin=1in]{geometry} 
\usepackage[toc,title,page]{appendix}
\usepackage{amsthm}
\usepackage{setspace}
\usepackage{eurosym}
\usepackage{subfigure}
\usepackage{xcolor}
\usepackage{url}
\usepackage{supertabular}
\usepackage{amssymb}
\usepackage{graphicx}
\usepackage[colorlinks,linkcolor=links,citecolor=cites,urlcolor=MyDarkBlue]%
{hyperref}
\usepackage{dsfont}
\usepackage{amsfonts}
\usepackage{amsmath}
\usepackage{amstext}
\usepackage{mathtools}
\usepackage{tikz}
\usepackage{rotating}
\usepackage{lscape}
\usepackage{chngpage}
\usepackage{esint}
\usepackage{natbib}
\usepackage{bm}
\usepackage{extarrows}
\usepackage[utf8]{inputenc}
\usepackage{calc}
\usepackage{accents}
\usepackage{booktabs}
\usepackage{abstract}
\usepackage{etoolbox}
\usetikzlibrary{arrows.meta, decorations.pathreplacing, calc}
 
\definecolor{nodeblue}{RGB}{155,194,230}
\definecolor{nodegreen}{RGB}{146,208,152}
\definecolor{nodeyellow}{RGB}{255,230,153}
\definecolor{nodegray}{RGB}{217,217,217}

\AtBeginEnvironment{footnote}{\singlespacing}
\AtBeginEnvironment{thebibliography}{\singlespacing}

\usetikzlibrary{arrows.meta, decorations.pathreplacing, calc}

\providecommand{\U}[1]{\protect\rule{.1in}{.1in}}

\newtheorem{theorem}{Theorem}

\newtheorem{claim}{Claim}

\newtheorem{definition}{Definition}
\newtheorem{example}{Example}

\newtheorem{lemma}{Lemma}

\newtheorem{proposition}{Proposition}
\newtheorem{remark}{Remark}

\newtheorem{fact}{Fact}
\newtheorem*{method}{Methodology}

\DeclareMathOperator*{\argmax}{argmax}

\definecolor{MyDarkBlue}{rgb}{0,0.08,0.45}
\definecolor{cites}{HTML}{324b13}
\definecolor{links}{HTML}{1a663b}
\definecolor{MyLightMbuyera}{cmyk}{0.1,0.8,0,0.1}
\hypersetup{
colorlinks,citecolor=blue,filecolor=black,linkcolor=blue,urlcolor=blue
}

\begin{document}

\title{Optimal linear-payment auction design with aftermarket collaboration\thanks{We would like to thank James Bergin, Xun Chen, Yunan Li, Tingjun Liu, Jingfeng Lu, Ning Sun, Guofu Tan, Mengxi Zhang, Jun Zhang, Charles Zheng, Yu Zhou, and all participants in the 2023 Asian Meeting of the Econometric Society, the 2025 World Congress of the Econometric Society and the 7th World Congress of the Game Theory Society for their helpful comments. All remaining errors are ours.}}
\author{\textsc{Dazhong Wang}\thanks{School of Digital Economics and Management, Nanjing University, No.1520 Taihu Avenue, Suzhou, Jiangsu 215163, China. E-mail: \textit{wangdazhong@nju.edu.cn}.} \and \textsc{Ruqu Wang}\thanks{Department of Economics, Queen's University. E-mail: \textit{wangr@queensu.ca}.}\and \textsc{Xinyi Xu}\thanks{\textit{Corresponding author.} Lingnan College, Sun Yat-Sen University, Guangzhou, Guangdong, China. E-mail: \textit{xuxinyi3@mail.sysu.edu.cn}.}}
\date{}
\maketitle

{\linespread{1}
\begin{abstract}

This paper studies optimal auction design when valuations depend endogenously on post-auction collaboration between the seller and the winning bidder. Both parties exert non-contractible efforts after the auction, generating a double moral hazard problem alongside adverse selection. We analyze two role structures---winner-pivotal and seller-pivotal collaboration---and characterize optimal direct mechanisms using linear payment schemes that combine cash transfers with proportional value sharing. The optimal mechanism allocates the asset to the bidder with the highest virtual surplus, employs a deterministic value-sharing rule, and achieves full type revelation through the signal realization rule. Comparing the two scenarios yields three main findings. First, regarding value sharing, the seller secures a strictly higher share under seller-pivotal collaboration: for sufficiently low-type winners, the seller extracts the entire value, whereas under winner-pivotal collaboration every winner must retain a positive share to sustain his critical effort. Second, regarding effort exertion, the pivotal party always exerts higher post-auction effort than the supporting party, and each party exerts greater effort when pivotal than when providing support. Third, seller-pivotal collaboration yields strictly higher seller revenue than winner-pivotal collaboration for any type distribution. Finally, these optimal mechanisms can be implemented through ascending auctions with endogenously determined linear contracts.

	
\end{abstract}
}



\newpage

\section{Introduction}

\subsection{Motivation}

Auction designers often sell assets or projects whose valuations endogenously depend on subsequent collaboration between the seller and winning bidder. Examples include franchising (potential franchisees bid for licensing rights and collaborate with franchisors in brand development), patent licensing (firms compete for technology licenses requiring ongoing collaboration in technical development), and infrastructure concessions (construction firms compete for public projects with post-auction collaboration determining project success).

These settings share a common feature: sellers employ auctions to select partners for post-auction collaboration, and this collaboration exhibits \textit{pivotal role divergence}. Either the seller or the winning bidder assumes the pivotal role in value creation. In franchising and infrastructure concessions, the winning bidder typically plays the pivotal role through operational responsibility. In patent licensing, the licensor's continued research efforts are often essential for value creation. This pivotal role divergence determines which party's effort is critical for generating value. We illustrate this through four examples.

\begin{example}[Pharmaceutical Drug Development]\label{eG1}
A biotechnology firm (seller) auctions commercialization rights for a drug compound to pharmaceutical companies (bidders). The winning company plays the pivotal role by directing mass production, regulatory approval, and marketing strategies. Value creation depends critically on the pharmaceutical company's commercialization efforts; without these efforts, the drug generates no market value. The biotechnology firm provides supporting technical data but cannot substitute for the winner's pivotal role. This exemplifies winner-pivotal collaboration.
\end{example}

\begin{example}[TikTok U.S. Operating Rights]\label{eG4}
ByteDance (seller) auctions operating rights for TikTok's U.S.\ business to American technology companies (bidders) through an informal competitive process. The winning company plays the pivotal role by managing user acquisition, content moderation, and advertiser relationships; without these operational efforts, the platform generates no commercial value. ByteDance provides algorithmic technology and ongoing technical support as the supporting party. This exemplifies winner-pivotal collaboration.
\end{example}

\begin{example}[University Technology Licensing]\label{eG2}
A research university (seller) auctions licensing rights for artificial intelligence algorithms to technology companies (bidders). The university plays the pivotal role by providing continuous algorithm refinements and research insights that determine commercial viability. Without the university's continued scientific development, the licensed technology becomes obsolete regardless of the company's implementation capabilities. This exemplifies seller-pivotal collaboration.
\end{example}

\begin{example}[Telecommunications Spectrum License Auction]\label{eG3}
A national telecommunications regulator (seller) auctions spectrum licenses to mobile network operators (bidders). In underserved rural markets, the regulator's continued contributions---subsidizing backbone infrastructure and coordinating utility access---are pivotal: without this support, commercially viable service is infeasible regardless of operator capability. This exemplifies seller-pivotal collaboration. In dense urban markets, by contrast, backbone infrastructure is already in place, and commercial success depends on the operator's capital deployment and marketing. The regulator provides the license and policy framework in a supporting role, exemplifying winner-pivotal collaboration.
\end{example}

These examples illustrate a common pattern: one party's effort is indispensable for value creation---if this party exerts zero effort, total value collapses to zero. We call this party the \textit{pivotal} party and model aftermarket value as Value $=$ $($Inherent productivity $+$ Support effort$)$ $\times$ Pivotal effort. We analyze two scenarios: \textbf{winner-pivotal} and \textbf{seller-pivotal} collaboration. In winner-pivotal scenarios, the winner's effort is pivotal, and productivity captures the synergy of the winning bidder's operational capabilities, determined by his inherent characteristics and enhanced by the seller's supporting contributions (Examples \ref{eG1} and \ref{eG4}). In seller-pivotal scenarios, the seller's effort is pivotal, and productivity reflects the winner's intrinsic implementation capacity, strengthened through the winner's resource investments (Example \ref{eG2}). The pivotal role need not be fixed: as Example \ref{eG3} illustrates, the same seller can be pivotal in one market segment and supportive in another.

\subsection{Results}

Both scenarios feature double moral hazard alongside adverse selection. Post-auction actions are non-contractible, and value is generated after the auction. Standard one-shot cash payments are insufficient to incentivize post-auction collaboration. Optimal payment schemes must combine fixed cash transfers with value-sharing arrangements contingent on realized outcomes. We focus on linear payments---a combination of fixed cash transfers and proportional value sharing---which are tractable and widely used in practice \citep[see, e.g.,][]{Stiglitz1974,Weitzman1980,arrow1985economics,Gibbons1992}.\footnote{When adverse selection and double-sided moral hazard coexist, the optimal mechanism design approach is typically inapplicable for non-linear contracts.} The winner compensates the seller through such a linear combination. Auctions with linear payments under double moral hazard remain largely unexplored in the existing literature.

We characterize optimal auctions with aftermarket collaboration under both scenarios. The seller commits to an auction rule comprising bidding patterns and linear payment schemes. In perfect Bayesian equilibrium, the seller and winner choose post-auction efforts that maximize their payoffs based on public signals revealed during the auction. By the revelation principle \citep{myerson1982}, we restrict attention to direct mechanisms with linear payments. These ``direct linear mechanisms'' generate allocation rules, payment rules (linear payments for winners, cash payments for losers), and signal realization rules conditional on reported types. Mechanism feasibility requires incentive-compatibility (truthful reporting constitutes a Bayesian Nash equilibrium) and individual rationality (non-negative expected payoffs).

To characterize the optimal mechanism, we first establish aftermarket equilibrium strategies, which differ across leadership scenarios. We then determine the upper bound of the seller's expected payoff across all feasible direct linear mechanisms, and present a feasible mechanism achieving this upper bound.

The optimal mechanism allocates the asset to the bidder with the highest virtual surplus rather than the highest type, where virtual surplus is non-negative and increasing in type. When bidders are non-identically distributed, hazard rates enter differently into each bidder's virtual surplus, generating discriminatory allocation. Under both scenarios, the optimal signal realization rule reveals the winner's type, achieving full information revelation. The optimal payment rule specifies a deterministic value-sharing rate and cash transfers satisfying a binding expected-payment condition.

Comparing optimal mechanisms across different pivotal structures reveals three main findings. First, regarding \textit{value sharing}, under seller-pivotal collaboration, the seller consistently secures a majority share (exceeding half) of the eventual value, with a pooling region where the seller extracts the entire value from sufficiently low-type winners. Under winner-pivotal collaboration, even the lowest-type winner must receive a positive value share to prevent value collapse. Consequently, the seller always obtains a higher value share under seller-pivotal collaboration. Second, concerning \textit{effort exertion}, the pivotal party consistently exerts higher effort: under winner-pivotal scenarios, the winner exerts higher effort than the seller; under seller-pivotal scenarios, the seller exerts higher effort than the winner. Each party also chooses higher effort when pivotal than when providing support. Third, regarding \textit{seller's revenue}, seller-pivotal collaboration yields higher revenue than winner-pivotal collaboration for any distribution. This revenue advantage stems from reduced information rents rather than higher social surplus. Under seller-pivotal collaboration, the seller's effort alone sustains value creation, allowing the mechanism to extract surplus from high types without leaving excessive rents. Under winner-pivotal collaboration, the winner's effort is essential, requiring the mechanism to compensate the winner more generously, resulting in higher information rent costs.

We implement optimal mechanisms through ascending auctions combining cash bids and linear contracts. The winner's linear contract---comprising value share and cash payment---is determined endogenously by his bid and the penultimate bidder's drop-out price. Perfect implementation is unattainable in seller-pivotal scenarios because the seller can extract the entire value by imposing full value-sharing on bidders below type thresholds. To approximate optimality, we allow sellers to retain some value share for winning bidders with extremely low types. In equilibrium, bidders adopt ``quasi-dominant strategies": each bidder exits at his true virtual surplus regardless of others' strategies.

We extend our analysis to examine effort substitution and a generalized value creation function nesting both role structures via an interdependence parameter. Under effort substitution, similar optimal mechanism structures emerge with full value extraction, since winner efforts do not drive value creation. Under the generalized nesting structure, we characterize sufficient conditions for the optimal mechanism to exist and show that the seller optimally chooses her own effort to be more critical, consistent with the revenue advantage of seller-pivotal collaboration.

\subsection{Related literature}

This study builds on the literature on auctions with security-form payments. \cite{DeMarzoetal2005} establish that steeper securities yield higher seller revenue in competitive auctions. \cite{GorbenkoMalenko2011} analyze competition among sellers in securities auctions and show that seller competition pushes equilibrium security designs toward flatter instruments. Subsequent work characterizes optimal equity and royalty auctions across a range of environments: \cite{Liu2016} for heterogeneous bidders, \cite{LIU2019} for two-dimensional private types, \cite{WANG2022} for interdependent valuations, and \cite{BERNHARDT2020} for settings with endogenous entry and fixed royalty rates. \cite{Abhishek2013} show that royalty auctions dominate their cash counterparts when valuations are interdependent. \cite{LiuJF2021} establish full surplus extraction via equity-plus-cash mechanisms when bidders' outside options relate concavely to their valuations. \cite{Sogo2017} shows that in equity auctions requiring post-auction investment, the seller may benefit from withholding information because equity payments distort the winner's investment incentives. \cite{WANG2023} compare standard royalty auctions in which bids signal and affect the seller's post-auction effort. \cite{Liyunan2023} characterize optimal linear mechanisms with pre-auction information acquisition. While \cite{Sogo2017} and \cite{WANG2023} incorporate one-sided post-auction moral hazard, none of these papers address bilateral moral hazard between the seller and winner. Our paper extends this literature by embedding a double moral hazard problem into the optimal security-bid auction framework: the seller's non-contractible post-auction effort shapes both the optimal value-sharing rate and the allocation rule.

We contribute to the literature on auction design with bidder moral hazard. \cite{McAfee1986} derive optimal cost-sharing contracts in competitive procurement where the winner exerts effort to reduce costs. \cite{Tirole1987} show that the optimal incentive auction can be implemented as a dominant-strategy mechanism with transfers that depend on the first and second bids, and establish the optimality of linear contracts when realized costs are non-contractible. \cite{McAfee1987} characterize conditions under which linear direct mechanisms are optimal when adverse selection and bidder moral hazard coexist. We extend these results by introducing seller moral hazard: the seller's active role in post-auction value creation means she cannot credibly commit to her effort level, requiring a mediator to restore mechanism tractability. Our setting also relates to the broader literature on post-auction interactions. \cite{Goeree2003} shows that aftermarket signaling biases bids upward and breaks the strategic equivalence between second-price and English auctions; \cite{Haile03} shows that resale opportunities introduce common-value elements that reshape bidder valuations. \cite{Giovannoni2014} study how reputational concerns arising from aftermarket competition distort bidding across disclosure rules. \cite{Piotr2020} characterizes optimal cutoff mechanisms when the designer's involvement extends beyond the auction, and \cite{Audrey2024} study how anticipated aftermarket competition shapes the optimal allocation rule. \cite{BosPollrich2025} analyze how bidders' signaling incentives toward an outside receiver shape optimal auction design and information disclosure. We share the concern with post-auction value creation but differ in that collaboration between seller and winner is cooperative rather than competitive, and involves bilateral moral hazard.

Our research relates to contract design under double moral hazard. \cite{Eswaran1985} first formalize bilateral moral hazard in agricultural share contracts, showing how linear sharing rules balance incentives under joint production. \cite{Romano1994} and \cite{Bhattacharyya1995} establish conditions under which linear sharing contracts are optimal when both contracting parties face moral hazard. \cite{KIM1998} show that the optimality of linear sharing is not robust to risk aversion: the unique optimal contract under double moral hazard with a risk-averse agent is non-linear. \cite{CarrollTE} demonstrate the robustness of linear contracts under double moral hazard: among contracts satisfying monotonicity and limited liability, linear contracts perform optimally across a broad class of effort distributions. We embed this structure into an auction mechanism, linking the linear-contract literature to optimal mechanism design. The linear-quadratic specification adopted here---bilinear value functions with quadratic effort costs---extends the canonical principal-agent framework of \cite{Holm1987} and \cite{Holm1991} to a bilateral effort setting (see \cite{TheoryofIncentive2002} and \cite{ContractTheory2005} for comprehensive treatments). Value-sharing and ad valorem arrangements of the kind studied here are empirically prevalent, as documented by \cite{Rao1995} and \cite{Hagiu2019}.

Since the seller cannot commit to her post-auction effort, our problem aligns with mechanism design with a form of limited commitment. \cite{Laffont1988} first demonstrate that the standard revelation principle fails in such settings. \cite{myerson1982} establishes the communication revelation principle, showing that outcomes achievable through any communication system can be replicated by an incentive-compatible direct mechanism operated by a mediator. We build on this insight: a trustworthy mediator manages the auction, generates post-auction public signals, and restores the tractability of the direct revelation approach. \cite{Bester2007} show that a mediator controlling the principal's observations resolves the commitment problem under single-crossing conditions; our framework adopts a similar structure. \cite{Skreta2015} characterizes optimal auctions for a sequentially rational seller who allocates to the highest positive virtual-surplus buyer conditional on updated beliefs. \cite{Doval2022} study a dynamic mechanism-selection game under limited commitment and characterize the set of direct mechanisms that replicate any equilibrium outcome; their signal-space construction directly informs our treatment of the outcome-equivalent direct mechanism.

This paper is organized as follows: Section \ref{S2} introduces the model, including functional forms representing pivotal role divergence, linear payment schemes, and aftermarket equilibrium. Section \ref{mechanimschapter} characterizes optimal mechanisms under both scenarios and compares their features. Section \ref{Secimp} presents implementations. Section \ref{Diss} discusses extensions. Section \ref{S6} concludes. Omitted proofs are in the \hyperref[APPPPPPPPPPPP]{Appendices}.

\section{The model \label{S2}}

A risk-neutral seller auctions an indivisible asset or a project to $n$ risk-neutral bidders. Let $\mathcal{N} = \{1,2,\ldots,n\}$ denote the set of bidders. Each bidder $i \in \mathcal{N}$ has a private type $\theta_i$ that represents the bidder's inherent productivity and ability, independently and identically drawn from a continuously differentiable distribution $F_i(\cdot)$ with density $f_i(\cdot)>0$ on support $\Theta_i \equiv [\underline{\theta}_i, \overline{\theta}_i]$, where $\underline{\theta}_i \geq 0$. We assume that $\underline{\theta}_i \cdot f_i(\underline{\theta}_i) < 1$ and that the hazard rate $\frac{f_i(\cdot)}{1 - F_i(\cdot)}$ is increasing for all $i \in \mathcal{N}$.\footnote{The increasing hazard rate condition is satisfied when the density is log-concave, a standard assumption in auction theory. This holds for many common distributions, including uniform, normal, and exponential distributions \citep[see][]{BB05}.}

We denote $\bm{\theta} \equiv (\theta_1, \theta_2, \dots, \theta_n)$, $\bm{\theta}_{-i} \equiv (\theta_1, \dots, \theta_{i-1}, \theta_{i+1}, \dots, \theta_n)$, $\bm{\Theta} \equiv \bigtimes_{i \in \mathcal{N}} \Theta_i$, and $\bm{\Theta}_{-i} \equiv \bigtimes_{j \in \mathcal{N}, j \neq i} \Theta_j$ as the profile of bidders' types and their corresponding supports. The joint densities are $\bm{f}(\bm{\theta}) \equiv \prod_{i \in \mathcal{N}} f_i(\theta_i)$ and $\bm{f}_{-i}(\bm{\theta}_{-i}) \equiv \prod_{j \in \mathcal{N}, j \neq i} f_j(\theta_j)$. We use tildes to distinguish random variables from their realizations: $\widetilde{\theta}_i$ denotes the random variable and $\theta_i$ denotes its realization. All parties have no time discounting.

\subsection{Value creation: pivotal role divergence}

The asset's value is created through aftermarket collaboration between the seller and winning bidder. Both parties simultaneously choose non-contractible efforts $e_i^s \in E_i^s \subseteq \mathbb{R}_+$ and $e_i^b \in E_i^b \subseteq \mathbb{R}_+$, respectively. These efforts, together with the winner's type $\theta_i$, determine the asset's value in a multiplicative form:
\begin{equation}\label{REV-generallll}
\text{Value $(V_i)$} = \left(\text{Inherent Productivity} + \text{Support Effort}\right) \times \text{Pivotal Effort},
\end{equation}
where one party plays the pivotal role while the other provides support, as illustrated by Examples~\ref{eG1}--\ref{eG3}. We analyze two cases characterized by this pivotal role divergence.

\paragraph{Winner-pivotal:} The ``Pivotal Effort'' term in \eqref{REV-generallll} is $e_i^b \in E^b_i$. The ``Inherent Productivity'' captures the winner's asset-specific production capability, determined by his type and enhanced by the seller's supporting effort $e_i^s$. The value function becomes:
\begin{equation}\label{Vstructure-1}
V_i(\theta_i,e_i^s, e_i^b) = (\theta_i + e_i^s) \cdot e_i^b.\footnote{In reality, the value creation function involves a random term $\widetilde{\epsilon}$, i.e., $\widetilde{V}_i = (\theta_i + e_i^s) \cdot e_i^b + \widetilde{\epsilon}$, rendering efforts unverifiable. However, the random term does not affect the linear payment scheme design (Section \ref{LPsec}). We omit it for simplicity.}
\end{equation}

\paragraph{Seller-pivotal:} The ``Pivotal Effort'' term in \eqref{REV-generallll} is $e_i^s \in E^s_i$. The ``Inherent Productivity'' reflects the winner's intrinsic ability to operate and promote the asset, enhanced by his supporting effort $e_i^b$. The value function becomes:
\begin{equation}\label{Vstructure-2}
V_i(\theta_i,e_i^s, e_i^b) = (\theta_i + e_i^b) \cdot e_i^s.
\end{equation}

In both scenarios, post-auction efforts incur quadratic costs of $\frac{1}{2}e^2$ for $e\in\{e_i^s, e_i^b\}$.\footnote{The quadratic cost function is standard in the principal-agent moral hazard literature \citep[see, e.g.,][]{Rao1995,Hagiu2019}.}

\begin{fact}[Pivotal role]\label{leadeffectDEF}
As shown in \eqref{Vstructure-1} and \eqref{Vstructure-2}, the ``pivotal role'' arises when one party's effort determines whether any value is created: if this party exerts zero effort, total value is zero regardless of the other party's contribution.
\end{fact}

The involvement of both the seller's and the winning bidder's non-contractible post-auction efforts in aftermarket collaboration generate a double moral hazard problem alongside adverse selection. We analyze these two scenarios separately.

\subsection{Linear-payment auction \label{LPsec}}

Since value is generated after the auction, one-shot cash payments are unable to incentivize post-auction investment. To provide adequate incentives for aftermarket collaboration, the seller and winner must adopt a value-sharing arrangement based on the realized value in addition to any cash payment during the auction. Therefore, the payment mode that the seller commits to must involve both cash and value sharing. 

For tractability and practical operability, we consider a linear payment scheme.\footnote{Linear payments and contracts are standard in the mechanism design literature.} Formally, we denote the linear payment scheme for winning bidder $i$ as:
\begin{equation*}
c_i = (\alpha_i, t_i^w),
\end{equation*}
where $\alpha_i \in [0,1]$ is the share of realized value that the winner transfers to the seller, and $t_i^w \in \mathcal{T}_i^w \subseteq \mathbb{R}_+$ is the winner's cash payment. Losing bidders may pay a fixed amount $t_j^l \in \mathcal{T}_j^l \subseteq \mathbb{R}_+$ for $j \neq i$, depending on the auction format. Losing bidders do not pay any value share.\footnote{The analysis follows \cite{LiuJF2021}, which examines scenarios where winners pay with equity and cash combinations.}

The seller may commit to a cash-bid auction, though other bidding formats are possible, such as bids on value-sharing rates or direct bidding on the linear payment scheme.\footnote{As noted in \cite{DeMarzoetal2005}, bidding formats can be diversified. The optimal mechanism can be implemented through a cash-bid auction with a linear payment scheme (Section \ref{Secimp}).} The seller determines a selection rule based on bids, selects a winner, and collects the linear payment after value realization.\footnote{Without time discounting, bidders are indifferent between making cash payments before or after aftermarket collaboration.}

For the value creation functions \eqref{Vstructure-1} and \eqref{Vstructure-2}, under the linear payment scheme $c_i = (\alpha_i, t_i^w)$, the seller's ex-post expected payoff is:
\begin{equation}\label{expostpayoff-seller}
\widehat{\text{Rev}} = \underbrace{t_i^w + \sum_{j\neq i}t_j^l}_{\text{seller's cash income}} + \underbrace{\alpha_i \cdot V_i(\theta_i,e_i^s, e_i^b) -\frac{1}{2}\cdot \left(e_i^s\right)^2}_{\text{seller's value-sharing-related payoff}},
\end{equation}
The winning bidder $i$'s ex-post expected payoff is:
\begin{equation}\label{expostpayoff-winner}
\hat{U}_i = - t_i^w + \underbrace{\left(1-\alpha_i\right) \cdot V_i(\theta_i,e_i^s, e_i^b) - \frac{1}{2}\cdot \left(e_i^b\right)^2}_{\text{winner's value-sharing-related payoff}},
\end{equation}
and $\hat{U}_j = -t_j^l$ for all $j\neq i$.

\subsection{Aftermarket equilibrium}

The interaction is a dynamic game where the seller commits to an auction rule (bidding pattern and linear payment scheme) but cannot commit to her post-auction effort. The seller and winner simultaneously choose efforts during aftermarket collaboration. In perfect Bayesian equilibrium (if any), both parties choose efforts to maximize their value-sharing payoffs (value shares net of effort costs, as in \eqref{expostpayoff-seller} and \eqref{expostpayoff-winner}). These equilibrium strategies depend on signals revealed during the auction. The seller relies on signal realizations to infer the winner's type.\footnote{Unlike classical signal/information disclosure literature \citep[see, e.g.,][]{Piotr2020,Audrey2024}, which considers disclosure to third parties, our setting involves disclosure to the seller who participates in value creation. If a separating equilibrium exists, the seller can infer the winner's type. If the auction induces pooling or semi-pooling equilibria, the seller's posterior beliefs cannot achieve perfect inference. Determining separating equilibrium existence requires conditions on bidders' valuations (such as single-crossing conditions). However, such conditions assume exogenous valuations, whereas in our setting, valuations are endogenously determined by non-contractible efforts, making it impossible to confirm separating equilibrium existence at the auction stage.}

Following the literature \citep[e.g.,][]{Genzkow2011,Piotr2020,Audrey2024}, signals are publicly announced after the auction.\footnote{The signal space is a subspace of the bid space, denoted $\bm{S} \subseteq \bm{B}$. For example, when $\bm{S} = \bm{B}$, this represents auctions with ex-post bid disclosure, commonly mandated for credibility, as in open-outcry ascending auctions and spectrum auctions.} This ensures the seller's posterior beliefs are common knowledge during aftermarket collaboration.\footnote{If the seller's posterior beliefs were private information, characterizing equilibrium efforts during aftermarket collaboration would be unsolvable.} For any public signal realization $\bm{s} \in \bm{S}$, the seller forms posterior beliefs $\mu(\cdot | \bm{s}) \in \Delta(\Theta_i)$ over winning bidder $i$'s type, which is common knowledge.

In perfect Bayesian equilibrium (if any), given the winner's identity $i \in \mathcal{N}$, each party's equilibrium strategy $\varepsilon^s_i \in \Delta( E^s_i)$ and $\varepsilon^b_i \in \Delta( E^b_i)$ satisfies (from \eqref{expostpayoff-seller} and \eqref{expostpayoff-winner}):
\begin{equation}\label{effortsderive}
\begin{split}
\varepsilon^s_i =  & \varepsilon^s_i(\mu(\cdot | \bm{s})) \in \argmax_{\hat{\varepsilon}^s_i \in \Delta( E^s_i)}\left\{\mathbb{E}_{\hat{\varepsilon}^s_i}\left[\alpha_i\cdot \mathbb{E}_{\widetilde{\theta}_i}\left[\left.\mathbb{E}_{\varepsilon^b_i(\widetilde{\theta}_i)}\left[V_i(\widetilde{\theta}_i, e_i^s, e^b_i)\right]\right|\mu(\cdot | \bm{s})\right] - \frac{1}{2}\cdot \left(e_i^s\right)^2\right]\right\},\\
\varepsilon^b_i = & \varepsilon^b_i(\theta_i, \mu(\cdot | \bm{s})) \in  \argmax_{\hat{\varepsilon}^b_i \in \Delta( E^b_i)}\left\{\mathbb{E}_{\hat{\varepsilon}^b_i}\left[\left(1-\alpha_i\right)\cdot \mathbb{E}_{\varepsilon^s_i(\mu)}\left[V_i(\theta_i, e_i^s, e^b_i)\right] - \frac{1}{2}\cdot \left(e^b_i\right)^2\right]\right\},
\end{split}
\end{equation}
which maximize the net value-sharing payoffs of both parties simultaneously. The timeline of interaction between the seller and bidders can be summarized in the following figure:
\vspace{1em}

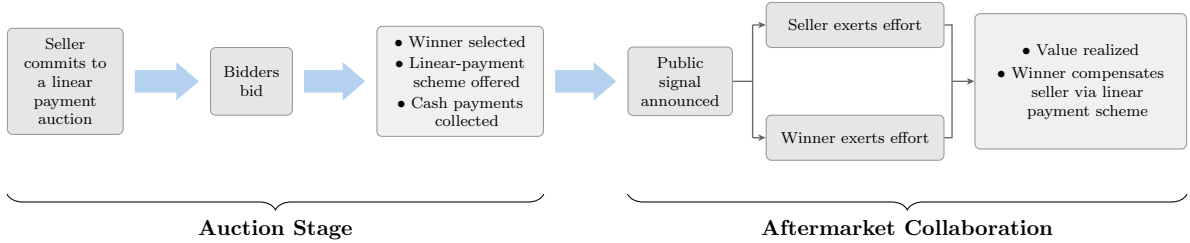
\begin{figure}[htp!]
\centering
\begin{tikzpicture}[font=\scriptsize, scale=0.75, every node/.style={scale=0.75}]
 
\definecolor{lightblue}{RGB}{180,210,240}
\definecolor{lgray}{RGB}{230,230,230}
 
\newcommand{\arrgap}{0.2}
\newcommand{\arrh}{0.20}
\newcommand{\arrtip}{0.34}
\newcommand{\arrtipx}{0.40}
 
\tikzset{
  basenode/.style={draw=black!40, rounded corners=2pt, fill=lgray,
    align=center, inner sep=5pt, font=\scriptsize},
  nd1/.style={basenode, text width=1.7cm, minimum height=1.2cm},
  nd2/.style={basenode, text width=1.1cm, minimum height=1.2cm},
  nd3/.style={basenode, fill=lgray!60, text width=2.6cm, minimum height=1.8cm, align=center},
  nd4/.style={basenode, text width=1.5cm, minimum height=1.2cm},
  nd5/.style={basenode, text width=2.8cm, minimum height=0.8cm},
  nd6/.style={basenode, fill=lgray!60, text width=3.4cm, minimum height=2.4cm, align=center},
  darkarr/.style={-{Stealth[length=3pt,width=2.5pt]}, line width=0.5pt, black!60},
  darkline/.style={line width=0.5pt, black!60},
}
 
\node[nd1] (B1) at (0, 0) {Seller commits to\\a linear payment\\auction};
\node[nd2] (B2) at (3.3, 0) {Bidders\\bid};
\node[nd3] (B3) at (7.0, 0) {$\bullet$~Winner selected\\[1pt]
$\bullet$~Linear-payment\\
\hspace{0.8em}scheme offered\\[1pt]
$\bullet$~Cash payments\\
\hspace{0.8em}collected};
\node[nd4] (B4) at (10.9, 0) {Public signal\\announced};
 
\path let \p1=(B1.east), \p2=(B2.west) in
  \pgfextra{
    \pgfmathsetmacro{\axs}{\x1/1cm + \arrgap}
    \pgfmathsetmacro{\axe}{\x2/1cm - \arrgap}
    \pgfmathsetmacro{\axb}{\axe - \arrtipx}
    \fill[lightblue] (\axs, \arrh) -- (\axb, \arrh) -- (\axb, \arrtip) -- (\axe, 0) -- (\axb, -\arrtip) -- (\axb, -\arrh) -- (\axs, -\arrh) -- cycle;
  };
 
\path let \p1=(B2.east), \p2=(B3.west) in
  \pgfextra{
    \pgfmathsetmacro{\axs}{\x1/1cm + \arrgap}
    \pgfmathsetmacro{\axe}{\x2/1cm - \arrgap}
    \pgfmathsetmacro{\axb}{\axe - \arrtipx}
    \fill[lightblue] (\axs, \arrh) -- (\axb, \arrh) -- (\axb, \arrtip) -- (\axe, 0) -- (\axb, -\arrtip) -- (\axb, -\arrh) -- (\axs, -\arrh) -- cycle;
  };
 
\path let \p1=(B3.east), \p2=(B4.west) in
  \pgfextra{
    \pgfmathsetmacro{\axs}{\x1/1cm + \arrgap}
    \pgfmathsetmacro{\axe}{\x2/1cm - \arrgap}
    \pgfmathsetmacro{\axb}{\axe - \arrtipx}
    \fill[lightblue] (\axs, \arrh) -- (\axb, \arrh) -- (\axb, \arrtip) -- (\axe, 0) -- (\axb, -\arrtip) -- (\axb, -\arrh) -- (\axs, -\arrh) -- cycle;
  };
 
\node[nd5] (B5a) at (14.0,  1.0) {Seller exerts effort};
\node[nd5] (B5b) at (14.0, -1.0) {Winner exerts effort};
 
\coordinate (junc) at ($(B4.east) + (0.35, 0)$);
\draw[darkline] (B4.east) -- (junc);
\draw[darkline] (junc) -- (junc |- B5a.west);
\draw[darkline] (junc) -- (junc |- B5b.west);
\draw[darkarr]  (junc |- B5a.west) -- (B5a.west);
\draw[darkarr]  (junc |- B5b.west) -- (B5b.west);
 
\node[nd6] (B6) at (18.0, 0) {$\bullet$~Value realized\\[2pt]
$\bullet$~Winner compensates\\
\hspace{0.8em}seller via linear\\
\hspace{0.8em}payment scheme};
 
\coordinate (merge) at ($(B6.west) + (-0.4, 0)$);
\draw[darkline] (B5a.east) -| (merge);
\draw[darkline] (B5b.east) -| (merge);
\draw[darkarr]  (merge) -- (B6.west);
 
\coordinate (braceY) at (0, -2.0);
 
\draw[decorate, decoration={brace, mirror, amplitude=6pt}, line width=0.4pt]
  (B1.south west |- braceY) -- (B3.south east |- braceY)
  node[midway, below=9pt, font=\small\bfseries] {Auction Stage};
 
\draw[decorate, decoration={brace, mirror, amplitude=6pt}, line width=0.4pt]
  (B4.south west |- braceY) -- (B6.south east |- braceY)
  node[midway, below=9pt, font=\small\bfseries] {Aftermarket Collaboration};
 
\end{tikzpicture}
\caption{Linear-payment auction with aftermarket collaboration}
\label{}
\end{figure}


\section{Optimal Mechanism \label{mechanimschapter}}

\subsection{Direct linear mechanism and feasibility \label{DirrrrrMech}}

We characterize the optimal linear-payment auction under the existence of a perfect Bayesian equilibrium. By the revelation principle \citep[see][]{myerson1982,Doval2022}, every linear-payment auction with aftermarket collaboration is outcome-equivalent to a direct mechanism with linear payment schemes executed by a trustworthy mediator. We focus on direct mechanisms (``direct linear mechanisms"), represented by $(\bm{q}, \bm{c}, \bm{\pi})$, where bidders truthfully report their types. A direct linear mechanism consists of:

\begin{enumerate}
\item[(1)] \textbf{Allocation rule}: Given reported types, the mechanism determines the winning probability $q_i(\bm{\theta}) \in [0,1]$ for each bidder $i$, subject to $\sum_{i \in \mathcal{N}} q_i(\bm{\theta}) \leq 1$. The asset may not be sold if the seller's expected payoff is negative.
\item[(2)] \textbf{Payment rule}: The mechanism determines payment schemes $(c_i, t_i^l) \in [0,1] \times \mathbb{R}_{+}^2$ for each bidder $i$:
\begin{enumerate}
\item[(2.1)] The winning bidder $i$ receives a linear payment scheme $c_i = (\alpha_i, t_i^w)$, drawn according to probability measure $\kappa_i(\cdot, \cdot | \bm{\theta}) \in \Delta([0,1] \times \mathcal{T}_i^w)$.
\item[(2.2)] Each losing bidder $i$ pays a fixed amount $t_i^l$, drawn according to probability measure $\nu_i^l(\cdot|\bm{\theta}) \in \Delta(\mathcal{T}_i^l)$.
\end{enumerate}
\item[(3)] \textbf{Signal realization rule}: Given the winner's identity $i \in \mathcal{N}$ and payment profile
\begin{equation}\label{bmC}
\bm{c} = \left(c_i, \bm{t}^l\right) = \left(\alpha_i, t_i^w, \bm{t}_{-i}^l\right),    
\end{equation}
the mechanism generates a public signal realization $s \in S_i$, where $S_i$ denotes the winner's signal space. The signal is drawn from probability measure $\pi_i(\cdot|i,\bm{c},\bm{\theta}) \in \Delta(S_i)$. No additional signals are transmitted.
\end{enumerate}

A trustworthy mediator operates the mechanism, revealing the winner's identity, determining payment schemes, and generating a public signal regarding the winner through a random draw.\footnote{\cite{Doval2022} study dynamic mechanism design under limited commitment---a dynamic mechanism-selection game between an uninformed designer and a privately informed agent where the designer commits only to short-term mechanisms. They identify a class of direct mechanisms that replicate the outcomes of any equilibrium in this mechanism-selection game by restricting the signal set to the designer's posterior beliefs about the agent's type. Our model is a two-period mechanism-selection game, so in the outcome-equivalent direct mechanism, the signal space can be $S_i = \Delta(\Theta_i)$ as in \cite{Doval2022}. In \cite{myerson1982}, the outcome-equivalent direct mechanism publicly recommends decisions; in our framework, this corresponds to recommending the seller's and winner's post-auction efforts, which induces the seller's posterior beliefs about the winner's type. \cite{Bester2007} similarly introduces a mediator to resolve the seller's commitment problems.} Based on the public signal realization $s\in S_i$ and payment profile \eqref{bmC}, the seller forms posterior beliefs about the winner's type:
\begin{equation}\label{sellerposteriorbelief}
\mu = \mu(\cdot |i, \bm{c}, s) \in \Delta(\Theta_i).
\end{equation}

Subsequently, both parties choose their aftermarket equilibrium strategies:
\begin{equation*}
\left(\varepsilon^s_i, \varepsilon^b_i\right) = \left(\varepsilon^s_i(\mu), \varepsilon^b_i(\theta_i, \mu)\right) \in \Delta( E^s_i)\times \Delta( E^b_i),
\end{equation*}
as specified in \eqref{effortsderive}. The seller receives cash payments collected by the mediator and value shares from the winner (with no discounting due to absence of time preference).

We denote bidder $i$'s expected payoff from reporting $\hat{\theta}_i \in \Theta_i$ as $U_i(\hat{\theta}_i, \varepsilon^b_i, \varepsilon^s_i, \theta_i)$, given that other bidders report truthfully and both parties adopt their aftermarket equilibrium strategies.\footnote{Due to seller moral hazard, our direct mechanism incorporates a signal revelation rule. Bidder misreporting affects allocation, payment, and the seller's posterior beliefs through the signal-revealing rule.} Let $U_i(\theta_i) \triangleq U_i(\theta_i, \varepsilon^b_i, \varepsilon^s_i, \theta_i)$ denote bidder $i$'s expected payoff from truthful reporting.

\begin{definition}\label{d1}
A direct linear mechanism $(\bm{q}, \bm{c}, \bm{\pi})$ is \textbf{feasible} if it satisfies both incentive compatibility and individual rationality. Incentive compatibility requires that truthful reporting by each bidder constitutes a Bayesian Nash equilibrium, such that $U_i(\theta_i) \geq U_i(\hat{\theta}_i, \varepsilon^b_i, \varepsilon^s_i,\theta_i)$ for all $i \in \mathcal{N}$ and all $\theta_i, \hat{\theta}_i \in \Theta_i$. Individual rationality requires that each bidder obtains non-negative expected payoff from participation, i.e., $U_i(\theta_i) \geq 0$ for all $i \in \mathcal{N}$ and all $\theta_i \in \Theta_i$.
\end{definition}

We characterize the optimal linear mechanism for each scenario based on the pivotal role in aftermarket collaboration.

\begin{remark}[Linear-quadratic specification]\label{rmk:LQ}
Our model adopts a linear-quadratic specification, a commonly used framework in the contract theory and principal-agent literature \citep[see, e.g.,][]{Holm1987,Holm1991,TheoryofIncentive2002,ContractTheory2005}. As shown in \eqref{expostpayoff-seller} and \eqref{expostpayoff-winner}, under both winner- and seller-pivotal collaboration, each party's optimization problem exhibits the standard linear-quadratic form given the other party's effort. Specifically, under the winner-pivotal scenario, the seller's ex-post expected payoff $t_i^w + \sum_{j\neq i}t_j^l + \alpha_i \cdot (\theta_i + e_i^s) \cdot e_i^b - \frac{1}{2}(e_i^s)^2$ is linear in $e_i^s$ with a quadratic cost, and symmetrically for the winner. The seller-pivotal case is analogous, with $e_i^s$ and $e_i^b$ swapped in the value function. Since the value function is bilinear in both efforts under either collaboration mode, this specification extends the classical single-agent linear-quadratic framework to a double moral hazard setting. It also preserves analytical tractability and transparently reveals how the pivotal role shapes the optimal mechanism---including the allocation rule, linear payment scheme, post-auction effort choices, and the seller's payoff.
\end{remark}

\subsection{Winner-pivotal collaboration \label{WLcaseSec}}

We begin with the scenario where the winner's effort plays the pivotal role, characterized by the function
\begin{equation*}
V_i(\theta_i, e_i^s, e^b_i) = (\theta_i + e_i^s) \cdot e_i^b,
\end{equation*}
as in \eqref{Vstructure-1}. Both parties' value-sharing payoffs are concave in their effort choices. Each party's aftermarket equilibrium strategy $\varepsilon^s_i \in \Delta( E^s_i)$ and $\varepsilon^b_i \in \Delta( E^b_i)$ degenerates to a Dirac measure, $\delta(\cdot)$. We use the superscript ``$wp$" to denote winner-pivotal collaboration. For any signal realization $s \in S_i$ and posterior belief \eqref{sellerposteriorbelief}, $(\varepsilon^{wp-s}_i, \varepsilon^{wp-b}_i) = (\delta_{\{e^{wp-s}_i\}}, \delta_{\{e^{wp-b}_i\}})$ satisfies:
\begin{equation*}
\begin{split}
e^{wp-s}_i \in & \argmax_{e_i^s \in E^s_i}\left\{
\begin{gathered}
\alpha_i\cdot \mathbb{E}_{\widetilde{\theta}_i}\left[\left.\underbrace{\mathbb{E}_{\varepsilon^{wp-b}_i(\widetilde{\theta}_i)}\left[\left(\widetilde{\theta}_i + e_i^s\right)\cdot e^b_i\right]}_{= \left(\widetilde{\theta}_i + e_i^s\right)\cdot e^{wp-b}_i(\widetilde{\theta}_i)}\right|\mu\right]\\
- \frac{1}{2}\cdot \left(e_i^s\right)^2
\end{gathered}
\right\} \xlongequal{F.O.C.} \alpha_i\cdot \mathbb{E}_{\widetilde{\theta}_i}\left[\left.e^{wp-b}_i(\widetilde{\theta}_i) \right|\mu\right],\\
e^{wp-b}_i \in & \argmax_{e_i^b \in E^b_i}\left\{
\begin{gathered}
\left(1-\alpha_i\right)\cdot \underbrace{\mathbb{E}_{\varepsilon^{wp-s}_i(\mu)}\left[\left(\theta_i + e_i^s\right)\cdot e^b_i\right]}_{= \left(\theta_i + e^{wp-s}_i(\mu)\right)\cdot e^b_i} \\
- \frac{1}{2}\cdot \left(e^b_i\right)^2
\end{gathered}
\right\} \xlongequal{F.O.C.} \left(1-\alpha_i\right) \cdot \left(\theta_i + e^{wp-s}_i(\mu)\right),
\end{split}
\end{equation*}
which are derived from \eqref{effortsderive}. Solving the above system of equations yields:
\begin{equation}\label{Optefforttype2}
\begin{split}
e^{wp-s}_i & = e^{wp-s}_i(\mu) = \frac{(1-\alpha_i) \cdot \alpha_i\cdot  \mathbb{E}\left[\left.\widetilde{\theta}_i\right|\mu\right]}{1 - (1-\alpha_i) \cdot \alpha_i},\\
e^{wp-b}_i & = e^{wp-b}_i(\theta_i,\mu) = (1-\alpha_i)\cdot \left(\theta_i + \frac{(1-\alpha_i) \cdot \alpha_i\cdot \mathbb{E}\left[\left.\widetilde{\theta}_i\right|\mu\right]}{1 - (1-\alpha_i) \cdot \alpha_i}\right),
\end{split}
\end{equation}
where $\mathbb{E}[\widetilde{\theta}_i | \mu] = \mathbb{E}[\widetilde{\theta}_i | \mu(\cdot |i, \bm{c}, s)]$ denotes the posterior mean of the winner's type. Therefore, for any direct linear mechanism $(\bm{q}, \bm{c}, \bm{\pi})$, bidder $i$'s expected payoff under winner-pivotal collaboration, $U^{wp}_i(\hat{\theta}_i, e^{wp-b}_i, e^{wp-s}_i, \theta_i)$, is given by:
\begin{equation}\label{Uireporting}
\begin{split}
& U^{wp}_i(\hat{\theta}_i, e^{wp-b}_i, e^{wp-s}_i, \theta_i) = \int_{\bm{\Theta}_{-i}} q_i(\hat{\theta}_i,\bm{\theta}_{-i}) \int_{\bm{\mathcal{T}}_{-i}^l}\int_{[0,1]\times \mathcal{T}_i^w} \\
& \cdot \int_{S_i} \left\{
\begin{gathered}
(1-\alpha_i) \cdot \left(\theta_i +  \underbrace{\frac{(1-\alpha_i) \cdot \alpha_i\cdot  \mathbb{E}\left[\left.\widetilde{\theta}_i\right|\mu(\cdot |i, \bm{c}, s)\right]}{1 - (1-\alpha_i) \cdot \alpha_i}}_{e^{wp-s}_i} \right)\\
\cdot \underbrace{(1-\alpha_i)\cdot \left(\theta_i + \frac{(1-\alpha_i) \cdot \alpha_i\cdot  \mathbb{E}\left[\left.\widetilde{\theta}_i\right|\mu(\cdot |i, \bm{c}, s)\right]}{1 - (1-\alpha_i) \cdot \alpha_i}\right)}_{e^{wp-b}_i}\\
- \frac{1}{2}\cdot\left(\underbrace{(1-\alpha_i)\cdot \left(\theta_i + \frac{(1-\alpha_i) \cdot \alpha_i\cdot \mathbb{E}\left[\left.\widetilde{\theta}_i\right|\mu(\cdot |i, \bm{c}, s)\right]}{1 - (1-\alpha_i) \cdot \alpha_i}\right)}_{e^{wp-b}_i}\right)^2 - t_i^w
\end{gathered}
\right\} \pi_i(d s|i,\bm{c},\hat{\theta}_i,\bm{\theta}_{-i}) \\
&\cdot \kappa_i(d\alpha_i,dt_i^w|\hat{\theta}_i,\bm{\theta}_{-i})\bm{\nu}_{-i}^l(d\bm{t}_{-i}^l|\hat{\theta}_i,\bm{\theta}_{-i})\bm{f}_{-i}(\bm{\theta}_{-i}) d\bm{\theta}_{-i}\\
& + \int_{\bm{\Theta}_{-i}} \left(1 - q_i(\hat{\theta}_i,\bm{\theta}_{-i})\right) \int_{\mathcal{T}_i^l} \left(-t_i^l\right)\nu_i^{l}(d t_i^l|\hat{\theta}_i,\bm{\theta}_{-i})\bm{f}_{-i}(\bm{\theta}_{-i}) d\bm{\theta}_{-i},
\end{split}
\end{equation}
where $\bm{\mathcal{T}}_{-i}^l \equiv \bigtimes_{j \in \mathcal{N}, j \neq i} \mathcal{T}_j^l$ and $\bm{\nu}_{-i}^l(d\bm{t}_{-i}^l|\hat{\theta}_i,\bm{\theta}_{-i}) \equiv \prod_{j \in \mathcal{N}, j \neq i} \nu_j^l(d t_j^l|\hat{\theta}_i,\bm{\theta}_{-i})$.

The seller's objective is to maximize expected payoff. According to \eqref{expostpayoff-seller}, the seller's payoff comprises the sum of bidders' cash payments and value shares, net of effort costs. Under winner-pivotal collaboration, for any feasible direct linear mechanism (cf. Definition \ref{d1}), the seller's expected equilibrium payoff is given by (see Appendix \ref{9pf} for derivations):
\begin{equation}\label{REV-000}
\begin{split}
& \text{Rev}^{wp}(\bm{q}, \bm{c}, \bm{\pi}) = \sum_{i\in\mathcal{N}} \int_{\bm{\Theta}} q_i(\bm{\theta}) \cdot \int_{\bm{\mathcal{T}}_{-i}^l}\int_{[0,1]\times \mathcal{T}_i^w} \\
& \cdot\int_{S_i} \left\{
\begin{gathered}
\frac{1-\alpha_i^2}{2}\cdot \theta_i^2 + (1-\alpha_i^2) \cdot \theta_i \cdot \frac{(1-\alpha_i) \cdot \alpha_i\cdot  \mathbb{E}\left[\left.\widetilde{\theta}_i\right|\mu(\cdot |i, \bm{c}, s)\right]}{1 - (1-\alpha_i) \cdot \alpha_i}\\
- \frac{\alpha_i^2}{2}\cdot\left(\frac{(1-\alpha_i) \cdot \alpha_i\cdot  \mathbb{E}\left[\left.\widetilde{\theta}_i\right|\mu(\cdot |i, \bm{c}, s)\right]}{1 - (1-\alpha_i) \cdot \alpha_i}\right)^2 \\
- (1-\alpha_i)^2 \cdot \theta_i \cdot \frac{1 - F_i(\theta_i)}{f_i(\theta_i)} \\
- (1-\alpha_i)^2 \cdot \frac{(1-\alpha_i) \cdot \alpha_i\cdot  \mathbb{E}\left[\left.\widetilde{\theta}_i\right|\mu(\cdot |i, \bm{c}, s)\right]}{1 - (1-\alpha_i) \cdot \alpha_i} \cdot \frac{1 - F_i(\theta_i)}{f_i(\theta_i)}
\end{gathered}
\right\} \pi_i(d s|i,\bm{c},\bm{\theta}) \\
&\cdot \kappa_i(d\alpha_i,dt_i^w|\bm{\theta}) \bm{\nu}_{-i}^l(d\bm{t}_{-i}^l|\hat{\theta}_i,\bm{\theta}_{-i}) \bm{f}(\bm{\theta}) d \bm{\theta} - \sum_{i\in\mathcal{N}}U_i^{wp}(\underline{\theta}_i).
\end{split}
\end{equation}

The seller's expected payoff depends on the allocation rule, payment rule, signal realization rule, and bidders' lowest-type payoffs. This complexity precludes direct application of the standard Myersonian approach \citep[][]{myerson1981}.\footnote{In standard mechanisms, the seller's expected payoff can be expressed as the expectation of $q_i(\bm{\theta})$ times virtual valuation. The optimal mechanism allocates to the bidder with highest non-negative virtual valuation, which is increasing in type.} We propose the following approach.

\begin{method}\label{Method}
We characterize the optimal mechanism in two steps:
\begin{itemize}
\item Step 1: Establish an upper bound on the seller's expected payoff across all feasible direct mechanisms.
\item Step 2: Propose a mechanism achieving this upper bound.
\end{itemize}
\end{method}

For \textbf{Step 1}, we establish:

\begin{lemma}\label{uppbdd}
Under winner-pivotal collaboration, for any feasible mechanism $(\bm{q}, \bm{c}, \bm{\pi})$, the seller's expected payoff is bounded above by:
\begin{equation*}
\overline{\text{Rev}}^{wp}(\bm{q}, \bm{c}, \bm{\pi}) = \int_{\bm{\Theta}} \max\left\{\phi_1^{wp}(\theta_1), \ldots, \phi_i^{wp}(\theta_i),\ldots,\phi_n^{wp}(\theta_n)\right\} \bm{f}(\bm{\theta}) d\bm{\theta}.
\end{equation*}
Here, $\phi_i^{wp}(\theta_i)$ represents the \textit{virtual surplus} for bidder $i$ with type $\theta_i$, defined as:
\begin{equation}\label{Lmm1-Psi}
\phi_i^{wp}(\theta_i) = \max_{\alpha_i \in [0,1]} \underbrace{\left\{\frac{1 - 2\alpha_i \cdot(1 - \alpha_i)\cdot \alpha_i - \alpha_i^4}{2(1 - (1-\alpha_i) \cdot \alpha_i)^2} \cdot \theta_i^2 - \frac{(1-\alpha_i)^2}{1 - (1-\alpha_i) \cdot \alpha_i}  \cdot \frac{1-F_i(\theta_i)}{f_i(\theta_i)}\cdot \theta_i
\right\}}_{\text{defined as }\Psi_i^{wp}(\alpha_i,\theta_i)}.
\end{equation}
\end{lemma}
\begin{proof}
Please refer to Appendix \ref{lm3pf}.
\end{proof}

Each bidder's virtual surplus is endogenously determined by the maximizing value-sharing scheme:
\begin{equation}\label{alphamax111-type2}
\alpha_i^{wp-\max} = \alpha_i^{wp-\max}(\theta_i) \in \argmax_{\alpha_i\in[0,1]} \Psi_i^{wp}(\alpha_i,\theta_i),
\end{equation}
which satisfies:

\begin{lemma}\label{Lm-characterize-alphamax}
$\alpha_i^{wp-\max}(\theta_i)$ is unique, decreasing in $\theta_i$, and interior with $\alpha_i^{wp-\max}(\theta_i) \in (0,1)$.
\end{lemma}
\begin{proof}
Please refer to Appendix \ref{Lm-characterize-alphamax-PF}.
\end{proof}

Consequently, the virtual surplus defined in \eqref{Lmm1-Psi} can be written as:
\begin{equation}\label{phiphiphi-type2}
\phi_i^{wp}(\theta_i) = \max_{\alpha_i \in [0,1]} \Psi_i^{wp}(\alpha_i,\theta_i) = \Psi_i^{wp}(\alpha_i^{wp-\max}(\theta_i),\theta_i),
\end{equation}
and satisfies:
\begin{lemma}\label{increasing-of-surplus}
Under winner-pivotal collaboration, each bidder's virtual surplus is non-negative and increasing in $\theta_i$.
\end{lemma}
\begin{proof}
Please refer to Appendix \ref{PFF-increasing-of-surplus}.
\end{proof}

For \textbf{Step 2}, we present a feasible mechanism achieving the upper bound in Lemma \ref{uppbdd}.

\begin{theorem}[Winner-pivotal collaboration]\label{ThmCompl1-type2}
A feasible direct linear mechanism $(\bm{q}^{wp^*}, \bm{c}^{wp^*},$ $\bm{\pi}^{wp^*})$ is optimal if:
\begin{enumerate}
\item[(i)] The allocation rule is given by:
\begin{equation}\label{tm11-11}
q_i^{wp^*}(\bm{\theta}) = \begin{cases}
1, &\text{if } \phi^{wp}_i(\theta_i)> \max_{j\neq i}\phi^{wp}_j(\theta_j)\\
0, & \text{otherwise}
\end{cases}
\end{equation}
for all $i \in \mathcal{N}$, where $\phi_i^{wp}(\cdot)$ is each bidder's virtual surplus as defined in \eqref{phiphiphi-type2}. Ties are broken uniformly at random.

\item[(ii)] The payment rule $\bm{c}^{wp^*}$ induces a deterministic value-sharing rule for the winner:
\begin{equation}\label{oppalphatype2}
\alpha_i^{wp^*}(\theta_i) = \alpha_i^{wp-\max}(\theta_i),
\end{equation}
where $\alpha_i^{wp-\max}(\theta_i)$ is defined in \eqref{alphamax111-type2}. Moreover, the cash payments upon winning and losing are determined by $\nu_i^{wp^*-w}(\cdot | \bm{\theta})$ and $\nu_i^{wp^*-l}(\cdot | \bm{\theta})$ that satisfy:\footnote{The definition of $\nu_i^{wp^*-w}(\cdot | \bm{\theta})$ refers to \eqref{parts_in_contract}.}
\begin{equation}\label{OptmechType2-cash}
\begin{split}
& \int_{\bm{\Theta}_{-i}} \left\{
\begin{gathered}
q_i^{wp^*}(\theta_i,\bm{\theta}_{-i}) \cdot \int_{\mathcal{T}_i^w} t_i^w\nu_i^{wp^*-w}(d t_i^w|\theta_i,\bm{\theta}_{-i})  \\
+ \left(1 - q_i^{wp^*}(\theta_i,\bm{\theta}_{-i})\right)\int_{\mathcal{T}_i^l} t_i^l\nu_i^{wp^*-l}(d t_i^l|\theta_i,\bm{\theta}_{-i})
\end{gathered}
\right\} \bm{f}_{-i}(\bm{\theta}_{-i}) d\bm{\theta}_{-i} \\
= & \int_{\bm{\Theta}_{-i}} \left\{
\begin{gathered}
q_i^{wp^*}(\theta_i,\bm{\theta}_{-i}) \cdot \frac{(1-\alpha_i^{wp-\max}(\theta_i))^2}{2} \\
\cdot \left(\theta_i + \frac{(1-\alpha_i^{wp-\max}(\theta_i)) \cdot \alpha_i^{wp-\max}(\theta_i)\cdot \theta_i}{1 - (1-\alpha_i^{wp-\max}(\theta_i)) \cdot \alpha_i^{wp-\max}(\theta_i)}\right)^2 \\
-\int_{\underline{\theta}_i}^{\theta_i} q_i^{wp^*}(\tau,\bm{\theta}_{-i}) \cdot (1-\alpha_i^{wp-\max}(\tau))^2\\
\cdot \left(\tau + \frac{(1-\alpha_i^{wp-\max}(\tau)) \cdot \alpha_i^{wp-\max}(\tau)\cdot \tau}{1 - (1-\alpha_i^{wp-\max}(\tau)) \cdot \alpha_i^{wp-\max}(\tau)}\right) d\tau
\end{gathered}
\right\} \bm{f}_{-i}(\bm{\theta}_{-i}) d\bm{\theta}_{-i}.
\end{split}
\end{equation}
This expression (integration form) shows that the cash payments for both the winner and the losers can be constructed in multiple ways.

\item[(iii)] The signal realization rule $\pi_i^{wp^*}(\cdot|i,\bm{c},\bm{\theta}) \in \Delta(S_i)$, conditional on the profile of bidders' reported types, winner's identity, and bidders' payments, is a Dirac measure:
\begin{equation*}
\pi_i^{wp^*}(\cdot|i,\bm{c},\bm{\theta}) = \delta_{\{\theta_i\}},
\end{equation*}
which assigns probability mass one to the winner's truthfully reported type $\theta_i$.
\end{enumerate}
\end{theorem}
\begin{proof}
Please refer to Appendix \ref{APPEEE}.
\end{proof}

Theorem \ref{ThmCompl1-type2} establishes that the optimal direct linear mechanism selects the bidder with the highest virtual surplus rather than the highest type, where each bidder's virtual surplus is increasing in type. The asset can always be sold, since $\phi_i^{wp}(\theta_i) \geq 0$ for all $\theta_i\in \Theta_i$ by Lemma \ref{increasing-of-surplus}. This allocation rule can lead to discrimination due to bidder heterogeneity. Specifically, given the deterministic value-sharing rule, $\alpha_i^{wp^*}(\theta_i) = \alpha_i^{wp-\max}(\theta_i)$ for each winning bidder, variations in hazard rates across bidders with identical types result in different winning probabilities.\footnote{Distribution $F_i$ dominates $F_j$ in hazard rate order, denoted $F_i \succ_{hr} F_j$, if $\frac{f_i(\theta)}{1 - F_i(\theta)} < \frac{f_j(\theta)}{1 - F_j(\theta)}$ for all $\theta$ \cite[see, e.g.,][]{Moshe1994,Kirk2012}.}

Under the winner-pivotal scenario, the optimal direct linear mechanism generates linear payments for the winner with deterministic value shares that are decreasing in the bidder's type and independent of the number of bidders (see Lemma \ref{Lm-characterize-alphamax}). The winner's cash payments and losing bidders' cash payments are jointly determined by the optimal allocation rule and the value-sharing rule, which admits multiple constructions. Theorem \ref{ThmCompl1-type2} also establishes that the optimal signal realization rule takes a Dirac measure form; that is, the optimal feasible linear mechanism publicly reveals the winning bidder's truthfully reported type, achieving full revelation. Consequently, during aftermarket collaboration and given the equilibrium strategies in \eqref{Optefforttype2}, the seller and winner choose post-auction efforts under $(\bm{q}^{wp^*}, \bm{c}^{wp^*}, \bm{\pi}^{wp^*})$ as follows:
\begin{equation}\label{Optefforttype2-opt}
\begin{split}
e^{wp-s}_i(\theta_i) = & \frac{\left(1-\alpha_i^{wp^*}(\theta_i)\right) \cdot \alpha_i^{wp^*}(\theta_i)\cdot \theta_i}{1 - \left(1-\alpha_i^{wp^*}(\theta_i)\right) \cdot \alpha_i^{wp^*}(\theta_i)},\\
e^{wp-b}_i(\theta_i) = & \left(1-\alpha_i^{wp^*}(\theta_i)\right)\cdot \left(\theta_i + \frac{\left(1-\alpha_i^{wp^*}(\theta_i)\right) \cdot \alpha_i^{wp^*}(\theta_i)\cdot \theta_i}{1 - \left(1-\alpha_i^{wp^*}(\theta_i)\right) \cdot \alpha_i^{wp^*}(\theta_i)}\right)\\
= & \frac{\left(1-\alpha_i^{wp^*}(\theta_i)\right) \cdot \theta_i}{1 - \left(1-\alpha_i^{wp^*}(\theta_i)\right) \cdot \alpha_i^{wp^*}(\theta_i)},
\end{split}
\end{equation}
both of which are non-negative.

\subsection{Seller-pivotal collaboration \label{SLcaseSec}}

We examine seller-pivotal collaboration. The value creation function is $V(\theta_i, e_i^s, e_i^b) = (\theta_i + e_i^b) \cdot e_i^s$. Zero seller effort yields zero value, establishing the seller's pivotal role (Fact \ref{leadeffectDEF}).

Similar to \eqref{Optefforttype2}, each party's equilibrium strategy degenerates to a Dirac measure. We use the superscript ``$sp$" to denote seller-pivotal collaboration and derive the equilibrium strategies:
\begin{equation}\label{Optefforttype2-2}
\begin{split}
e^{sp-s}_i & = e^{sp-s}_i(\mu) = \frac{\alpha_i\cdot  \mathbb{E}\left[\left.\widetilde{\theta}_i\right|\mu(\cdot |i, \bm{c}, s)\right]}{1 - (1-\alpha_i) \cdot \alpha_i},\\
e^{sp-b}_i & = e^{sp-b}_i(\mu) = \frac{(1-\alpha_i) \cdot \alpha_i\cdot  \mathbb{E}\left[\left.\widetilde{\theta}_i\right|\mu(\cdot |i, \bm{c}, s)\right]}{1 - (1-\alpha_i) \cdot \alpha_i}.
\end{split}
\end{equation}

Unlike the winner-pivotal scenario, the winner's optimal effort no longer depends on his own type. Bidder $i$'s expected payoff $U^{sp}_i(\hat{\theta}_i, e^{sp-b}_i, e^{sp-s}_i, \theta_i)$ is:
\begin{equation}\label{Uireporting-2}
\begin{split}
& U^{sp}_i(\hat{\theta}_i, e^{sp-b}_i, e^{sp-s}_i, \theta_i)\\
= & \int_{\bm{\Theta}_{-i}} q_i(\hat{\theta}_i,\bm{\theta}_{-i}) \cdot \int_{\bm{\mathcal{T}}_{-i}^l}\int_{[0,1]\times \mathcal{T}_i^w} \\
& \cdot \int_{S_i} \left\{
\begin{gathered}
(1-\alpha_i) \cdot \left(\theta_i +  \frac{(1-\alpha_i) \cdot \alpha_i\cdot  \mathbb{E}\left[\left.\widetilde{\theta}_i\right|\mu(\cdot |i, \bm{c}, s)\right]}{1 - (1-\alpha_i) \cdot \alpha_i} \right)\\
\cdot \frac{\alpha_i\cdot  \mathbb{E}\left[\left.\widetilde{\theta}_i\right|\mu(\cdot |i, \bm{c}, s)\right]}{1 - (1-\alpha_i) \cdot \alpha_i}\\
- \frac{1}{2}\cdot\left(\frac{(1-\alpha_i) \cdot \alpha_i\cdot \mathbb{E}\left[\left.\widetilde{\theta}_i\right|\mu(\cdot |i, \bm{c}, s)\right]}{1 - (1-\alpha_i) \cdot \alpha_i}\right)^2 - t_i^w
\end{gathered}
\right\} \pi_i(d s|i,\bm{c},\hat{\theta}_i,\bm{\theta}_{-i}) \\
&\cdot \kappa_i(d\alpha_i,dt_i^w|\hat{\theta}_i,\bm{\theta}_{-i})\bm{\nu}_{-i}^l(d\bm{t}_{-i}^l|\hat{\theta}_i,\bm{\theta}_{-i}) \bm{f}_{-i}(\bm{\theta}_{-i}) d\bm{\theta}_{-i}\\
& - \int_{\bm{\Theta}_{-i}} \left(1 - q_i(\hat{\theta}_i,\bm{\theta}_{-i})\right) \int_{\mathcal{T}_i^l} t_i^l\nu_i^{l}(d t_i^l|\hat{\theta}_i,\bm{\theta}_{-i})\bm{f}_{-i}(\bm{\theta}_{-i}) d\bm{\theta}_{-i},
\end{split}
\end{equation}
which is analogous to $U^{wp}_i$ (see \eqref{Uireporting}).

Following the winner-pivotal scenario, we derive the seller's expected payoff for any feasible mechanism (see \eqref{Sllll-REV0}) and apply the same \textbf{two-step} approach from \hyperref[Method]{Methodology}.

\begin{lemma}\label{S-lead-LM}
Under seller-pivotal collaboration, for any feasible mechanism $(\bm{q}, \bm{c}, \bm{\pi})$, the seller's expected payoff is bounded above by:
\begin{equation*}
\overline{\text{Rev}}^{sp}(\bm{q}, \bm{c}, \bm{\pi}) = \int_{\bm{\Theta}} \max\left\{\phi_1^{sp}(\theta_1), \ldots, \phi_i^{sp}(\theta_i),\ldots,\phi_n^{sp}(\theta_n)\right\} \bm{f}(\bm{\theta}) d\bm{\theta}.
\end{equation*}
Here, $\phi_i^{sp}(\theta_i)$ represents the \textit{virtual surplus} for bidder $i$ with type $\theta_i$, defined as:
\begin{equation}\label{Lmm1-Psi-1}
\phi_i^{sp}(\theta_i) = \max_{\alpha_i \in [0,1]} \underbrace{\left\{\frac{2\alpha_i \cdot (1  - (1 - \alpha_i)\cdot \alpha_i) - \alpha_i^4}{2(1 - (1-\alpha_i) \cdot \alpha_i)^2}\cdot \theta_i^2 - \frac{(1-\alpha_i) \cdot\alpha_i}{1 - (1-\alpha_i) \cdot \alpha_i} \cdot \frac{1-F_i(\theta_i)}{f_i(\theta_i)}\cdot \theta_i
\right\}}_{\text{defined as }\Psi_i^{sp}(\alpha_i,\theta_i)}.
\end{equation}
\end{lemma}
\begin{proof}
Please refer to Appendix \ref{PF-S-lead-LM}.
\end{proof}

Each bidder's virtual surplus is endogenously determined by the maximizing value-sharing scheme (like the winner-pivotal case in \eqref{phiphiphi-type2}):
\begin{equation}\label{alphastarSL}
\alpha_i^{sp-\max} = \alpha_i^{sp-\max}(\theta_i) \in \argmax_{\alpha_i\in[0,1]} \Psi_i^{sp}(\alpha_i,\theta_i),
\end{equation}
which satisfies:
\begin{lemma}\label{lm-sellerlead-alphamax}
$\alpha_i^{sp-\max}(\theta_i)$ is unique, non-increasing, and satisfies $\alpha_i^{sp-\max}(\theta_i) \geq \frac{1}{2}$ for all $\theta_i \in [\underline{\theta}_i, \overline{\theta}_i]$. Moreover,
\begin{equation*}
\alpha_i^{sp-\max}(\theta_i)\begin{cases}
= 1, & \text{for } \theta_i \in [\underline{\theta}_i, \theta_i^c]\\
\text{is decreasing in }\theta_i& \text{for } \theta_i \in [\theta_i^c, \overline{\theta}_i]
\end{cases},
\end{equation*}
where the threshold type $\theta_i^c$ is uniquely determined by $\frac{1 - F_i(\theta_i^c)}{\theta_i^c \cdot f_i(\theta_i^c)} = 1$.\footnote{The threshold type $\theta_i^c \in [\underline{\theta}_i, \overline{\theta}_i]$ exists since $\underline{\theta}_i \cdot f_i(\underline{\theta}_i) < 1$.}
\end{lemma}
\begin{proof}
Please refer to Appendix \ref{lm-sellerlead-alphamax-PF}.
\end{proof}

The virtual surplus defined in \eqref{Lmm1-Psi-1} can be expressed as:
\begin{equation}\label{phiSL}
\phi_i^{sp}(\theta_i) = \max_{\alpha_i \in [0,1]} \Psi_i^{sp}(\alpha_i,\theta_i) = \Psi_i^{sp}(\alpha_i^{sp-\max}(\theta_i),\theta_i),
\end{equation}
and satisfies:
\begin{lemma}\label{increasing-of-surplus-SP}
Under seller-pivotal collaboration, each bidder's virtual surplus is non-negative and increasing in $\theta_i$.
\end{lemma}
\begin{proof}
Please refer to Appendix \ref{increasing-of-surplus-SP-pf}.
\end{proof}

We characterize the optimal mechanism using the upper bound in Lemma \ref{S-lead-LM}, following Theorem \ref{ThmCompl1-type2}.

\begin{theorem}[Seller-pivotal collaboration]\label{ThmCompl1-typeSL-seller}
A feasible direct linear mechanism $(\bm{q}^{sp^*}, \bm{c}^{sp^*},$ $\bm{\pi}^{sp^*})$ is optimal if:
\begin{enumerate}
\item[(i)] The allocation rule is given by:
\begin{equation*}
q_i^{sp^*}(\bm{\theta}) = \begin{cases}
1, &\text{if } \phi^{sp}_i(\theta_i)> \max_{j\neq i}\phi^{sp}_j(\theta_j)\\
0, & \text{otherwise}
\end{cases}
\end{equation*}
for all $i \in \mathcal{N}$, where $\phi_i^{sp}(\cdot)$ is each bidder's virtual surplus as defined in \eqref{phiSL}. Ties are broken uniformly at random.

\item[(ii)] The payment rule $\bm{c}^{sp^*}$ induces a deterministic value-sharing rule for the winner:
\begin{equation*}
\alpha_i^{sp^*}(\theta_i) = \alpha_i^{sp-\max}(\theta_i),
\end{equation*}
where $\alpha_i^{sp-\max}(\theta_i)$ is defined in \eqref{alphastarSL}. Moreover, the cash payments upon winning and losing are determined by $\nu_i^{sp^*-w}(\cdot | \bm{\theta})$ and $\nu_i^{sp^*-l}(\cdot | \bm{\theta})$ that satisfy:\footnote{The marginal probability measure $\nu_i^{sp^*-w}(\cdot | \bm{\theta})$ is also derived from $\kappa_i(\cdot, \cdot | \bm{\theta}) \in \Delta([0,1] \times \mathcal{T}_i^w)$ (see \eqref{parts_in_contract}), as in the winner-pivotal collaboration scenario.}
\begin{equation}\label{sl-lead-op-cash}
\begin{split}
& \int_{\bm{\Theta}_{-i}} \left\{
\begin{gathered}
q_i^{sp^*}(\theta_i,\bm{\theta}_{-i}) \cdot \int_{\mathcal{T}_i^w} t_i^w\nu_i^{sp^*-w}(d t_i^w|\theta_i,\bm{\theta}_{-i})  \\
+ \left(1 - q_i^{sp^*}(\theta_i,\bm{\theta}_{-i})\right)\int_{\mathcal{T}_i^l} t_i^l\nu_i^{sp^*-l}(d t_i^l|\theta_i,\bm{\theta}_{-i})
\end{gathered}
\right\} \bm{f}_{-i}(\bm{\theta}_{-i}) d\bm{\theta}_{-i} \\
= & \int_{\bm{\Theta}_{-i}} \left\{
\begin{gathered}
q_i^{sp^*}(\theta_i,\bm{\theta}_{-i}) \cdot \left(
\begin{gathered}
\frac{(1-\alpha_i^{sp-\max}(\theta_i)) \cdot \alpha_i^{sp-\max}(\theta_i)}{1 - (1-\alpha_i^{sp-\max}(\theta_i)) \cdot \alpha_i^{sp-\max}(\theta_i)} \\
+ \frac{1}{2}\cdot \biggl(\frac{(1-\alpha_i^{sp-\max}(\theta_i)) \cdot \alpha_i^{sp-\max}(\theta_i)\cdot \theta_i}{1 - (1-\alpha_i^{sp-\max}(\theta_i)) \cdot \alpha_i^{sp-\max}(\theta_i)}\biggr)^{\!2}
\end{gathered}
\right)\cdot \theta_i^2\\
- \int_{\underline{\theta}_i}^{\theta_i} q_i^{sp^*}(\tau,\bm{\theta}_{-i}) \cdot \frac{(1-\alpha_i^{sp-\max}(\tau)) \cdot\alpha_i^{sp-\max}(\tau)}{1 - (1-\alpha_i^{sp-\max}(\tau)) \cdot \alpha_i^{sp-\max}(\tau)} \cdot \tau d\tau
\end{gathered}
\right\} \bm{f}_{-i}(\bm{\theta}_{-i}) d\bm{\theta}_{-i}.
\end{split}
\end{equation}
This expression parallels \eqref{OptmechType2-cash}, where the cash payments can be constructed in multiple ways.

\item[(iii)] The signal realization rule $\pi_i^{sp^*}(\cdot|i,\bm{c},\bm{\theta}) \in \Delta(S_i)$ is a Dirac measure:
\begin{equation*}
\pi_i^{sp^*}(\cdot|i,\bm{c},\bm{\theta}) = \delta_{\{\theta_i\}},
\end{equation*}
assigning probability mass one to the winner's truthfully reported type $\theta_i$.
\end{enumerate}
\end{theorem}
\begin{proof}
Please refer to Appendix \ref{THMoptSL}.
\end{proof}

The optimal mechanism $(\bm{q}^{sp^*}, \bm{c}^{sp^*}, \bm{\pi}^{sp^*})$ exhibits similar properties to the winner-pivotal case (Theorem \ref{ThmCompl1-type2}). The asset can always be sold since $\phi_i^{sp}(\theta_i) \geq 0$ for all $\theta_i\in \Theta_i$ (Lemma \ref{lm-sellerlead-alphamax}). The mechanism selects the bidder with highest virtual surplus and generates deterministic value-sharing rules and  multi-constructed cash payments in both scenarios. The optimal mechanism achieves full revelation by disclosing the winner's truthfully reported type. Under $(\bm{q}^{sp^*}, \bm{c}^{sp^*}, \bm{\pi}^{sp^*})$, the seller and winner choose post-auction efforts (from \eqref{Optefforttype2-2}):
\begin{equation}\label{Optefforttype2-2-2}
\begin{split}
e^{sp-s}_i(\theta_i) = & \frac{\alpha_i^{sp^*}(\theta_i) \cdot \theta_i}{1 - \left(1-\alpha_i^{sp^*}(\theta_i)\right) \cdot \alpha_i^{sp^*}(\theta_i)},\\
e^{sp-b}_i(\theta_i) = & \frac{\left(1-\alpha_i^{sp^*}(\theta_i)\right) \cdot \alpha_i^{sp^*}(\theta_i) \cdot \theta_i}{1 - \left(1-\alpha_i^{sp^*}(\theta_i)\right) \cdot \alpha_i^{sp^*}(\theta_i)},
\end{split}
\end{equation}
both of which are non-negative.

\subsection{Impact of pivotal role divergence \label{Seccompare}}

This subsection compares optimal linear mechanisms under winner-pivotal and seller-pivotal collaboration. We examine how the pivotal role affects value sharing, effort exertion, and the seller's revenue.

\subsubsection{Pivotal role on value sharing\label{FULLsec}}

Does the pivotal party benefit from value sharing under the optimal mechanism? Under seller-pivotal collaboration, the optimal value-sharing rule satisfies $\alpha_i^{sp^*}(\theta_i) \geq \frac{1}{2}$ by Lemma \ref{lm-sellerlead-alphamax}. The seller consistently secures a majority share of the eventual value when she drives the aftermarket value creation. This structural feature aligns with the inherent power dynamics in seller-pivotal cooperation.

In contrast, under winner-pivotal collaboration, the pivotal role does not guarantee the winner a majority share. Specifically, $1-\alpha_i^{wp^*}(\theta_i) \leq \frac{1}{2}$ may occur for winners with relatively low types. Figure \ref{Fig_alphaCom} illustrates this asymmetry across four different distributions.

\begin{figure}[htp!]
\centering
\includegraphics[width=0.75\textwidth]{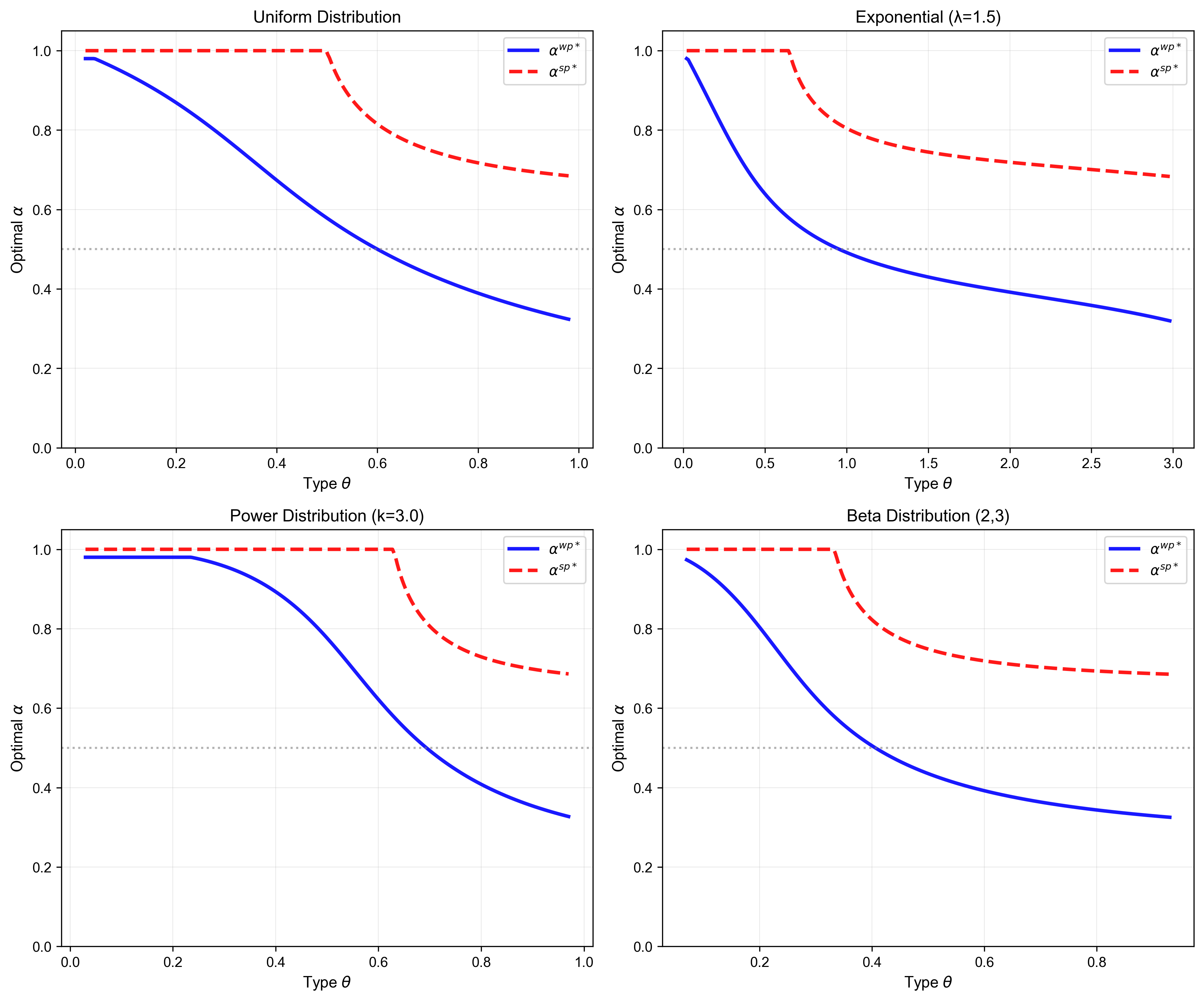}
\caption{Optimal value-sharing comparison across different distributions}
\label{Fig_alphaCom}
\end{figure}

More precisely, we establish that:

\begin{proposition}\label{CompareAlpha}
Fix the same distribution, $F_i(\cdot)$, we have $\alpha_i^{sp^*}(\theta_i) \geq \alpha_i^{wp^*}(\theta_i)$ for all $\theta_i \in [\underline{\theta}_i, \overline{\theta}_i]$.
\end{proposition}
\begin{proof}
Please refer to Appendix \ref{Pf-CompareAlpha}.
\end{proof}

Figure \ref{Fig_alphaCom} reveals a pooling region where $\alpha_i^{sp^*}(\theta_i) = 1$ for $\theta_i \in [\underline{\theta}_i, \theta_i^c]$ under seller-pivotal collaboration, while $0 < \alpha_i^{wp^*}(\theta_i) < 1$ for all types under winner-pivotal collaboration, which is consistent with the optimal deterministic value sharing derived in Theorem \ref{ThmCompl1-type2} and \ref{ThmCompl1-typeSL-seller}. This asymmetry reflects the different incentive structures. Under winner-pivotal collaboration, the winner's effort is essential for value creation. Positive value shares must be allocated to all types; otherwise, eventual value would be zero (Fact \ref{leadeffectDEF}). Under seller-pivotal collaboration, the seller's effort alone sustains positive value. For sufficiently low types, the seller extracts the entire value, leaving the winner indifferent between participating and not. Following the standard convention that ties are broken in favor of the principal, we focus on equilibria where bidders participate when indifferent.

Example \ref{eG3} illustrates this asymmetry in a spectrum licensing context. In rural markets, the regulator (seller) is pivotal: the optimal mechanism assigns her a majority share ($\alpha_i^{sp^*} \geq \frac{1}{2}$), and for operators with limited rural capacity ($\theta_i \leq \theta_i^c$), the regulator extracts the entire realized value ($\alpha_i^{sp^*} = 1$). This is feasible because the regulator's own infrastructure investment sustains service viability regardless of operator contribution---such operators add little marginal value and can be left with zero share without destroying it. In urban markets, where the operator's deployment is pivotal, withdrawing the operator's revenue share entirely would collapse service quality to zero, so every winning operator must retain a positive share.

\subsubsection{Pivotal role on effort exertion}

We now examine how the pivotal role affects post-auction effort exertion. Based on the optimal effort choices given by \eqref{Optefforttype2-opt} (winner-pivotal) and \eqref{Optefforttype2-2-2} (seller-pivotal), we establish that:

\begin{proposition}\label{Prop-Effort}
For any distribution $F_i(\cdot)$ and winner's type $\theta_i$, we have:
\begin{enumerate}
\item[(1)] Under winner-pivotal collaboration, the winner exerts more post-auction effort than the seller at the optimum, i.e., $e^{wp-b}_i(\theta_i) > e^{wp-s}_i(\theta_i)$; under seller-pivotal collaboration, the seller exerts more post-auction effort than the winner at the optimum, i.e., $e^{sp-s}_i(\theta_i) > e^{sp-b}_i(\theta_i)$.
\item[(2)] The winner's post-auction effort is higher under winner-pivotal collaboration than under seller-pivotal collaboration, i.e., $e^{wp-b}_i(\theta_i) > e^{sp-b}_i(\theta_i)$; the seller's post-auction effort is higher under seller-pivotal collaboration than under winner-pivotal collaboration, i.e., $e^{sp-s}_i(\theta_i) > e^{wp-s}_i(\theta_i)$.
\end{enumerate}
\end{proposition}
\begin{proof}
Please refer to Appendix \ref{pf-Prop-Effort}.
\end{proof}

The results can be summarized in the following Figure \ref{Figfig-effort}. Under the optimal linear mechanism, the pivotal party exerts higher effort than the non-pivotal party (row comparison), and each party exerts higher effort when pivotal than when non-pivotal (column comparison). In the figure, each shaded area dominates the corresponding unshaded areas.

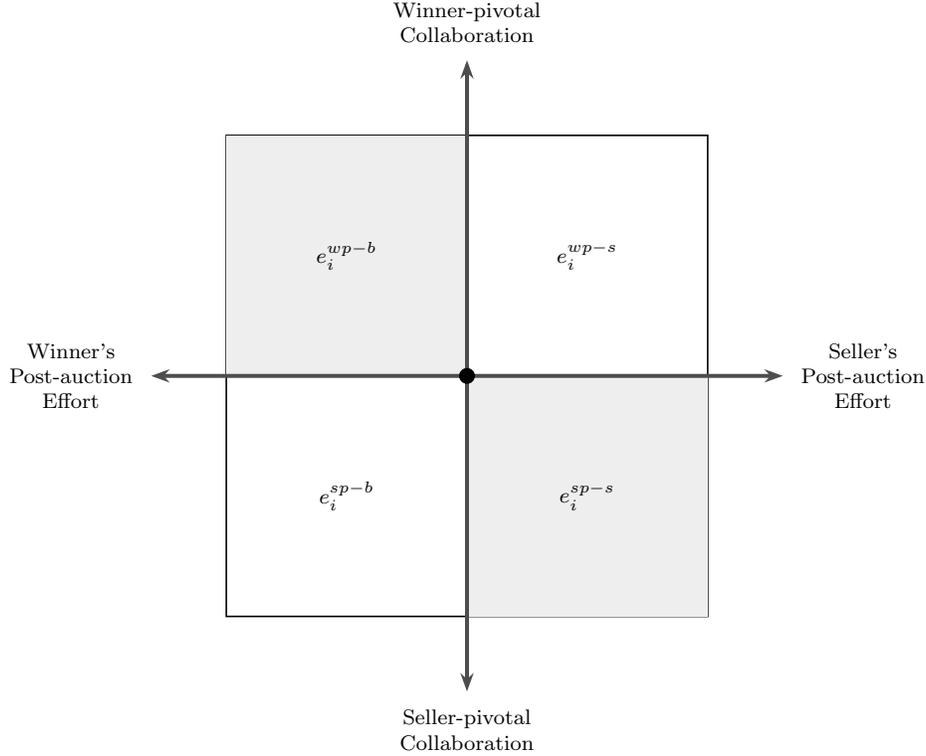
\begin{figure}[htp!]
\centering
\begin{tikzpicture}[font=\scriptsize]
 
\definecolor{shadegray}{RGB}{220,220,220}
 
\def\S{3.2}
\draw[black, line width=0.6pt] (-\S, -\S) rectangle (\S, \S);
 
\fill[shadegray!50] (-\S, 0) rectangle (0, \S);
\fill[shadegray!50] (0, -\S) rectangle (\S, 0);
 
\def\A{4.2}
\draw[-{Stealth[length=7pt,width=5pt]}, line width=1.5pt, black!70]
  (0,0) -- (\A, 0);
\draw[-{Stealth[length=7pt,width=5pt]}, line width=1.5pt, black!70]
  (0,0) -- (-\A, 0);
\draw[-{Stealth[length=7pt,width=5pt]}, line width=1.5pt, black!70]
  (0,0) -- (0, \A);
\draw[-{Stealth[length=7pt,width=5pt]}, line width=1.5pt, black!70]
  (0,0) -- (0, -\A);
 
\fill[black] (0,0) circle (3pt);
 
\node[align=center, above] at (0, \A+0.1) {Winner-pivotal\\Collaboration};
\node[align=center, below] at (0, -\A-0.1) {Seller-pivotal\\Collaboration};
\node[align=center, left]  at (-\A-0.1, 0) {Winner's\\Post-auction\\Effort};
\node[align=center, right] at (\A+0.1, 0)  {Seller's\\Post-auction\\Effort};
 
\node at (-1.6, 1.6) {$\displaystyle e_i^{wp-b}$};
\node at (1.6, 1.6)  {$\displaystyle e_i^{wp-s}$};
\node at (-1.6, -1.6) {$\displaystyle e_i^{sp-b}$};
\node at (1.6, -1.6)  {$\displaystyle e_i^{sp-s}$};
 
\end{tikzpicture}
\caption{Effort combinations under different collaboration structures}
\label{Figfig-effort}
\end{figure}

\subsubsection{Pivotal role on seller's revenue}

The previous two subsections establish that the seller secures a larger share of the eventual value under seller-pivotal collaboration (Proposition \ref{CompareAlpha}), while bearing higher post-auction effort costs. These opposing effects raise a natural question: does the pivotal party's role in aftermarket collaboration affect the seller's optimal revenue? We establish that:

\begin{proposition}\label{ppproprev}
For any distribution $F_i(\cdot)$, seller-pivotal collaboration yields higher revenue than winner-pivotal collaboration, that is:
\begin{equation*}
\text{Rev}^{sp}(\bm{q}^{sp^*}, \bm{c}^{sp^*}, \bm{\pi}^{sp^*}) > \text{Rev}^{wp}(\bm{q}^{wp^*}, \bm{c}^{wp^*}, \bm{\pi}^{wp^*}).
\end{equation*}
\end{proposition}
\begin{proof}
Please refer to Appendix \ref{ppproprev-pf}.
\end{proof}

Proposition \ref{ppproprev} shows that seller-pivotal collaboration yields the seller a higher expected payoff. This follows because each bidder's virtual surplus under seller-pivotal collaboration pointwise dominates that under winner-pivotal collaboration. The virtual surplus decomposes as $\phi_i(\theta_i) = A(\alpha) \cdot \theta_i^2 - B(\alpha) \cdot \frac{1-F_i(\theta_i)}{f_i(\theta_i)}\cdot \theta_i$, where the first term represents the social surplus and the second captures information rents. The parameters $A(\alpha)$ and $B(\alpha)$ differ across different pivotal role scenarios. The pointwise dominance of $\phi_i^{sp}(\cdot)$ arises from a higher social surplus and lower information rents under seller-pivotal collaboration.

Example \ref{eG3} illustrates the revenue channel. In rural spectrum auctions (seller-pivotal), the regulator's infrastructure investment generates substantial value even for low-capability operators, compressing the value advantage that high-type operators would otherwise command. With less to gain from misrepresenting type, operators leave smaller information rents, and the regulator captures greater expected auction revenue. In urban auctions (winner-pivotal), value creation depends more directly on operator type, amplifying information rents and reducing the regulator's expected revenue---despite urban spectrum typically commanding higher per-unit commercial value.



\section{Implementation \label{Secimp}}

We implement the optimal direct linear mechanism under both scenarios through ascending auctions.

\paragraph{Winner-pivotal collaboration.} For winner-pivotal collaboration, we employ an ascending auction with payments $\widetilde{\bm{c}}^{wp} = (\widetilde{\alpha}_i^{wp}, \widetilde{t}_i^{wp-w}; \widetilde{t}_i^{wp-l})$ for the winner and losing bidders, respectively, to implement the optimal direct linear mechanism in Theorem \ref{ThmCompl1-type2}.

The ascending clock displays virtual surplus values (defined in \eqref{phiphiphi-type2}), increasing continuously from $\underline{\phi}^{wp}$ until the last active bidder exits. For simplicity, we assume $\phi_i^{wp}(\underline{\theta}_i) \triangleq \underline{\phi}^{wp}$ for all $i \in \mathcal{N}$. Each bidder can freely drop out by releasing the button; quitting the auction is irrevocable. The last remaining active bidder wins and can continue to increase his bid to exit the auction at a higher price $p_n$, even after the penultimate bidder has exited at $p_{n-1}$.

The payment scheme $\widetilde{\bm{c}}^{wp}$ determines winning bidder $i$'s value shares endogenously from the final price $p_n$:
\begin{equation}\label{oryoryroy}
\widetilde{\alpha}_i^{wp}(p_n) \triangleq \alpha_i^{wp-\max}(\mathcal{G}(p_n)),
\end{equation}
where $\alpha_i^{wp-\max}(\cdot)$ is defined in \eqref{alphamax111-type2}, and
\begin{equation}\label{mathcalG}
\mathcal{G}(\cdot) \triangleq (\phi_i^{wp})^{-1}(\cdot)
\end{equation}
denotes the inverse of the bidder's virtual surplus function in \eqref{phiphiphi-type2}. Winning bidder $i$'s cash payment depends on $(p_{n-1}, p_n)$:
\begin{equation}\label{tttttt}
\begin{split}
\widetilde{t}_i^{wp-w}(p_n, p_{n-1}) = & \frac{\left(1-\widetilde{\alpha}_i^{wp}(p_n)\right)^2}{2}\cdot \left(\mathcal{G}(p_n) + \frac{(1-\widetilde{\alpha}_i^{wp}(p_n)) \cdot \widetilde{\alpha}_i^{wp}(p_n)\cdot \mathcal{G}(p_n)}{1 - (1-\widetilde{\alpha}_i^{wp}(p_n)) \cdot \widetilde{\alpha}_i^{wp}(p_n)}\right)^2\\
& - \int_{\mathcal{G}(p_{n-1})}^{\mathcal{G}(p_n)} \left(1- \alpha_i^{wp-\max}(\tau)\right)^2 \cdot \left(\tau + \frac{(1-\alpha_i^{wp-\max}(\tau)) \cdot \alpha_i^{wp-\max}(\tau)\cdot \tau}{1 - (1-\alpha_i^{wp-\max}(\tau)) \cdot \alpha_i^{wp-\max}(\tau)}\right) d\tau,
\end{split}
\end{equation}
while losing bidders pay nothing, i.e., $\widetilde{t}_i^{wp-l} = 0$.

Each bidder chooses his exit price based on the number $k$ of exited bidders and their exit prices $p_1 \leq \cdots \leq p_k$. Assuming each bidder follows strategy $\widetilde{\beta}_i^{wp}(\theta_i | p_1, \dots, p_k)$, we characterize the following equilibrium.

\begin{proposition}\label{ppproooo222}
The ascending auction with payments $\widetilde{\bm{c}}^{wp}$ admits a quasi-dominant strategy equilibrium, denoted $(\widetilde{\bm{\beta}}^{wp}, \widetilde{\mu}^{wp}, \widetilde{\bm{e}}^{wp})$, where:
\begin{enumerate}
\item[(i)] Each bidder's drop-out strategy is given by:
\begin{equation}\label{dfdffd11111}
\widetilde{\beta}_i^{wp}(\theta_i|p_1, \cdots, p_k) = \phi_i^{wp}(\theta_i),
\end{equation}
for each $i\in\mathcal{N}$.

\item[(ii)] The seller holds a Dirac-form posterior belief over bidder $i$'s type upon his winning:
\begin{equation}\label{xcxc}
\widetilde{\mu}^{wp}(\cdot|i,p_n,p_{n-1}) = \delta_{\left\{\mathcal{G}(p_n)\right\}},
\end{equation}
for $p_{n-1} \leq p_n \leq \phi_i^{wp}(\overline{\theta}_i)$, and $\widetilde{\mu}^{wp}(\cdot | i, p_n, p_{n-1}) = \overline{\theta}_i$ for $p_n > \phi_i^{wp}(\overline{\theta}_i)$. In case of a tie, the seller holds a Dirac-form belief over the randomly selected winner's type per \eqref{xcxc}.

\item[(iii)] The post-auction effort choices of the seller and winning bidder $i$ are given by:
\begin{equation}\label{opopop1111}
\widetilde{\bm{e}}^{wp}(p_n,p_{n-1}) =
\begin{cases}
\widetilde{e}_i^{wp-s}(p_n,p_{n-1}) = \frac{(1-\widetilde{\alpha}_i^{wp}(p_n)) \cdot \widetilde{\alpha}_i^{wp}(p_n) \cdot \mathcal{G}(p_n)}{1 - (1-\widetilde{\alpha}_i^{wp}(p_n)) \cdot \widetilde{\alpha}_i^{wp}(p_n)},\\
\widetilde{e}_i^{wp-b}(p_n,p_{n-1}) = (1-\widetilde{\alpha}_i^{wp}(p_n))\cdot \left(\theta_i + \frac{(1-\widetilde{\alpha}_i^{wp}(p_n)) \cdot \widetilde{\alpha}_i^{wp}(p_n)\cdot \mathcal{G}(p_n)}{1 - (1-\widetilde{\alpha}_i^{wp}(p_n)) \cdot \widetilde{\alpha}_i^{wp}(p_n)}\right).
\end{cases}
\end{equation}
\end{enumerate}
This ascending auction implements the optimal direct linear mechanism under winner-pivotal collaboration.
\end{proposition}
\begin{proof}
Please refer to Appendix \ref{winnerlead-IM}.
\end{proof}

We call this equilibrium ``quasi-dominant strategy'' because, given payment rule $\widetilde{\bm{c}}^{wp}$ and the seller's rational posterior beliefs, each bidder's payoff is maximized by the drop-out strategy in \eqref{dfdffd11111}, regardless of other bidders' strategies. The winner's fixed payment \eqref{tttttt} increases with his type, while the value-sharing payment \eqref{oryoryroy} decreases with his type.

\paragraph{Seller-pivotal collaboration.} For seller-pivotal collaboration, zero value-shares for low-type winners (Section \ref{FULLsec}) allow the seller to extract the entire surplus through full value-sharing below a threshold. Consequently, perfect implementation of the optimal mechanism cannot be achieved.

We employ an ascending auction with payments $\widetilde{\bm{c}}^{sp} = (\widetilde{\alpha}_i^{sp}, \widetilde{t}_i^{sp-w}; \widetilde{t}_i^{sp-l})$ for winning and losing bidders, respectively. The auction rules mirror those for winner-pivotal collaboration, except for the winner's value shares and cash payments. Conditional on the final winner's and penultimate bidder's drop-out prices $(p_n, p_{n-1})$, we define:
\begin{equation*}
\widetilde{\alpha}_i^{sp}(p_n) \triangleq \alpha_i^{sp-\max}(\mathcal{H}(p_n)) - \epsilon,
\end{equation*}
for sufficiently small $\epsilon > 0$, where $\alpha_i^{sp-\max}(\cdot)$ is defined in \eqref{alphastarSL}. Analogous to \eqref{mathcalG}, we denote the inverse function of $\phi_i^{sp}(\cdot)$ (defined in \eqref{phiSL}) as:
\begin{equation*}
\mathcal{H}(\cdot) \triangleq (\phi_i^{sp})^{-1}(\cdot).
\end{equation*}

The cash payments for winning and losing bidders are given by:
\begin{equation*}
\begin{split}
\widetilde{t}_i^{sp-w}(p_n, p_{n-1}) = & \frac{(1-\widetilde{\alpha}_i^{sp}(p_n)) \cdot \widetilde{\alpha}_i^{sp}(p_n)\cdot \mathcal{H}(p_n)}{1 - (1-\widetilde{\alpha}_i^{sp}(p_n)) \cdot \widetilde{\alpha}_i^{sp}(p_n)} \cdot \left(\mathcal{H}(p_n) + \frac{1}{2}\cdot\frac{(1-\widetilde{\alpha}_i^{sp}(p_n)) \cdot \widetilde{\alpha}_i^{sp}(p_n)\cdot \mathcal{H}(p_n)}{1 - (1-\widetilde{\alpha}_i^{sp}(p_n)) \cdot \widetilde{\alpha}_i^{sp}(p_n)}\right)\\
& - \int_{\mathcal{H}(p_{n-1})}^{\mathcal{H}(p_n)} \frac{(1-\alpha_i^{sp-\max}(\tau) + \epsilon) \cdot (\alpha_i^{sp-\max}(\tau) - \epsilon)}{1 - (1-\alpha_i^{sp-\max}(\tau) + \epsilon) \cdot (\alpha_i^{sp-\max}(\tau) - \epsilon)}\cdot \tau d\tau,
\end{split}
\end{equation*}
while losing bidders pay nothing, i.e., $\widetilde{t}_i^{sp-l} = 0$. This ascending auction approximately implements the optimal mechanism under seller-pivotal collaboration.

\begin{proposition}\label{sllead-IM}
The ascending auction with payments $\widetilde{\bm{c}}^{sp}$ admits a quasi-dominant strategy equilibrium, denoted $(\widetilde{\bm{\beta}}^{sp}, \widetilde{\mu}^{sp}, \widetilde{\bm{e}}^{sp})$, where:
\begin{enumerate}
\item[(i)] Each bidder's drop-out strategy is given by:
\begin{equation*}
\widetilde{\beta}_i^{sp}(\theta_i|p_1, \cdots, p_k) = \phi_i^{sp}(\theta_i),
\end{equation*}
for each $i\in\mathcal{N}$.

\item[(ii)] The seller holds a Dirac-form posterior belief over bidder $i$'s type upon his winning:
\begin{equation}\label{xcxc1}
\widetilde{\mu}^{sp}(\cdot|i,p_n,p_{n-1}) = \delta_{\left\{\mathcal{H}(p_n)\right\}},
\end{equation}
for $p_{n-1} \leq p_n \leq \phi_i^{sp}(\overline{\theta}_i)$, and $\widetilde{\mu}^{sp}(\cdot | i, p_n, p_{n-1}) = \overline{\theta}_i$ for $p_n > \phi_i^{sp}(\overline{\theta}_i)$. In case of a tie, the seller holds a Dirac-form belief over the randomly selected winner's type per \eqref{xcxc1}.

\item[(iii)] The post-auction effort choices of the seller and winning bidder $i$ are given by:
\begin{equation*}
\widetilde{\bm{e}}^{sp}(p_n,p_{n-1}) =
\begin{cases}
\widetilde{e}_i^{sp-s}(p_n,p_{n-1}) = \frac{\widetilde{\alpha}_i^{sp}(p_n) \cdot \mathcal{H}(p_n)}{1 - (1-\widetilde{\alpha}_i^{sp}(p_n)) \cdot \widetilde{\alpha}_i^{sp}(p_n)},\\
\widetilde{e}_i^{sp-b}(p_n,p_{n-1}) = \frac{(1-\widetilde{\alpha}_i^{sp}(p_n)) \cdot \widetilde{\alpha}_i^{sp}(p_n)\cdot \mathcal{H}(p_n)}{1 - (1-\widetilde{\alpha}_i^{sp}(p_n)) \cdot \widetilde{\alpha}_i^{sp}(p_n)}.
\end{cases}
\end{equation*}
\end{enumerate}
This ascending auction approximately implements the optimal direct linear mechanism under seller-pivotal collaboration.
\end{proposition}
\begin{proof}
Please refer to Appendix \ref{selllead-IM}.
\end{proof}

The approximation arises in the winner's value shares: $\widetilde{\alpha}_i^{sp} \rightarrow \alpha_i^{sp^*}$ as $\epsilon\rightarrow 0$. The seller must leave a small value share $\epsilon > 0$ for low-type winners; otherwise, the highest-type bidder may not be selected, preventing optimal mechanism implementation.

\section{Extensions \label{Diss}}

\subsection{Aftermarket collaboration with effort substitution \label{Diss01}}

In the preceding analysis, regardless of whether winner-pivotal or seller-pivotal aftermarket collaboration is employed, both parties' post-auction efforts exhibit a complementary form, i.e., $e_i^b \cdot e_i^s$ appearing in multiplicative form. However, seller and winner efforts may exhibit \textit{substitutability}, where one party's input reduces the other's need. In \textit{Infrastructure concessions}, project quality relies on either government oversight or contractor controls. In \textit{Franchising}, national advertising may replace local promotion. If one party exerts high effort, the other may reduce their own investment. 

We consider the extreme case where aftermarket efforts are perfect substitutes:
\begin{equation}\label{Vstructure-Esub}
V_i(\theta_i,e_i^s, e_i^b) = \theta_i\cdot(e_i^s + e_i^b).
\end{equation}
Following the framework in Sections \ref{DirrrrrMech}-\ref{SLcaseSec}, we apply the same \hyperref[Method]{Methodology} to derive the optimal mechanism. Denoting the direct linear mechanism as $(\hat{\bm{q}}, \hat{\bm{c}}, \hat{\bm{\pi}})$, we characterize the seller's payoff upper bound analogously to Lemmas \ref{uppbdd} and \ref{S-lead-LM}:

\begin{lemma}
For effort substitution (ES) with value creation \eqref{Vstructure-Esub}, the seller's expected equilibrium payoff in any feasible direct linear mechanism $(\hat{\bm{q}}, \hat{\bm{c}}, \hat{\bm{\pi}})$ is bounded above by $\overline{\text{Rev}}^{ES}$:
\begin{equation*}
\overline{\text{Rev}}^{ES}(\hat{\bm{q}}, \hat{\bm{c}}, \hat{\bm{\pi}}) = \int_{\bm{\Theta}} \max\left\{\phi_1^{ES}(\theta_1), \ldots, \phi_i^{ES}(\theta_i),\ldots,\phi_n^{ES}(\theta_n)\right\} \bm{f}(\bm{\theta}) d\bm{\theta}.
\end{equation*}
Here, $\phi_i^{ES}(\theta_i)$ represents the \textit{virtual surplus} for bidder $i$ with type $\theta_i$, defined as:
\begin{equation}\label{Virtul-Esub}
\phi_i^{ES}(\theta_i) = \begin{cases}
\frac{3}{4}\cdot \theta_i^2 + \frac{1}{4}\left(\frac{1-F_i(\theta_i)}{f_i(\theta_i)}\right)^2 - \frac{1}{2}\cdot\frac{1-F_i(\theta_i)}{f_i(\theta_i)}\cdot \theta_i, & \text{for } \hat{\theta}_i^c\leq \theta_i \leq \overline{\theta}_i\\
\frac{\theta_i^2}{2}, & \text{for } \underline{\theta}_i \leq \theta_i \leq \hat{\theta}_i^c
\end{cases},
\end{equation}
where the threshold $\hat{\theta}_i^c$ is defined as:
\begin{equation}\label{Lm-SUB-thres}
\hat{\theta}_i^c \triangleq \inf\left\{\theta_i\in[\underline{\theta}_i,\overline{\theta}_i] : \frac{1-F_i(\theta_i)}{f_i(\theta_i)}\cdot \frac{1}{\theta_i} \leq 1\right\}.
\end{equation}
\end{lemma}

The derivation follows Appendices \ref{lm3pf} and \ref{PF-S-lead-LM} (omitted; available upon request). Applying Step 2 of the \hyperref[Method]{Methodology}, we construct a feasible mechanism achieving this bound:

\begin{proposition}[Effort substitution]
A feasible direct linear mechanism $(\hat{\bm{q}}^*, \hat{\bm{c}}^*, \hat{\bm{\pi}}^*)$ is optimal if:
\begin{enumerate}
\item[(i)] The allocation rule is given by:
\begin{equation*}
\hat{q}_i^*(\bm{\theta}) = \begin{cases}
1, &\text{if } \phi^{ES}_i(\theta_i)> \max_{j\neq i}\phi^{ES}_j(\theta_j)\\
0, & \text{otherwise}
\end{cases}
\end{equation*}
for all $i \in \mathcal{N}$, where $\phi_i^{ES}(\cdot)$ is each bidder's virtual surplus as defined in \eqref{Virtul-Esub}. Ties are broken uniformly at random.

\item[(ii)] The payment rule $\hat{\bm{c}}^*$ induces a deterministic value-sharing rule for the winner:
\begin{equation}\label{Sub-alpha}
\hat{\alpha}_i^*(\theta_i) = \begin{cases}
\frac{1}{2}\cdot\left(1+ \frac{1-F_i(\theta_i)}{f_i(\theta_i)}\cdot \frac{1}{\theta_i}\right), & \text{for } \hat{\theta}_i^c\leq \theta_i \leq \overline{\theta}_i\\
1, & \text{for }\theta_i < \hat{\theta}_i^c
\end{cases}
\end{equation}
for the threshold $\theta_i^c$ in \eqref{Lm-SUB-thres}. Cash payments are determined by measures $\hat{\nu}_i^{w^*}(\cdot | \bm{\theta})$ and $\hat{\nu}_i^{l^*}(\cdot | \bm{\theta})$, analogous to \eqref{OptmechType2-cash} and \eqref{sl-lead-op-cash}, satisfying:
\begin{equation*}
\begin{split}
& \int_{\bm{\Theta}_{-i}} \left\{
\begin{gathered}
\hat{q}_i^*(\theta_i,\bm{\theta}_{-i}) \cdot \int_{\mathcal{T}_i^w} t_i^w\hat{\nu}_i^{w^*}(d t_i^w|\theta_i,\bm{\theta}_{-i})  \\
+ \left(1 - \hat{q}_i^*(\theta_i,\bm{\theta}_{-i})\right)\int_{\mathcal{T}_i^l} t_i^l\hat{\nu}_i^{l^*}(d t_i^l|\theta_i,\bm{\theta}_{-i})
\end{gathered}
\right\} \bm{f}_{-i}(\bm{\theta}_{-i}) d\bm{\theta}_{-i} \\
= & \int_{\bm{\Theta}_{-i}} \left\{
\begin{gathered}
\frac{1- \left(\hat{\alpha}_i^*(\theta_i)\right)^2}{2}\cdot \theta_i^2\cdot \hat{q}_i^*(\theta_i,\bm{\theta}_{-i})\\
- \int_{\underline{\theta}_{i}}^{\theta_i} \left(1-\hat{\alpha}_i^*(\tau)\right)\cdot \tau \cdot \hat{q}_i^*(\tau,\bm{\theta}_{-i})d\tau
\end{gathered}
\right\}\bm{f}_{-i}(\bm{\theta}_{-i}) d\bm{\theta}_{-i}.
\end{split}
\end{equation*}

\item[(iii)] The signal realization rule $\hat{\pi}_i^*(\cdot|i,\hat{\bm{c}},\bm{\theta}) \in \Delta(S_i)$, conditional on the profile of bidders' reported types, winner's identity, and bidders' payments, is a Dirac measure:
\begin{equation*}
\hat{\pi}_i^*(\cdot|i,\hat{\bm{c}},\bm{\theta}) = \delta_{\{\theta_i\}},
\end{equation*}
which assigns probability mass one to the winner's truthfully reported type $\theta_i$.
\end{enumerate}
\end{proposition}

The proof follows Appendices \ref{APPEEE} and \ref{THMoptSL} (omitted; available upon request). The optimal mechanism shares key structural features with effort-complementarity cases: virtual surplus allocation, deterministic value-sharing, and full revelation. A notable difference is that value-sharing is explicitly given by \eqref{Sub-alpha}, whereas it is implicitly determined by maximization problems \eqref{alphamax111-type2} and \eqref{alphastarSL} in winner-pivotal and seller-pivotal cases.

Similar to seller-pivotal collaboration, optimal value-sharing decreases with the winner's type, and full surplus extraction occurs for types below threshold $\hat{\theta}_i^c$ (equation \eqref{Sub-alpha}). The reason mirrors the seller-pivotal case: when $e_i^b=0$, neither $V_i(\theta_i,e_i^s, e_i^b) = (\theta_i + e_i^b) \cdot e_i^s$ nor $V_i(\theta_i,e_i^s, e_i^b) = \theta_i\cdot( e_i^b+ e_i^s)$ yields zero value, making low-type winners unacceptable.

\subsection{Discussions on value function}

Optimal mechanism characterization relies on aftermarket equilibrium existence. However, general value creation functions may preclude such equilibrium.

For example, consider a value creation function of the form:
\begin{equation}\label{nnnnnnnnn}
V_i(\theta_i,e_i^s, e_i^b) = \theta_i \cdot e_i^s \cdot e_i^b.
\end{equation}

The seller's and winner's aftermarket equilibrium strategies satisfy $e^b_i = (1-\alpha_i)\cdot\theta_i \cdot e^s_i(\mu)$ and $e^s_i =  \alpha_i\cdot\mathbb{E}_{\widetilde{\theta}_i}[\widetilde{\theta}_i \cdot e^b_i(\widetilde{\theta}_i)|\mu]$. Solving this system yields $e^b_i(\theta_i) = (1-\alpha_i) \cdot \alpha_i\cdot\theta_i \cdot \mathbb{E}_{\widetilde{\theta}_i}[\widetilde{\theta}_i \cdot e^b_i(\widetilde{\theta}_i)|\mu]$, which implies:
\begin{equation*}
\mathbb{E}_{\widetilde{\theta}_i}\left[\left.\widetilde{\theta}_i \cdot e^b_i(\widetilde{\theta}_i)\right|\mu\right] = \mathbb{E}_{\widetilde{\theta}_i}\left[\left.(1-\alpha_i) \cdot \alpha_i\cdot \widetilde{\theta}_i^2\right|\mu\right] \cdot \mathbb{E}_{\widetilde{\theta}_i}\left[\left.\widetilde{\theta}_i \cdot e^b_i(\widetilde{\theta}_i)\right|\mu\right],
\end{equation*}
holds for any posterior $\mu$. For equilibrium existence, the following must hold:
\begin{equation}\label{bnfnjfeajfhna}
\mathbb{E}_{\widetilde{\theta}_i}\left[\left.\widetilde{\theta}_i^2\right|\mu\right] = \frac{1}{(1-\alpha_i) \cdot \alpha_i}
\end{equation}
when $\mathbb{E}_{\widetilde{\theta}_i}[\widetilde{\theta}_i \cdot e^b_i(\widetilde{\theta}_i)|\mu] \neq 0$. Since \eqref{bnfnjfeajfhna} fails generally, \eqref{nnnnnnnnn} may preclude equilibrium existence.

General functions may also prevent the seller's payoff upper bound (Step 1 of \hyperref[Method]{Methodology}) from being attained, or yield non-monotonic virtual surplus. We therefore focus on two representative effort-complement forms balancing tractability with practical relevance in the main body of the article. Below, we consider a generalized nesting structure that encompasses both cases.

\subsection{Generalized nesting structure and seller's preference}

We consider a generalized value creation function nesting both cases via an interdependence term:
\begin{equation}\label{Vgeneral}
V_i(\theta_i, e_i^s, e^b_i) = \left(\zeta \cdot \theta_i + e_i^s\right) \cdot \left((1-\zeta) \cdot \theta_i + e_i^b\right),
\end{equation}
for $\zeta \in [0,1]$. When $\zeta\in[0,\frac{1}{2}]$, the seller's effort is more critical ($\zeta=0$ recovers the seller-pivotal case); when $\zeta\in(\frac{1}{2},1]$, the winner's effort is more critical ($\zeta=1$ recovers the winner-pivotal case).

\paragraph{Optimal mechanism.} Similar to the winner- and seller-pivotal scenarios, we use the superscript ``$opt$" to denote this generalized nesting structure. Applying the \hyperref[Method]{Methodology}, for any $\zeta \in [0,1]$, we first establish the seller's payoff upper bound:
\begin{equation}\label{uppergen}
\text{Rev}^{opt}(\bm{q}, \bm{c}, \bm{\pi}) \leq \overline{\text{Rev}}^{opt}(\bm{q}, \bm{c}, \bm{\pi}) \triangleq \int_{\bm{\Theta}} \max\left\{\phi_1^{opt}(\theta_1, \zeta), \cdots, \phi_n^{opt}(\theta_n, \zeta)\right\} \bm{f}(\bm{\theta}) d\bm{\theta},
\end{equation}
where
\begin{equation}\label{Nest-phi}
\begin{split}
\phi_i^{opt}(\theta_i,\zeta) \triangleq & \max_{\sigma_i(d\alpha_i|\bm{\theta})}\int_{0}^1 \left\{
\begin{gathered}
\underbrace{\left(\begin{gathered}
\zeta - \frac{\zeta^2 \cdot (1+\alpha_i^2)}{2} + \frac{\alpha_i \cdot (1-\alpha_i \cdot \zeta)\cdot(1-\zeta\cdot\alpha_i^2)}{1 - (1-\alpha_i) \cdot \alpha_i}\\
 - \frac{\alpha_i^4\cdot(1-\alpha_i\cdot\zeta)^2}{2(1 - (1-\alpha_i) \cdot \alpha_i)^2}
\end{gathered}
\right)}_{\triangleq A(\alpha_i,\zeta)} \cdot \theta_i^2\\
- \underbrace{(1-\alpha_i)\cdot\left(\begin{gathered}
\zeta\cdot(2 - \zeta\cdot(1+\alpha_i))\\
 + \frac{\alpha_i\cdot(1-\alpha_i\cdot\zeta)^2}{1 - (1-\alpha_i) \cdot \alpha_i}
\end{gathered}
\right)}_{\triangleq B(\alpha_i,\zeta)} \cdot \frac{1-F_i(\theta_i)}{f_i(\theta_i)}\cdot \theta_i
\end{gathered}
\right\}\sigma_i(d\alpha_i|\bm{\theta})\\
= & \int_{0}^1 \underbrace{\left\{A(\alpha_i,\zeta) \cdot \theta_i^2 - B(\alpha_i,\zeta) \cdot \frac{1-F_i(\theta_i)}{f_i(\theta_i)}\cdot \theta_i\right\}}_{\text{defined as }\Psi_i^{opt}(\alpha_i,\theta)}\underbrace{\sigma_i^{opt-\max}(d\alpha_i|\bm{\theta})}_{\text{Dirac measure maximizer}},
\end{split}
\end{equation}
which denotes bidder $i$'s virtual surplus (non-negative and increasing in $\theta_i$). The maximizer $\sigma_i^{opt-\max}(\cdot|\bm{\theta}) \in \Delta([0,1])$ is a Dirac measure concentrated at $\alpha_i^{opt-\max} = \alpha_i^{opt-\max}(\theta_i,\zeta)$, so $\phi_i^{opt}(\theta_i,\zeta) = \Psi_i^{opt}(\alpha_i^{opt-\max},\theta_i,\zeta)$. This function nests:
\begin{equation*}
\begin{cases}
\left.\phi_i^{opt}(\theta_i, \zeta)\right|_{\zeta = 1} = \phi_i^{wp}(\theta_i), & \text{winner-pivotal collaboration, cf. \eqref{Lmm1-Psi}},\\
\left.\phi_i^{opt}(\theta_i, \zeta)\right|_{\zeta = 0} = \phi_i^{sp}(\theta_i), & \text{seller-pivotal collaboration, cf. \eqref{Lmm1-Psi-1}}.
\end{cases}
\end{equation*}
The detailed derivation is provided in Appendix \ref{opt-mech-nesting-pf}.

Given the upper bound in \eqref{uppergen}, following Step 2 of the \hyperref[Method]{Methodology}, one can construct a feasible mechanism achieving $\overline{\text{Rev}}(\bm{q}, \bm{c}, \bm{\pi})$, thereby identifying the optimal mechanism, analogously to Theorems \ref{ThmCompl1-type2} and \ref{ThmCompl1-typeSL-seller}. Under the general nesting structure, however, additional conditions are required to guarantee feasibility and attainability. Specifically,
\begin{proposition}\label{ppppp07}
For any $\zeta \in [0,1]$, if the maximizing value-sharing scheme, $\alpha_i^{opt-\max}(\theta_i,\zeta)$, satisfies $\frac{\partial A(\alpha_i, \zeta)}{\partial \alpha_i}|_{\alpha_i = \alpha_i^{opt-\max}} < 0$ and $\frac{\partial H(\alpha_i, \zeta)}{\partial \alpha_i}|_{\alpha_i = \alpha_i^{opt-\max}} < 0$, where $H(\alpha_i, \zeta) \triangleq (1-\alpha_i)^2 \cdot (\zeta + \frac{\alpha_i\cdot(1-\alpha_i \cdot\zeta)}{1-(1-\alpha_i)\cdot\alpha_i})$. Then, an optimal linear mechanism exists that is feasible and achieves the seller's payoff upper bound in \eqref{uppergen}.
\end{proposition}
\begin{proof}
Please refer to Appendix \ref{ppppp07-pff}.
\end{proof}

These conditions are automatically satisfied in the winner-pivotal and seller-pivotal cases. The optimal mechanism selects the bidder with the highest virtual surplus, assigns deterministic value-sharing $\alpha_i^{\max}(\theta_i, \zeta)$ from \eqref{Nest-phi}, and reveals the winner's type. 

\paragraph{Seller's preference.} If the seller can choose the interdependence term $\zeta$ to determine aftermarket dominance, we claim that:

\begin{proposition}\label{pppp07}
The interdependence term that maximizes the seller's payoff upper bound $\overline{\text{Rev}}(\bm{q}, \bm{c}, \bm{\pi})$ satisfies $\zeta^* \leq \frac{1}{2}$. Additionally, if such $\zeta^*$ satisfies the conditions in Proposition \ref{ppppp07}, then the seller will choose this $\zeta^*$ such that:
\begin{equation*}
\text{Rev}^*(\zeta^*) = \overline{\text{Rev}}(\bm{q}, \bm{c}, \bm{\pi}).
\end{equation*}
Otherwise, the optimal $\zeta^*$ is intractable.
\end{proposition}
\begin{proof}
Please refer to Appendix \ref{pppp07-pf}.
\end{proof}

The seller's preference for $\zeta^* \leq \frac{1}{2}$ means she would choose her own effort to be more critical in value creation, consistent with the revenue advantage of seller-pivotal collaboration (Proposition \ref{ppproprev}).

\section{Conclusion \label{S6}}

This paper characterizes optimal auction mechanisms in which the asset's value is created endogenously through post-auction collaboration between the seller and the winning bidder. The role structure---winner-pivotal or seller-pivotal---fundamentally shapes the optimal mechanism. 

The optimal mechanisms in both settings share three structural properties: allocation to the bidder with the highest virtual surplus, a deterministic value-sharing rule, and full information revelation through the signal realization rule. The role structure nevertheless generates systematic differences: the seller captures a strictly larger value share under seller-pivotal collaboration---with full extraction from sufficiently low-type winners---whereas under winner-pivotal collaboration every winner must retain a positive share to sustain his critical effort. The pivotal party exerts higher post-auction effort in every setting, and seller-pivotal collaboration yields strictly higher seller revenue for any type distribution.

We implement the optimal mechanisms through ascending auctions with endogenously determined linear contracts. While perfect implementation under seller-pivotal collaboration is challenging due to surplus extraction incentives, close approximation is achievable. These results provide practical guidance for franchising, patent licensing, infrastructure concessions, and platform operating rights, where value is co-created through post-auction interaction.








\newpage
\appendix
\renewcommand{\thesection}{\Alph{section}.\arabic{section}}
\setcounter{section}{0}

\begin{appendices}\label{APPPPPPPPPPPP}

\section{Derivation of seller's expected payoff under winner-pivotal collaboration \label{9pf}}

Under the winner-pivotal scenario, in a feasible direct linear mechanism, the incentive compatibility condition for each bidder requires that $U_i^{wp}(\theta_i) = \max_{\hat{\theta}_i \in \Theta_i} U^{wp}_i(\hat{\theta}_i, e^{wp-b}_i, e^{wp-s}_i, \theta_i)$, where $U^{wp}_i(\hat{\theta}_i, e^{wp-b}_i, e^{wp-s}_i, \theta_i)$ is specified by \eqref{Uireporting}. Applying the general envelope theorem (see \cite{Milgrom2002}), any incentive-compatible direct mechanism must satisfy:
\begin{equation}\label{Spayoff}
U_i^{wp}(\theta_i) = U_i^{wp}(\underline{\theta}_i) + \int_{\underline{\theta}_i}^{\theta_i} W_i(\tau) d\tau,
\end{equation}
where $W_i(\tau)$ is defined as follows:
\begin{equation*}
\begin{split}
W_i(\tau) \triangleq & \int_{\bm{\Theta}_{-i}} q_i(\tau,\bm{\theta}_{-i}) \int_{\bm{\mathcal{T}}_{-i}^l}\int_{[0,1]\times \mathcal{T}_i^w} \\
& \cdot \int_{S_i} \left\{ (1-\alpha_i)^2\cdot \left(\tau + \frac{(1-\alpha_i) \cdot \alpha_i\cdot \mathbb{E}\left[\left.\widetilde{\theta}_i\right|\mu(\cdot |i, \bm{c}, s)\right]}{1 - (1-\alpha_i) \cdot \alpha_i}\right)
\right\} \pi_i(d s|i,\bm{c},\tau,\bm{\theta}_{-i}) \\
&\cdot \kappa_i(d\alpha_i,dt_i^w|\tau,\bm{\theta}_{-i})\bm{\nu}_{-i}^l(d\bm{t}_{-i}^l|\tau,\bm{\theta}_{-i})\bm{f}_{-i}(\bm{\theta}_{-i}) d\bm{\theta}_{-i}.
\end{split}
\end{equation*}
In addition, according to \eqref{Uireporting}, bidder $i$'s expected payoff from truthful reporting is:
\begin{equation}\label{nnbnbnbn}
\begin{split}
U_i^{wp}(\theta_i) 
= & \int_{\bm{\Theta}_{-i}} q_i(\theta_i,\bm{\theta}_{-i}) \int_{\bm{\mathcal{T}}_{-i}^l}\int_{[0,1]\times \mathcal{T}_i^w} \\
& \cdot \int_{S_i} \left\{\frac{1}{2}\cdot (1-\alpha_i)^2\cdot \left(\theta_i + \frac{(1-\alpha_i) \cdot \alpha_i\cdot \mathbb{E}\left[\left.\widetilde{\theta}_i\right|\mu(\cdot |i, \bm{c}, s)\right]}{1 - (1-\alpha_i) \cdot \alpha_i}\right)^2\right\} \pi_i(d s|i,\bm{c},\theta_i,\bm{\theta}_{-i}) \\
&\cdot \kappa_i(d\alpha_i,dt_i^w|\theta_i,\bm{\theta}_{-i}) \bm{\nu}_{-i}^l(d\bm{t}_{-i}^l|\theta_i,\bm{\theta}_{-i}) \bm{f}_{-i}(\bm{\theta}_{-i}) d\bm{\theta}_{-i}\\
& - \underbrace{\int_{\bm{\Theta}_{-i}} \left\{
\begin{gathered}
q_i(\theta_i,\bm{\theta}_{-i}) \int_{\mathcal{T}_i^w} t_i^w\nu_i^{w}(d t_i^w|\theta_i,\bm{\theta}_{-i})\\
+ \left(1 - q_i(\theta_i,\bm{\theta}_{-i})\right)\int_{\mathcal{T}_i^l} t_i^l\nu_i^{l}(d t_i^l|\theta_i,\bm{\theta}_{-i})
\end{gathered}
\right\}\bm{f}_{-i}(\bm{\theta}_{-i}) d\bm{\theta}_{-i}}_{\triangleq P_i(\theta_i)},
\end{split}
\end{equation}
where
\begin{equation}\label{parts_in_contract}
\nu_i^{w}(\cdot|\bm{\theta}) \triangleq \int_0^1 \kappa_i(d\alpha_i,\cdot|\bm{\theta})
\end{equation}
denotes the marginal probability measure of $\kappa_i(\cdot, \cdot|\bm{\theta})$, representing the winner's cash payment rule. Similarly, the winner's value-sharing rule is derived from the linear payment rule as:
\begin{equation}\label{parts_in_contract1}
\sigma_i(\cdot|\bm{\theta}) \triangleq \int_{\mathcal{T}_i^w} \kappa_i(\cdot,dt_i^w|\bm{\theta}).
\end{equation}

Combining \eqref{Spayoff} and \eqref{nnbnbnbn}, we obtain an expression for $P_i(\theta_i)$:
\begin{equation*}
\begin{split}
& \sum_{i\in\mathcal{N}}\int_{\Theta_i} P_i(\theta_i)f_i(\theta_i)d\theta_i \\
= & \sum_{i\in\mathcal{N}} \int_{\bm{\Theta}} q_i(\bm{\theta}) \cdot \int_{\bm{\mathcal{T}}_{-i}^l} \int_{[0,1]\times \mathcal{T}_i^w} \\
&\cdot \int_{S_i} \left\{
\begin{gathered}
\frac{1}{2}\cdot (1-\alpha_i)^2\cdot \left(\theta_i + \frac{(1-\alpha_i) \cdot \alpha_i\cdot  \mathbb{E}\left[\left.\widetilde{\theta}_i\right|\mu(\cdot |i, \bm{c}, s)\right]}{1 - (1-\alpha_i) \cdot \alpha_i}\right)^2\\
- (1-\alpha_i)^2 \cdot \left(\theta_i + \frac{(1-\alpha_i) \cdot \alpha_i\cdot  \mathbb{E}\left[\left.\widetilde{\theta}_i\right|\mu(\cdot |i, \bm{c}, s)\right]}{1 - (1-\alpha_i) \cdot \alpha_i}\right) \cdot \frac{1 - F_i(\theta_i)}{f_i(\theta_i)}
\end{gathered}
\right\}  \pi_i(d s|i,\bm{c},\bm{\theta}) \\
&\cdot \kappa_i(d\alpha_i,dt_i^w|\bm{\theta})\bm{\nu}_{-i}^l(d\bm{t}_{-i}^l|\bm{\theta}) \bm{f}(\bm{\theta}) d \bm{\theta} - \sum_{i\in\mathcal{N}}U_i^{wp}(\underline{\theta}_i).
\end{split}
\end{equation*}

Therefore, the desired result in \eqref{REV-000} follows from:
\begin{equation*}
\begin{split}
\text{Rev}^{wp}(\bm{q}, \bm{c}, \bm{\pi}) = &\sum_{i\in\mathcal{N}}\int_{\Theta_i} P_i(\theta_i)f_i(\theta_i)d\theta_i + \sum_{i\in\mathcal{N}} \int_{\bm{\Theta}} q_i(\bm{\theta}) \cdot \int_{\bm{\mathcal{T}}_{-i}^l} \int_{[0,1]\times \mathcal{T}_i^w} \\
& \cdot\int_{S_i} \left\{
\begin{gathered}
\alpha_i \cdot \left(\theta_i +  \frac{(1-\alpha_i) \cdot \alpha_i\cdot  \mathbb{E}\left[\left.\widetilde{\theta}_i\right|\mu(\cdot |i, \bm{c}, s)\right]}{1 - (1-\alpha_i) \cdot \alpha_i}\right) \\
(1-\alpha_i)\cdot \left(\theta_i + \frac{(1-\alpha_i) \cdot \alpha_i\cdot \mathbb{E}\left[\left.\widetilde{\theta}_i\right|\mu(\cdot |i, \bm{c}, s)\right]}{1 - (1-\alpha_i) \cdot \alpha_i}\right)\\
- \frac{1}{2}\cdot\left(\frac{(1-\alpha_i) \cdot \alpha_i\cdot  \mathbb{E}\left[\left.\widetilde{\theta}_i\right|\mu(\cdot |i, \bm{c}, s)\right]}{1 - (1-\alpha_i) \cdot \alpha_i}\right)^2
\end{gathered}
\right\} \pi_i(d s|i,\bm{c},\bm{\theta}) \\
&\cdot \kappa_i(d\alpha_i,dt_i^w|\bm{\theta})\bm{\nu}_{-i}^l(d\bm{t}_{-i}^l|\bm{\theta}) \bm{f}(\bm{\theta}) d \bm{\theta}.
\end{split}
\end{equation*}

\section{Proof of Lemma \ref{uppbdd} \label{lm3pf}}

Let $\xi(i, d\bm{c}, ds)$ denote the probability measure associated with bidder $i$ winning, where the direct linear mechanism generates payment profile $\bm{c} = (c_i, \bm{t}^l) = (\alpha_i, t_i^w, (t_j^l)_{j\neq i})$ and public signal $s\in S_i$. Conditional on this, from the seller's perspective, bidder types are drawn according to the probability measure $\varrho(\cdot|i, \bm{c}, s)\in \Delta (\bm{\Theta} )$, which implies $\int_{\bm{\Theta}_{-i}}\varrho(d\bm{\theta}|i, \bm{c}, s) = \mu(d\theta_{i}|i, \bm{c}, s)$. The seller's expected equilibrium payoff in \eqref{REV-000} can then be expressed as:
\begin{equation}\label{change-meas}
\begin{split}
& \text{Rev}^{wp}(\bm{q}, \bm{c}, \bm{\pi})\\
= & \sum_{i\in\mathcal{N}} \int_{[0,1]\times \mathcal{T}_i^w \times \bm{\mathcal{T}}_{-i}^l \times S_i} \left\{
\begin{gathered}
\frac{1-\alpha_i^2}{2} \cdot \mathbb{E}\left[\left.\widetilde{\theta}_i^2\right|\mu(\cdot |i, \bm{c}, s)\right] \\
+  \frac{(1-\alpha_i^2) \cdot(1-\alpha_i) \cdot \alpha_i}{1 - (1-\alpha_i) \cdot \alpha_i}\cdot  \mathbb{E}^2\left[\left.\widetilde{\theta}_i\right|\mu(\cdot |i, \bm{c}, s)\right]\\
- \frac{\alpha_i^2}{2}\cdot\left(\frac{(1-\alpha_i) \cdot \alpha_i}{1 - (1-\alpha_i) \cdot \alpha_i}\right)^2 \cdot  \mathbb{E}^2\left[\left.\widetilde{\theta}_i\right|\mu(\cdot |i, \bm{c}, s)\right] \\
- (1-\alpha_i)^2 \\
\cdot \mathbb{E}\left[\left.\widetilde{\theta}_i\right|\mu(\cdot |i, \bm{c}, s)\right] \cdot \mathbb{E}\left[\left.\frac{1 - F_i(\widetilde{\theta}_i)}{f_i(\widetilde{\theta}_i)}\right|\mu(\cdot |i, \bm{c}, s)\right] \\
- \frac{(1-\alpha_i)^2 \cdot (1-\alpha_i) \cdot \alpha_i}{1 - (1-\alpha_i) \cdot \alpha_i} \\
\cdot  \mathbb{E}\left[\left.\widetilde{\theta}_i\right|\mu(\cdot |i, \bm{c}, s)\right] \cdot \mathbb{E}\left[\left.\frac{1 - F_i(\widetilde{\theta}_i)}{f_i(\widetilde{\theta}_i)}\right|\mu(\cdot |i, \bm{c}, s)\right] 
\end{gathered}
\right\} \xi(i, d\bm{c}, ds)\\
& - \sum_{i\in\mathcal{N}}U_i^{wp}(\underline{\theta}_i).
\end{split}
\end{equation}

Given that $\mathbb{E}^2[\widetilde{\theta}_i|\mu(\widetilde{\theta}_i|i, \bm{c}, s)] \leq \mathbb{E}[\widetilde{\theta}_i^2|\mu(\widetilde{\theta}_i|i, \bm{c}, s)]$ by Jensen's inequality, and
\begin{equation*}
\mathbb{E}\left[\left.\frac{F_i(\widetilde{\theta}_i)-1}{f_i(\widetilde{\theta}_i)}\right|\mu(\widetilde{\theta}_i|i, \bm{c}, s)\right] \cdot \mathbb{E}\left[\left.\widetilde{\theta}_i\right|\mu(\widetilde{\theta}_i|i, \bm{c}, s)\right] \leq \mathbb{E}\left[\left.\frac{F_i(\widetilde{\theta}_i)-1}{f_i(\widetilde{\theta}_i)} \cdot \widetilde{\theta}_i\right|\mu(\widetilde{\theta}_i|i, \bm{c}, s)\right],
\end{equation*}
which follows from the covariance inequality for monotonic functions \citep[Theorem 2.2]{SCHMIDT201491}, given that the expectations $\mathbb{E}[\widetilde{\theta}_i|\mu(\widetilde{\theta}_i|i, \bm{c}, s)]$ and $\mathbb{E}[\frac{1-F_i(\widetilde{\theta}_i)}{f_i(\widetilde{\theta}_i)}|\mu(\widetilde{\theta}_i|i, \bm{c}, s)]$ are finite, and that $\frac{F_i(\widetilde{\theta}_i)-1}{f_i(\widetilde{\theta}_i)}$ is increasing due to the log-concavity of $f_i(\cdot)$. Moreover, individual rationality implies $\sum_{i \in \mathcal{N}} U_i^{wp}(\underline{\theta}_i) \geq 0$ since $U_i^{wp}(\theta_i) \geq 0$ for all $\theta_i \in \Theta_i$. We therefore obtain the following inequality:
\begin{equation}\label{winnerlead-inequa}
\begin{split}
& \text{Rev}^{wp}(\bm{q}, \bm{c}, \bm{\pi}) \\
\leq & \sum_{i\in\mathcal{N}} \int_{[0,1]\times \mathcal{T}_i^w \times \bm{\mathcal{T}}_{-i}^l \times S_i}\\
& \left\{
\begin{gathered}
\underbrace{\left(\frac{1-\alpha_i^2}{2} + \frac{(1-\alpha_i^2) \cdot (1-\alpha_i) \cdot \alpha_i}{1 - (1-\alpha_i) \cdot \alpha_i} - \frac{\alpha_i^2}{2} \cdot \frac{(1-\alpha_i)^2 \cdot \alpha_i^2}{(1 - (1-\alpha_i) \cdot \alpha_i)^2}\right)}_{\text{non-negative}} \\
\cdot \mathbb{E}\left[\left.\widetilde{\theta}_i^2\right|\mu(\cdot|i, \bm{c}, s)\right] \\
- \underbrace{\left((1-\alpha_i)^2 + \frac{(1-\alpha_i)^2 \cdot (1-\alpha_i) \cdot \alpha_i}{1 - (1-\alpha_i) \cdot \alpha_i} \right)}_{\text{non-negative}} \\
\cdot \mathbb{E}\left[\left.\frac{1 - F_i(\widetilde{\theta}_i)}{f_i(\widetilde{\theta}_i)}\cdot \widetilde{\theta}_i \right|\mu(\cdot|i, \bm{c}, s)\right]
\end{gathered}
\right\} \xi(i, d\bm{c}, ds)\\
= & \sum_{i\in\mathcal{N}} \int_{[0,1]\times \mathcal{T}_i^w \times \bm{\mathcal{T}}_{-i}^l \times S_i} \int_{\bm{\Theta}} \left\{
\begin{gathered}
\frac{1 - 2\alpha_i \cdot(1 - \alpha_i)\cdot \alpha_i - \alpha_i^4}{2(1 - (1-\alpha_i) \cdot \alpha_i)^2} \cdot \theta_i^2 \\
- \frac{(1-\alpha_i)^2}{1 - (1-\alpha_i) \cdot \alpha_i} \cdot \frac{1-F_i(\theta_i)}{f_i(\theta_i)}\cdot \theta_i
\end{gathered}
\right\} \varrho(d\bm{\theta}|i, \bm{c}, s)\xi(i, d\bm{c}, ds)\\
= & \sum_{i\in\mathcal{N}}\int_{\bm{\Theta}} q_i(\bm{\theta}) \cdot \int_{\bm{\mathcal{T}}_{-i}^l} \int_{[0,1]\times \mathcal{T}_i^w} \\
& \cdot\int_{S_i} \left\{
\begin{gathered}
\frac{1 - 2\alpha_i \cdot(1 - \alpha_i)\cdot \alpha_i - \alpha_i^4}{2(1 - (1-\alpha_i) \cdot \alpha_i)^2} \cdot \theta_i^2 \\
- \frac{(1-\alpha_i)^2}{1 - (1-\alpha_i) \cdot \alpha_i} \cdot \frac{1-F_i(\theta_i)}{f_i(\theta_i)}\cdot \theta_i
\end{gathered}
\right\} \pi_i(d s|i,\bm{c},\bm{\theta}) \\
&\cdot \kappa_i(d\alpha_i,dt_i^w|\bm{\theta}) \bm{\nu}_{-i}^l(d\bm{t}_{-i}^l|\bm{\theta}) \bm{f}(\bm{\theta}) d \bm{\theta}\\
= & \sum_{i\in\mathcal{N}}\int_{\bm{\Theta}} q_i(\bm{\theta}) \cdot \int_{0}^1 \left\{
\begin{gathered}
\underbrace{\frac{1 - 2\alpha_i \cdot(1 - \alpha_i)\cdot \alpha_i - \alpha_i^4}{2(1 - (1-\alpha_i) \cdot \alpha_i)^2}}_{\triangleq A^{wp}(\alpha_i)} \cdot \theta_i^2 \\
- \underbrace{\frac{(1-\alpha_i)^2}{1 - (1-\alpha_i) \cdot \alpha_i}}_{\triangleq B^{wp}(\alpha_i)} \cdot \frac{1-F_i(\theta_i)}{f_i(\theta_i)}\cdot \theta_i
\end{gathered}
\right\}\sigma_i(d\alpha_i|\bm{\theta}) \bm{f}(\bm{\theta}) d\bm{\theta},
\end{split}
\end{equation}
where the equalities follow from a change of measure and the marginal probability measure $\sigma_i(\cdot|\bm{\theta})$ defined in \eqref{parts_in_contract1}. 

We focus on the maximization problem with respect to $\sigma_i(\cdot|\bm{\theta})$. For any type profile $\bm{\theta}\in \bm{\Theta}$, we solve:
\begin{equation}\label{firstABdef}
\max_{\sigma_i(d\alpha_i|\bm{\theta})}\int_{0}^1 \left\{A^{wp}(\alpha_i) \cdot \theta_i^2 - B^{wp}(\alpha_i) \cdot \frac{1-F_i(\theta_i)}{f_i(\theta_i)}\cdot \theta_i\right\} \sigma_i(d\alpha_i|\bm{\theta}).
\end{equation}

The maximizer, denoted $\sigma_i^{wp-\max}(\cdot|\bm{\theta}) \in \Delta([0,1])$, is a Dirac measure that places unit mass on $\alpha_i^{wp-\max} = \alpha_i^{wp-\max}(\theta_i)$, where
\begin{equation*}
\alpha_i^{wp-\max}(\theta_i) \in \argmax_{\alpha_i\in[0,1]} \Psi_i^{wp}(\alpha_i,\theta_i) \triangleq \argmax_{\alpha_i\in[0,1]} \left\{A^{wp}(\alpha_i) \cdot \theta_i^2 - B^{wp}(\alpha_i) \cdot \frac{1-F_i(\theta_i)}{f_i(\theta_i)}\cdot \theta_i\right\}.
\end{equation*}

We define the corresponding maximal value as bidder $i$'s virtual surplus:
\begin{equation}\label{phiphiphi-type2-pf}
\phi_i^{wp}(\theta_i) \triangleq \int_{0}^1 \Psi_i^{wp}(\alpha_i,\theta_i)\sigma_i^{wp-\max}(d\alpha_i|\bm{\theta}) = \Psi_i^{wp}(\alpha_i^{wp-\max},\theta_i),
\end{equation}
which matches the expression in \eqref{Lmm1-Psi}. For any feasible direct linear mechanism $(\bm{q}, \bm{c}, \bm{\pi})$, the seller's expected payoff under winner-pivotal collaboration is therefore bounded above by $\overline{\text{Rev}}^{wp}(\bm{q}, \bm{c}, \bm{\pi})$, characterized by:
\begin{equation*}
\begin{split}
\overline{\text{Rev}}^{wp}(\bm{q}, \bm{c}, \bm{\pi}) \triangleq & \max_{\substack{0\leq q_i(\bm{\theta}) \leq 1;\\ \sum_{i\in\mathcal{N}} q_i(\bm{\theta}) \leq 1}} \sum_{i\in\mathcal{N}}\int_{\bm{\Theta}} q_i(\bm{\theta}) \cdot \phi^{wp}(\theta_i) \bm{f}(\bm{\theta}) d\bm{\theta}\\
= & \int_{\bm{\Theta}} \max\left\{\phi_1^{wp}(\theta_1), \cdots, \phi_i^{wp}(\theta_i),\cdots,\phi_n^{wp}(\theta_n)\right\} \bm{f}(\bm{\theta}) d\bm{\theta}.
\end{split}
\end{equation*}

\section{Proof of Lemma \ref{Lm-characterize-alphamax} \label{Lm-characterize-alphamax-PF}}

\paragraph{Interior maximizer.} The maximizer $\alpha_i^{wp-\max} = \alpha_i^{wp-\max}(\theta_i)$ in \eqref{alphamax111-type2} is the solution to the following constrained optimization problem:
\begin{equation*}
\max_{\alpha_i} \left\{A^{wp}(\alpha_i) \cdot \theta_i^2 - B^{wp}(\alpha_i) \cdot \frac{1-F_i(\theta_i)}{f_i(\theta_i)}\cdot \theta_i\right\}, \;\; \text{s.t. } \;\; 1 - \alpha_i \geq 0\;\;\text{ and } \;\; \alpha_i \geq 0.
\end{equation*}

The Kuhn-Tucker conditions give:
\begin{equation}\label{KKT}
\begin{cases}
\frac{1 - 4\alpha_i^{wp-\max} + 3(\alpha_i^{wp-\max})^2 - (\alpha_i^{wp-\max})^3}{(1 - \alpha_i^{wp-\max} + (\alpha_i^{wp-\max})^2)^3} \cdot \theta_i^2 + \frac{1 - (\alpha_i^{wp-\max})^2}{(1 - \alpha_i^{wp-\max} + (\alpha_i^{wp-\max})^2)^2} \cdot \frac{1-F_i(\theta_i)}{f_i(\theta_i)}\cdot \theta_i - \lambda_1 + \lambda_2 = 0;\\
\lambda_1 \cdot (1 - \alpha_i^{wp-\max}) = 0,\;\; \lambda_2 \cdot \alpha_i^{wp-\max} = 0;\\
1 - \alpha_i^{wp-\max} \geq 0, \;\; \alpha_i^{wp-\max} \geq 0;\\
\lambda_1 \geq 0, \;\; \lambda_2 \geq 0.
\end{cases}
\end{equation}

Applying the complementary slackness conditions: (1) if $\lambda_1 > 0$ and $\lambda_2 = 0$, then $\alpha_i^{wp-\max} = 1$, which yields $-\theta_i^2 - \lambda_1 = 0$, contradicting $\lambda_1 > 0$; (2) if $\lambda_1 = 0$ and $\lambda_2 > 0$, then $\alpha_i^{wp-\max} = 0$, which yields $\theta_i^2 + \frac{1-F_i(\theta_i)}{f_i(\theta_i)} \cdot \theta_i + \lambda_2 = 0$, contradicting $\lambda_2 > 0$. Hence, $\alpha_i^{wp-\max}$ cannot be a corner solution.

\paragraph{Decreasing.} The interior solution $\alpha_i^{wp-\max} \in (0,1)$ implies that $\lambda_1 = 0$ and $\lambda_2 = 0$ in \eqref{KKT}, yielding:
\begin{equation*}
\frac{1 - 4\alpha_i^{wp-\max} + 3(\alpha_i^{wp-\max})^2 - (\alpha_i^{wp-\max})^3}{(1 - \alpha_i^{wp-\max} + (\alpha_i^{wp-\max})^2)^3} \cdot \theta_i^2 + \frac{1 - (\alpha_i^{wp-\max})^2}{(1 - \alpha_i^{wp-\max} + (\alpha_i^{wp-\max})^2)^2} \cdot \frac{1-F_i(\theta_i)}{f_i(\theta_i)}\cdot \theta_i = 0,
\end{equation*}
which gives the following stationary condition for the interior maximizer:
\begin{equation}\label{CL-decrease-OptAlpha-00}
\frac{(\alpha_i^{wp-\max})^3 - 3(\alpha_i^{wp-\max})^2 + 4\alpha_i^{wp-\max} - 1}{(1 - \alpha_i^{wp-\max} + (\alpha_i^{wp-\max})^2)\cdot (1 - (\alpha_i^{wp-\max})^2)} = \frac{1-F_i(\theta_i)}{\theta_i \cdot f_i(\theta_i)}.
\end{equation}

Since $\frac{1-F_i(\theta_i)}{\theta_i \cdot f_i(\theta_i)} \geq 0$, we have:
\begin{equation}\label{CL-decrease-OptAlpha-01}
(\alpha_i^{wp-\max})^3 - 3(\alpha_i^{wp-\max})^2 + 4\alpha_i^{wp-\max} - 1 \geq 0,
\end{equation}
where the equality in \eqref{CL-decrease-OptAlpha-01} holds only if $\theta_i = \overline{\theta}_i$ since $\frac{1-F_i(\overline{\theta}_i)}{\overline{\theta}_i \cdot f_i(\overline{\theta}_i)} = 0$. For all $\theta_i < \overline{\theta}_i$, taking the logarithm of \eqref{CL-decrease-OptAlpha-00} yields:
\begin{equation}\label{take-log-interiormax}
\left(
\begin{gathered}
\ln\left((\alpha_i^{wp-\max})^3 - 3(\alpha_i^{wp-\max})^2 + 4\alpha_i^{wp-\max} - 1\right) \\
- \ln\left(1 - \alpha_i^{wp-\max} + (\alpha_i^{wp-\max})^2\right) \\
- \ln\left(1 - \alpha_i^{wp-\max}\right) - \ln\left(1 +  \alpha_i^{wp-\max}\right)
\end{gathered}
\right) = \ln\left(\frac{1-F_i(\theta_i)}{f_i(\theta_i)}\right) - \ln\left(\theta_i\right).
\end{equation}

Differentiating both sides with respect to $\theta_i$ yields:
\begin{equation}\label{CL-decrease-OptAlpha-01-01}
\begin{split}
& \left(
\begin{gathered}
\underbrace{\frac{3(\alpha_i^{wp-\max})^2 - 6 \alpha_i^{wp-\max} + 4}{(\alpha_i^{wp-\max})^3 - 3(\alpha_i^{wp-\max})^2 + 4\alpha_i^{wp-\max} - 1} - \frac{2\alpha_i^{wp-\max} - 1}{1 - \alpha_i^{wp-\max} + (\alpha_i^{wp-\max})^2}}_{\text{term 1}} \\
+ \underbrace{\frac{1}{1 - \alpha_i^{wp-\max}} - \frac{1}{1 + \alpha_i^{wp-\max}}}_{\text{term 2}}
\end{gathered}
\right)\cdot \frac{d \alpha_i^{wp-\max}(\theta_i)}{d\theta_i} \\
= & \frac{f_i(\theta_i)}{1-F_i(\theta_i)} \cdot \frac{d}{d \theta_i}\left(\frac{1-F_i(\theta_i)}{f_i(\theta_i)}\right) - \frac{1}{\theta_i} < 0,
\end{split}
\end{equation}
where the right-hand side of \eqref{CL-decrease-OptAlpha-01-01} is negative because $\frac{1-F_i(\theta_i)}{f_i(\theta_i)}$ decreases in $\theta_i$ under the log-concavity of $f_i(\cdot)$. For the left-hand side, term 2 is positive and
\begin{equation*}
\text{term 1} = \frac{(1 - \alpha_i^{wp-\max})\cdot(3 - \alpha_i^{wp-\max} + (1 - \alpha_i^{wp-\max})\cdot (\alpha_i^{wp-\max})^2)}{((\alpha_i^{wp-\max})^3 - 3(\alpha_i^{wp-\max})^2 + 4\alpha_i^{wp-\max} - 1) \cdot (1 - \alpha_i^{wp-\max} + (\alpha_i^{wp-\max})^2)} > 0,
\end{equation*}
by \eqref{CL-decrease-OptAlpha-01}. Therefore, $\frac{d \alpha_i^{wp-\max}(\theta_i)}{d\theta_i} < 0$.

\paragraph{Uniqueness.} The second-order partial derivative with respect to $\alpha_i$ in \eqref{Lmm1-Psi} is:
\begin{equation}\label{SOpartial}
\frac{\partial^2 \Psi_i^{wp}(\alpha_i, \theta_i)}{\partial \alpha_i^2} = \frac{\theta_i^2}{(1 - \alpha_i + \alpha_i^2)^4} \cdot \underbrace{\left(
\begin{gathered}
(2 - 8\alpha_i + 8\alpha_i^2 - 4\alpha_i^3 - 2\alpha_i^4 + 2\alpha_i^5) \cdot \frac{1 - F_i(\theta_i)}{\theta_i \cdot f_i(\theta_i)} \\
+ (-1 - 8\alpha_i + 20\alpha_i^2 - 12\alpha_i^3 + 3\alpha_i^4)
\end{gathered}
\right)}_{\triangleq \Gamma_i(\alpha_i, \theta_i)}.
\end{equation}

Applying the stationary condition \eqref{CL-decrease-OptAlpha-00}, we evaluate \eqref{SOpartial} at $\alpha_i = \alpha_i^{wp-\max}$ as:
\begin{equation}\label{alphamaxinconcave}
\begin{split}
\left. \frac{\partial^2 \Psi_i^{wp}(\alpha_i, \theta_i)}{\partial \alpha_i^2} \right|_{\alpha_i = \alpha_i^{wp-\max}} = & -\frac{(\alpha_i^{wp-\max})^3 - 3(\alpha_i^{wp-\max})^2 + 4\alpha_i^{wp-\max} - 1}{(1 - \alpha_i^{wp-\max} + (\alpha_i^{wp-\max})^2)^3 \cdot (1 - (\alpha_i^{wp-\max})^2)}\\ & \cdot \left(
\begin{gathered}
\frac{3(\alpha_i^{wp-\max})^2 - 6\alpha_i^{wp-\max} + 4}{(\alpha_i^{wp-\max})^3 - 3(\alpha_i^{wp-\max})^2 + 4\alpha_i^{wp-\max} - 1} \\
- \frac{2\alpha_i^{wp-\max} - 1}{1 - \alpha_i^{wp-\max} + (\alpha_i^{wp-\max})^2} \\
+ \frac{2\alpha_i^{wp-\max}}{1 - (\alpha_i^{wp-\max})^2}
\end{gathered}
\right) \cdot \theta_i^2 \leq 0,
\end{split}
\end{equation}
where the inequality holds by \eqref{CL-decrease-OptAlpha-01} and the positivity of the bracketed terms in \eqref{CL-decrease-OptAlpha-01-01}. Thus, \eqref{alphamaxinconcave} establishes the local concavity at the interior maximizer $\alpha_i^{wp-\max}$.

For global uniqueness, from \eqref{SOpartial} we compute:
\begin{equation*}
\begin{split}
\frac{\partial^2 \Gamma_i(\alpha_i, \theta_i)}{\partial \alpha_i^2} = & (16 - 24\alpha_i - 24\alpha_i^2 + 40\alpha_i^3) \cdot \frac{1 - F_i(\theta_i)}{\theta_i \cdot f_i(\theta_i)} + (40 - 72\alpha_i + 36\alpha_i^2)\\
= & 8\underbrace{(2 - 3\alpha_i - 3\alpha_i^2 + 5\alpha_i^3)}_{\triangleq \Lambda(\alpha_i)} \cdot \frac{1 - F_i(\theta_i)}{\theta_i \cdot f_i(\theta_i)} + 36(1-\alpha_i)^2 + 4.
\end{split}
\end{equation*}

Since $\Lambda(\alpha_i) \geq \Lambda(\frac{1+\sqrt{6}}{5}) > 0$ for all $\alpha_i \in [0,1]$, $\Gamma_i(\alpha_i, \theta_i)$ is strictly convex in $\alpha_i$. Thus,
\begin{equation*}
\frac{\partial \Gamma_i(\alpha_i, \theta_i)}{\partial \alpha_i} = (-8 + 16\alpha_i - 12\alpha_i^2 - 8\alpha_i^3 + 10\alpha_i^4) \cdot \frac{1 - F_i(\theta_i)}{\theta_i \cdot f_i(\theta_i)} + (-8 + 40\alpha_i - 36\alpha_i^2 + 12\alpha_i^3)
\end{equation*}
increases in $\alpha_i$.

(1) Case 1. For $0 \leq \frac{1 - F_i(\theta_i)}{\theta_i \cdot f_i(\theta_i)} \leq \frac{1}{2}$, we have:
\begin{equation}\label{cornerrrrrr}
\left. \frac{\partial \Gamma_i(\alpha_i, \theta_i)}{\partial \alpha_i} \right|_{\alpha_i = 0} = -8 \cdot \frac{1 - F_i(\theta_i)}{\theta_i \cdot f_i(\theta_i)} - 8 < 0; \; \left. \frac{\partial \Gamma_i(\alpha_i, \theta_i)}{\partial \alpha_i} \right|_{\alpha_i = 1} = -2 \cdot \frac{1 - F_i(\theta_i)}{\theta_i \cdot f_i(\theta_i)} + 8 > 0,
\end{equation}
thus, $\frac{\partial \Gamma_i(\alpha_i, \theta_i)}{\partial \alpha_i}$ has a unique real root; that is, there exists a unique $\alpha_i^0 \in (0,1)$ such that $\frac{\partial \Gamma_i(\alpha_i, \theta_i)}{\partial \alpha_i}|_{\alpha_i = \alpha_i^0} = 0$. Hence, $\Gamma_i(\alpha_i, \theta_i)$ is decreasing on $[0,\alpha_i^0]$, which implies $\Gamma_i(\alpha_i, \theta_i) < 0$ for any $\alpha_i \in (0,\alpha_i^0)$ since $\Gamma_i(\alpha_i^0, \theta_i) < \Gamma_i(0, \theta_i) = 2 \cdot \frac{1 - F_i(\theta_i)}{\theta_i \cdot f_i(\theta_i)} - 1 < 0$ by the definition of $\Gamma_i(\alpha_i, \theta_i)$ in \eqref{SOpartial}. Meanwhile, $\Gamma_i(\alpha_i, \theta_i)$ is increasing on $[\alpha_i^0,1]$. Combined with $\Gamma_i(1, \theta_i) = - 2 \cdot \frac{1 - F_i(\theta_i)}{\theta_i \cdot f_i(\theta_i)} + 2 > 0$, this implies that $\Gamma_i(\alpha_i, \theta_i)$ has a unique real root on $[\alpha_i^0,1]$, denoted $\hat{\alpha}_i^0$. Therefore, $\Gamma_i(\alpha_i, \theta_i) < 0$ for any $\alpha_i \in (\alpha_i^0,\hat{\alpha}_i^0)$ and $\Gamma_i(\alpha_i, \theta_i) > 0$ for any $\alpha_i \in (\hat{\alpha}_i^0,1)$. From the second-order partial derivative in \eqref{SOpartial}, we have:
\begin{equation*}
\frac{\partial^2 \Psi_i^{wp}(\alpha_i, \theta_i)}{\partial \alpha_i^2} = \frac{\theta_i^2}{(1 - \alpha_i + \alpha_i^2)^4} \cdot \Gamma_i(\alpha_i, \theta_i) \begin{cases}
\leq 0, & \text{for any } \alpha_i \in [0, \hat{\alpha}_i^0]\\
\geq 0, & \text{for any } \alpha_i \in [\hat{\alpha}_i^0, 1]
\end{cases}.
\end{equation*}

Note that the maximum of $\Psi_i^{wp}(\alpha_i, \theta_i)$ on $[\hat{\alpha}_i^0, 1]$ is attained at the corner since $\Psi_i^{wp}(\alpha_i, \theta_i)$ is convex on $[\hat{\alpha}_i^0, 1]$. Given the concavity of $\Psi_i^{wp}(\alpha_i, \theta_i)$ on $[0, \hat{\alpha}_i^0]$ and $\alpha_i^{wp-\max} \in [0, \hat{\alpha}_i^0]$ by \eqref{alphamaxinconcave}, to verify that $\alpha_i^{wp-\max} \in [0, 1]$ is the global unique maximizer, we need only compare the corresponding objective value to that at the corner $\alpha_i = 1$. Specifically, we have $\Psi_i^{wp}(\alpha_i^{wp-\max}, \theta_i) \geq \Psi_i^{wp}(0, \theta_i) = (\frac{1}{2} - \frac{1 - F_i(\theta_i)}{\theta_i \cdot f_i(\theta_i)}) \cdot \theta_i^2 \geq 0 = \Psi_i^{wp}(1, \theta_i)$, where the first inequality follows from the concavity of $\Psi_i^{wp}(\alpha_i, \theta_i)$ on $[0, \hat{\alpha}_i^0]$.

(2) Case 2. For $\frac{1}{2} \leq \frac{1 - F_i(\theta_i)}{\theta_i \cdot f_i(\theta_i)} \leq 1$, the relations in \eqref{cornerrrrrr} still hold. Thus, we obtain the same results as in Case 1: $\frac{\partial \Gamma_i(\alpha_i, \theta_i)}{\partial \alpha_i}$ has a unique real root $\alpha_i^0 \in (0, 1)$, and $\Gamma_i(\alpha_i, \theta_i)$ is decreasing on $[0, \alpha_i^0]$ and increasing on $[\alpha_i^0, 1]$. From \eqref{SOpartial} and \eqref{alphamaxinconcave}, we have:
\begin{equation*}
\left. \frac{\partial^2 \Psi_i^{wp}(\alpha_i, \theta_i)}{\partial \alpha_i^2} \right|_{\alpha_i = \alpha_i^{wp-\max}} = \frac{\theta_i^2}{(1 - \alpha_i^{wp-\max} + (\alpha_i^{wp-\max})^2)^4} \cdot \Gamma_i(\alpha_i^{wp-\max}, \theta_i) \leq 0,
\end{equation*}
which implies $\Gamma_i(\alpha_i^{wp-\max}, \theta_i) \leq 0$. The minimum value $\Gamma_i(\alpha_i^0, \theta_i)$ satisfies:
\begin{equation}\label{alpha0negative}
\Gamma_i(\alpha_i^0, \theta_i) \leq \Gamma_i(\alpha_i^{wp-\max}, \theta_i) \leq 0.
\end{equation}

Given that $\Gamma_i(0, \theta_i) = 2 \cdot \frac{1 - F_i(\theta_i)}{\theta_i \cdot f_i(\theta_i)} - 1 > 0$ and $\Gamma_i(1, \theta_i) = - 2 \cdot \frac{1 - F_i(\theta_i)}{\theta_i \cdot f_i(\theta_i)} + 2 > 0$, there exist two real roots, $\bar{\alpha}_i^0 \in (0,\alpha_i^0)$ and $\hat{\alpha}_i^0 \in (\alpha_i^0,1)$, such that:
\begin{equation*}
\begin{split}
& \Gamma_i(\alpha_i, \theta_i) \begin{cases}
> 0, & \text{for any } \alpha_i \in [0, \bar{\alpha}_i^0) \cup (\hat{\alpha}_i^0,1]\\
< 0, & \text{for any } \alpha_i \in (\bar{\alpha}_i^0, \hat{\alpha}_i^0)
\end{cases}\\
\Longrightarrow \; & \frac{\partial^2 \Psi_i^{wp}(\alpha_i, \theta_i)}{\partial \alpha_i^2} = \frac{\theta_i^2}{(1 - \alpha_i + \alpha_i^2)^4} \cdot \Gamma_i(\alpha_i, \theta_i) \begin{cases}
\geq 0, & \text{for any } \alpha_i \in [0, \bar{\alpha}_i^0] \cup [\hat{\alpha}_i^0,1]\\
\leq 0, & \text{for any } \alpha_i \in [\bar{\alpha}_i^0, \hat{\alpha}_i^0]
\end{cases}.
\end{split}
\end{equation*}

To establish the global unique maximizer of $\alpha_i^{wp-\max} \in [0, 1]$, we compare the corresponding objective values: $\Psi_i^{wp}(\alpha_i^{wp-\max}, \theta_i)$, $\Psi_i^{wp}(0, \theta_i)$, and $\Psi_i^{wp}(1, \theta_i)$. Given the objective function in \eqref{Lmm1-Psi} and the stationary condition \eqref{CL-decrease-OptAlpha-00}, we show that $\Psi_i^{wp}(\alpha_i^{wp-\max}, \theta_i) > 0$:
\begin{equation}\label{Psimaxpositive}
\begin{split}
& \Psi_i^{wp}(\alpha_i^{wp-\max}, \theta_i)\\
= & \frac{(1 - \alpha_i^{wp-\max}) \cdot \left((3 + (\alpha_i^{wp-\max})^2) \cdot (1 - \alpha_i^{wp-\max})^2 + 2(\alpha_i^{wp-\max})^2 \right)}{2 \left( 1 - \alpha_i^{wp-\max} + (\alpha_i^{wp-\max})^2 \right)^2 \cdot (1 + \alpha_i^{wp-\max})} \cdot \theta_i^2 > 0.
\end{split}
\end{equation}

Therefore, we establish that $\Psi_i^{wp}(\alpha_i^{wp-\max}, \theta_i)>\Psi_i^{wp}(0, \theta_i),\Psi_i^{wp}(1, \theta_i)$, since $\Psi_i^{wp}(0, \theta_i) = (\frac{1}{2} - \frac{1 - F_i(\theta_i)}{\theta_i \cdot f_i(\theta_i)}) \cdot \theta_i^2 < 0$ and $\Psi_i^{wp}(1, \theta_i) = 0$.

(3) Case 3. For $\frac{1 - F_i(\theta_i)}{\theta_i \cdot f_i(\theta_i)} > 1$, following the analysis in the previous two cases, there exists a unique real root $\alpha_i^0 \in (0, 1)$ such that $\Gamma_i(\alpha_i, \theta_i)$ decreases on $[0, \alpha_i^0]$ and increases on $[\alpha_i^0, 1]$. Given inequalities \eqref{alpha0negative}, and that $\Gamma_i(0, \theta_i) = 2 \cdot \frac{1 - F_i(\theta_i)}{\theta_i \cdot f_i(\theta_i)} - 1 > 0$ and $\Gamma_i(1, \theta_i) = - 2 \cdot \frac{1 - F_i(\theta_i)}{\theta_i \cdot f_i(\theta_i)} + 2 < 0$, there exists a unique real root $\hat{\alpha}_i^0 \in (0, \alpha_i^0)$ such that:
\begin{equation*}
\begin{split}
& \Gamma_i(\alpha_i, \theta_i) \begin{cases}
> 0, & \text{for any } \alpha_i \in [0, \hat{\alpha}_i^0)\\
< 0, & \text{for any } \alpha_i \in (\hat{\alpha}_i^0,1]
\end{cases}\\
\Longrightarrow \; & \frac{\partial^2 \Psi_i^{wp}(\alpha_i, \theta_i)}{\partial \alpha_i^2} = \frac{\theta_i^2}{(1 - \alpha_i + \alpha_i^2)^4} \cdot \Gamma_i(\alpha_i, \theta_i) \begin{cases}
\geq 0, & \text{for any } \alpha_i \in [0, \hat{\alpha}_i^0]\\
\leq 0, & \text{for any } \alpha_i \in [\hat{\alpha}_i^0,1]
\end{cases}.
\end{split}
\end{equation*}

Hence, to establish the global unique maximizer $\alpha_i^{wp-\max} \in [0, 1]$, we need only compare the corresponding objective value to that at the corner. Specifically, we have $\Psi_i^{wp}(\alpha_i^{wp-\max}, \theta_i) > \Psi_i^{wp}(0, \theta_i)$ since $\Psi_i^{wp}(\alpha_i^{wp-\max}, \theta_i) > 0$ by \eqref{Psimaxpositive} and $\Psi_i^{wp}(0, \theta_i) = (\frac{1}{2} - \frac{1 - F_i(\theta_i)}{\theta_i \cdot f_i(\theta_i)}) \cdot \theta_i^2 < 0$.

\section{Proof of Lemma \ref{increasing-of-surplus} \label{PFF-increasing-of-surplus}}

Applying the general envelope theorem, we differentiate \eqref{phiphiphi-type2-pf} with respect to $\theta_i$ to obtain:
\begin{equation*}
\begin{split}
\frac{d \phi_i^{wp}(\theta_i)}{d \theta_i} = & A^{wp}(\alpha_i^{wp-\max}) \cdot  2 \theta_i - B^{wp}(\alpha_i^{wp-\max}) \cdot \left(\frac{1-F_i(\theta_i)}{f_i(\theta_i)} + \frac{d}{d \theta_i} \left(\frac{1-F_i(\theta_i)}{f_i(\theta_i)}\right) \cdot \theta_i\right).
\end{split}
\end{equation*}
The term $B^{wp}(\alpha_i^{wp-\max})\cdot \frac{d}{d \theta_i} (\frac{1-F_i(\theta_i)}{f_i(\theta_i)}) \cdot \theta_i$ is negative because $B^{wp}(\alpha_i^{wp-\max}) > 0$ (since $\alpha_i^{wp-\max}$ is an interior maximizer) and $\frac{1-F_i(\theta_i)}{f_i(\theta_i)}$ is decreasing in $\theta_i$. For the remaining terms, using the stationary condition for the interior maximizer \eqref{CL-decrease-OptAlpha-00}, we obtain that $A^{wp}(\alpha_i^{wp-\max}) \cdot 2\theta_i - B^{wp}(\alpha_i^{wp-\max})\cdot \frac{1-F_i(\theta_i)}{f_i(\theta_i)}$ equals:
\begin{equation*}
\frac{\left(1-\alpha_i^{wp-\max}\right)\cdot \left(2 - 2\alpha_i^{wp-\max} + 3(\alpha_i^{wp-\max})^2 -(\alpha_i^{wp-\max})^3 + (\alpha_i^{wp-\max})^4\right) \cdot \theta_i}{(1-\alpha_i^{wp-\max}+(\alpha_i^{wp-\max})^2)^2 \cdot (1+\alpha_i^{wp-\max})} > 0.
\end{equation*}
This establishes that $\frac{d \phi_i^{wp}(\theta_i)}{d \theta_i} > 0$. It remains to show that $\phi_i^{wp}(\theta_i) \geq 0$ for each $i\in\mathcal{N}$ and any $\theta_i\in\Theta_i$. By the definition of virtual surplus in \eqref{phiphiphi-type2-pf}, the result follows from:
\begin{equation}\label{non-negateive-phiWL}
\phi_i^{wp}(\theta_i) = \max_{\alpha_i \in [0,1]} \Psi_i^{wp}(\alpha_i,\theta_i) \geq \Psi_i^{wp}(1,\theta_i) = 0,
\end{equation}
which holds for any $\theta_i\in\Theta_i$.

\section{Proof of Theorem \ref{ThmCompl1-type2} \label{APPEEE}}

\paragraph{Feasibility.} Under the mechanism $(\bm{q}^{wp^*}, \bm{c}^{wp^*}, \bm{\pi}^{wp^*})$, using the bidder's expected payoff in \eqref{Uireporting}, we derive bidder $i$'s expected payoff from truth-telling as follows:
\begin{equation}\label{sdssdsdsssssssss}
U_i^{wp}(\theta_i) = \int_{\bm{\Theta}_{-i}}\int_{\underline{\theta}_i}^{\theta_i} \left\{
\begin{gathered}
q_i^{wp^*}(\tau,\bm{\theta}_{-i}) \cdot (1-\alpha_i^{wp-\max}(\tau))^2 \\
\cdot \left(\tau + \frac{(1-\alpha_i^{wp-\max}(\tau)) \cdot \alpha_i^{wp-\max}(\tau)\cdot \tau}{1 - (1-\alpha_i^{wp-\max}(\tau)) \cdot \alpha_i^{wp-\max}(\tau)}\right)
\end{gathered}
\right\}d\tau\bm{f}_{-i}(\bm{\theta}_{-i}) d\bm{\theta}_{-i}.
\end{equation}

Note that $U_i^{wp}(\theta_i)$ is increasing in $\theta_i$ and satisfies individual rationality: $U_i^{wp}(\theta_i) \geq U_i^{wp}(\underline{\theta}_i) = 0$. To establish incentive compatibility, according to Definition \ref{d1}, it suffices to establish the following inequality based on \eqref{sdssdsdsssssssss}:
\begin{equation}\label{GIC-2-type2}
\begin{split}
& U_i^{wp}(\theta_i) - U_{i}^{wp}(\hat{\theta}_{i},\theta_{i})\\
= & U_i^{wp}(\theta_i) - U_i^{wp}(\hat{\theta}_i) + U_i^{wp}(\hat{\theta}_i) - U_{i}^{wp}(\hat{\theta}_{i},\theta_{i})\\
= & \int_{\bm{\Theta}_{-i}} \underbrace{\left\{
\begin{gathered}
\int_{\hat{\theta}_i}^{\theta_i} \left\{
\begin{gathered}
q_i^{wp^*}(\tau,\bm{\theta}_{-i}) \cdot (1-\alpha_i^{wp-\max}(\tau))^2\\
\cdot \left(\tau + \frac{(1-\alpha_i^{wp-\max}(\tau)) \cdot \alpha_i^{wp-\max}(\tau)\cdot \tau}{1 - (1-\alpha_i^{wp-\max}(\tau)) \cdot \alpha_i^{wp-\max}(\tau)}\right)
\end{gathered}
\right\} d\tau \\
- q_i^{wp^*}(\hat{\theta}_i,\bm{\theta}_{-i}) \cdot \frac{(1-\alpha_i^{wp-\max}(\hat{\theta}_i))^2}{2} \\
\cdot \left(\theta_i + \hat{\theta}_i + 2\cdot\frac{(1-\alpha_i^{wp-\max}(\hat{\theta}_i)) \cdot \alpha_i^{wp-\max}(\hat{\theta}_i)\cdot \hat{\theta}_i}{1 - (1-\alpha_i^{wp-\max}(\hat{\theta}_i)) \cdot \alpha_i^{wp-\max}(\hat{\theta}_i)}\right) \cdot \left(\theta_i - \hat{\theta}_i\right)
\end{gathered}
\right\}}_{\triangleq \mathcal{I}(\theta_i)} \bm{f}_{-i}(\bm{\theta}_{-i}) d\bm{\theta}_{-i} \geq  0,
\end{split}
\end{equation}
for any $\theta_i, \hat{\theta}_i\in\Theta_i$ (without loss of generality, assume $\theta_i > \hat{\theta}_i$). We first claim that:

\begin{claim}\label{CL-derivative}
According to Lemma \ref{Lm-characterize-alphamax}, $\alpha_i^{wp-\max}(\theta_i)$ is such that:
\begin{equation*}
\frac{d}{d \theta_i}\left((1-\alpha_i^{wp-\max}(\theta_i))^2 \cdot \frac{(1-\alpha_i^{wp-\max}(\theta_i)) \cdot \alpha_i^{wp-\max}(\theta_i)}{1 - (1-\alpha_i^{wp-\max}(\theta_i)) \cdot \alpha_i^{wp-\max}(\theta_i)}\cdot \theta_i\right) \geq 0.
\end{equation*}
\end{claim}
\begin{proof}[Proof of Claim \ref{CL-derivative}]
Taking the logarithm and differentiating the objective with respect to $\theta_i$:
\begin{equation}\label{CL-derivative-00}
\begin{split}
& \left(
\begin{gathered}
-\frac{3}{1 - \alpha_i^{wp-\max}(\theta_i)} \\
+ \frac{1}{\alpha_i^{wp-\max}(\theta_i)}\\
- \frac{-1 + 2\alpha_i^{wp-\max}(\theta_i)}{1 - \alpha_i^{wp-\max}(\theta_i) + (\alpha_i^{wp-\max}(\theta_i))^2}
\end{gathered}
\right) \cdot \frac{d \alpha_i^{wp-\max}(\theta_i)}{d\theta_i} + \frac{1}{\theta_i} \\
= & \frac{1}{\theta_i} + \frac{d \alpha_i^{wp-\max}(\theta_i)}{d\theta_i}\\
& \cdot \frac{1 - 4\alpha_i^{wp-\max}(\theta_i) + 2(\alpha_i^{wp-\max}(\theta_i))^2 - 2(\alpha_i^{wp-\max}(\theta_i))^3}{(1 - \alpha_i^{wp-\max}(\theta_i))\cdot \alpha_i^{wp-\max}(\theta_i) \cdot (1 - \alpha_i^{wp-\max}(\theta_i) + (\alpha_i^{wp-\max}(\theta_i))^2)},
\end{split}
\end{equation}
where the numerator equals
\begin{equation*}
\begin{split}
& 1 - 4\alpha_i^{wp-\max}(\theta_i) + 2(\alpha_i^{wp-\max}(\theta_i))^2 - 2(\alpha_i^{wp-\max}(\theta_i))^3 \\
= & 2(1 - \alpha_i^{wp-\max}(\theta_i))^3 - 1 + 2\alpha_i^{wp-\max}(\theta_i) - 4(\alpha_i^{wp-\max}(\theta_i))^2\\
= & 2 \left(1 - 4\alpha_i^{wp-\max}(\theta_i) + 3(\alpha_i^{wp-\max}(\theta_i))^2 - (\alpha_i^{wp-\max}(\theta_i))^3 + \alpha_i^{wp-\max}(\theta_i)\right)\\
& - 1 + 2\alpha_i^{wp-\max}(\theta_i) - 4(\alpha_i^{wp-\max}(\theta_i))^2\\
\leq & 2\alpha_i^{wp-\max}(\theta_i) - 1 + 2\alpha_i^{wp-\max}(\theta_i) - 4(\alpha_i^{wp-\max}(\theta_i))^2\\
= & -(1-2\alpha_i^{wp-\max}(\theta_i))^2\\
\leq & 0,
\end{split}
\end{equation*}
where the first inequality follows from \eqref{CL-decrease-OptAlpha-01}.

Therefore, combined with Lemma \ref{Lm-characterize-alphamax}, the derivative in \eqref{CL-derivative-00} is non-negative.
\end{proof}

To establish that inequality \eqref{GIC-2-type2} holds for any $\theta_i, \hat{\theta}_i \in \Theta_i$, we show that $\mathcal{I}(\theta_i) \geq 0$:
\begin{equation*}
\begin{split}
\mathcal{I}(\theta_i) \geq & q_i^{wp^*}(\hat{\theta}_i,\bm{\theta}_{-i}) \cdot \left(
\begin{gathered}
(1-\alpha_i^{wp-\max}(\hat{\theta}_i))^2 \cdot \int_{\hat{\theta}_i}^{\theta_i} \tau d\tau \\
+ \int_{\hat{\theta}_i}^{\theta_i} \left\{
\begin{gathered}
(1-\alpha_i^{wp-\max}(\tau))^2 \\
\cdot \frac{(1-\alpha_i^{wp-\max}(\tau)) \cdot \alpha_i^{wp-\max}(\tau)}{1 - (1-\alpha_i^{wp-\max}(\tau)) \cdot \alpha_i^{wp-\max}(\tau)}\cdot \tau
\end{gathered}\right\} d\tau
\end{gathered}\right)\\
& - q_i^{wp^*}(\hat{\theta}_i,\bm{\theta}_{-i}) \cdot \left(
\begin{gathered}
(1-\alpha_i^{wp-\max}(\hat{\theta}_i))^2 \cdot \frac{\theta_i^2 - \hat{\theta}_i^2}{2}\\
+ \left(
\begin{gathered}
(1-\alpha_i^{wp-\max}(\hat{\theta}_i))^2 \\
\cdot \frac{(1-\alpha_i^{wp-\max}(\hat{\theta}_i)) \cdot \alpha_i^{wp-\max}(\hat{\theta}_i)\cdot \hat{\theta}_i}{1 - (1-\alpha_i^{wp-\max}(\hat{\theta}_i)) \cdot \alpha_i^{wp-\max}(\hat{\theta}_i)}\cdot \left(\theta_i - \hat{\theta}_i\right)
\end{gathered}\right)
\end{gathered}\right)\\
> & q_i^{wp^*}(\hat{\theta}_i,\bm{\theta}_{-i}) \cdot \left(
\begin{gathered}
(1-\alpha_i^{wp-\max}(\hat{\theta}_i))^2 \cdot \frac{\theta_i^2 - \hat{\theta}_i^2}{2}\\
+ \left(
\begin{gathered}
(1-\alpha_i^{wp-\max}(\hat{\theta}_i))^2 \\
\cdot \frac{(1-\alpha_i^{wp-\max}(\hat{\theta}_i)) \cdot \alpha_i^{wp-\max}(\hat{\theta}_i)}{1 - (1-\alpha_i^{wp-\max}(\hat{\theta}_i)) \cdot \alpha_i^{wp-\max}(\hat{\theta}_i)}\cdot \hat{\theta}_i \cdot \int_{\hat{\theta}_i}^{\theta_i}  d\tau
\end{gathered}\right)
\end{gathered}\right)\\
& - q_i^{wp^*}(\hat{\theta}_i,\bm{\theta}_{-i}) \cdot \left(
\begin{gathered}
(1-\alpha_i^{wp-\max}(\hat{\theta}_i))^2 \cdot \frac{\theta_i^2 - \hat{\theta}_i^2}{2}\\
+ \left(
\begin{gathered}
(1-\alpha_i^{wp-\max}(\hat{\theta}_i))^2 \\
\cdot \frac{(1-\alpha_i^{wp-\max}(\hat{\theta}_i)) \cdot \alpha_i^{wp-\max}(\hat{\theta}_i)\cdot \hat{\theta}_i}{1 - (1-\alpha_i^{wp-\max}(\hat{\theta}_i)) \cdot \alpha_i^{wp-\max}(\hat{\theta}_i)}\cdot \left(\theta_i - \hat{\theta}_i\right)
\end{gathered}\right)
\end{gathered}\right)\\
= & 0,
\end{split}
\end{equation*}
where the first inequality follows from: (1) $q_i^{wp^*}(\theta_i, \bm{\theta}_{-i})$ is increasing in $\theta_i$ since the virtual surplus is increasing in $\theta_i$ by Lemma \ref{increasing-of-surplus}; and (2) $(1-\alpha_i^{wp-\max}(\theta_i))^2$ is increasing in $\theta_i$ since $\alpha_i^{wp-\max}(\cdot)$ is an interior maximizer and $\alpha_i^{wp-\max}(\cdot)$ is decreasing by Lemma \ref{Lm-characterize-alphamax}. The second inequality follows from Claim \ref{CL-derivative}.

\paragraph{Optimality.} We establish optimality by showing that this mechanism attains the upper bound $\overline{\text{Rev}}^{wp}(\bm{q}, \bm{c}, \bm{\pi})$ in Lemma \ref{uppbdd}. Under the feasible linear mechanism $(\bm{q}^{wp^*}, \bm{c}^{wp^*}, \bm{\pi}^{wp^*})$, according to \eqref{REV-000} and the expression for $\phi_i^{wp}(\cdot)$ in \eqref{phiphiphi-type2}, we obtain:
\begin{equation*}
\begin{split}
& \text{Rev}^{wp}(\bm{q}^{wp^*}, \bm{c}^{wp^*}, \bm{\pi}^{wp^*}) \\
= & \sum_{i\in\mathcal{N}} \int_{\bm{\Theta}} q_i^{wp^*}(\bm{\theta}) \cdot \left\{
\begin{gathered}
\frac{1}{2}\cdot (1-(\alpha^{wp-\max}_i(\theta_i))^2)\cdot \theta_i^2\\
+ (1-(\alpha^{wp-\max}_i(\theta_i))^2) \cdot \theta_i \cdot \frac{(1-\alpha^{wp-\max}_i(\theta_i)) \cdot \alpha^{wp-\max}_i(\theta_i)\cdot \theta_i}{1 - (1-\alpha^{wp-\max}_i(\theta_i)) \cdot \alpha^{wp-\max}_i(\theta_i)}\\
- \frac{(\alpha^{wp-\max}_i(\theta_i))^2}{2}\cdot\left(\frac{(1-\alpha^{wp-\max}_i(\theta_i)) \cdot \alpha^{wp-\max}_i(\theta_i)\cdot \theta_i}{1 - (1-\alpha^{wp-\max}_i(\theta_i)) \cdot \alpha^{wp-\max}_i(\theta_i)}\right)^2 \\
- (1-\alpha^{wp-\max}_i(\theta_i))^2 \cdot \theta_i \cdot \frac{1 - F_i(\theta_i)}{f_i(\theta_i)} \\
- (1-\alpha^{wp-\max}_i(\theta_i))^2 \cdot \frac{(1-\alpha^{wp-\max}_i(\theta_i)) \cdot \alpha^{wp-\max}_i(\theta_i)\cdot \theta_i}{1 - (1-\alpha^{wp-\max}_i(\theta_i)) \cdot \alpha^{wp-\max}_i(\theta_i)} \\
\cdot \frac{1 - F_i(\theta_i)}{f_i(\theta_i)}
\end{gathered}
\right\} \bm{f}(\bm{\theta}) d\bm{\theta}\\
= & \sum_{i\in\mathcal{N}} \int_{\bm{\Theta}} q_i^{wp^*}(\bm{\theta}) \cdot \phi_i^{wp}(\theta_i) \bm{f}(\bm{\theta}) d\bm{\theta}\\
= & \int_{\bm{\Theta}} \max\left\{0, \phi_1^{wp}(\theta_1), \cdots, \phi_i^{wp}(\theta_i),\cdots,\phi_n^{wp}(\theta_n)\right\} \bm{f}(\bm{\theta}) d\bm{\theta}\\
= & \overline{\text{Rev}}^{wp}(\bm{q}, \bm{c}, \bm{\pi}).
\end{split}
\end{equation*}

\section{Proof of Lemma \ref{S-lead-LM} \label{PF-S-lead-LM}}

Following the derivation of the seller's expected payoff under winner-pivotal collaboration (see Appendix \ref{9pf}, which uses the general envelope theorem), we apply the same change-of-measure approach as in \eqref{change-meas} to derive:
\begin{equation}\label{Sllll-REV0}
\begin{split}
& \text{Rev}^{sp}(\bm{q}, \bm{c}, \bm{\pi})\\
= & \sum_{i\in\mathcal{N}} \int_{[0,1]\times \mathcal{T}_i^w \times \bm{\mathcal{T}}_{-i}^l \times S_i} \left\{ 
\begin{gathered}
\frac{2\alpha_i \cdot (1 - (1 - \alpha_i)\cdot \alpha_i) - \alpha_i^4}{2(1 - (1-\alpha_i) \cdot \alpha_i)^2}\\ \cdot  \mathbb{E}^2\left[\left.\widetilde{\theta}_i\right|\mu(\cdot |i, \bm{c}, s)\right] \\
- \frac{(1-\alpha_i) \cdot\alpha_i}{1 - (1-\alpha_i) \cdot \alpha_i} \\
\cdot  \mathbb{E}\left[\left.\widetilde{\theta}_i\right|\mu(\cdot |i, \bm{c}, s)\right]\cdot \mathbb{E}\left[\left.\frac{1 - F_i(\widetilde{\theta}_i)}{f_i(\widetilde{\theta}_i)}\right|\mu(\cdot |i, \bm{c}, s)\right]
\end{gathered}
\right\} \xi(i, d\bm{c}, ds) \\
& - \sum_{i\in\mathcal{N}}U_i^{sp}(\underline{\theta}_i).
\end{split} 
\end{equation}

Following the same approach as in \eqref{winnerlead-inequa}, we obtain:
\begin{equation*}
\text{Rev}^{sp}(\bm{q}, \bm{c}, \bm{\pi})\leq \sum_{i\in\mathcal{N}}\int_{\bm{\Theta}} q_i(\bm{\theta}) \int_{0}^1 \left\{
\begin{gathered}
\underbrace{\frac{2\alpha_i \cdot (1 - (1 - \alpha_i)\cdot \alpha_i) - \alpha_i^4}{2(1 - (1-\alpha_i) \cdot \alpha_i)^2}}_{\triangleq A^{sp}(\alpha_i)} \cdot \theta_i^2 \\
- \underbrace{\frac{(1-\alpha_i) \cdot\alpha_i}{1 - (1-\alpha_i) \cdot \alpha_i}}_{\triangleq B^{sp}(\alpha_i)} \cdot \frac{1-F_i(\theta_i)}{f_i(\theta_i)}\cdot \theta_i
\end{gathered}
\right\} \sigma_i(d\alpha_i|\bm{\theta}) \bm{f}(\bm{\theta}) d\bm{\theta}.
\end{equation*}

The maximizer is:
\begin{equation}\label{ABspsp}
\sigma_i^{sp-\max}(\cdot|\bm{\theta}) \in \argmax_{\sigma_i(d\alpha_i|\bm{\theta})\in\Delta([0,1])} \int_{0}^1 \underbrace{\left\{ A^{sp}(\alpha_i) \cdot \theta_i^2  - B^{sp}(\alpha_i) \cdot \frac{1-F_i(\theta_i)}{f_i(\theta_i)}\cdot \theta_i \right\}}_{\triangleq \Psi_i^{sp}(\alpha_i,\theta_i)}\sigma_i(d\alpha_i|\bm{\theta}),
\end{equation}
which is a Dirac measure that places unit mass on $\alpha_i^{sp-\max}(\theta_i) \in \argmax_{\alpha_i\in[0,1]} \Psi_i^{sp}(\alpha_i,\theta_i)$.

Accordingly, we define the corresponding maximal value as bidder $i$'s virtual surplus:
\begin{equation}\label{phiphiphi-sp-pf}
\phi_i^{sp}(\theta_i) \triangleq \int_{0}^1 \Psi_i^{sp}(\alpha_i,\theta_i)\sigma_i^{sp-\max}(d\alpha_i|\bm{\theta}) = \Psi_i^{sp}(\alpha_i^{sp-\max},\theta_i),
\end{equation}
which matches the expression in \eqref{Lmm1-Psi-1}. Therefore, for any feasible direct linear mechanism $(\bm{q}, \bm{c}, \bm{\pi})$, the seller's expected payoff under seller-pivotal collaboration is bounded above by:
\begin{equation*}
\begin{split}
\overline{\text{Rev}}^{sp}(\bm{q}, \bm{c}, \bm{\pi}) \triangleq & \max_{\substack{0\leq q_i(\bm{\theta}) \leq 1;\\ \sum_{i\in\mathcal{N}} q_i(\bm{\theta}) \leq 1}} \sum_{i\in\mathcal{N}}\int_{\bm{\Theta}} q_i(\bm{\theta}) \cdot \phi^{sp}(\theta_i) \bm{f}(\bm{\theta}) d\bm{\theta}\\
= & \int_{\bm{\Theta}} \max\left\{\phi_1^{sp}(\theta_1), \cdots, \phi_i^{sp}(\theta_i),\cdots,\phi_n^{sp}(\theta_n)\right\} \bm{f}(\bm{\theta}) d\bm{\theta}.
\end{split}
\end{equation*}

\section{Proof of Lemma \ref{lm-sellerlead-alphamax} \label{lm-sellerlead-alphamax-PF}}

\paragraph{Equivalence.} We first show the equivalence between the following two constrained optimization problems:
\begin{equation*}
\max_{\alpha_i\in[\frac{1}{2},1]} \Psi_i^{sp}(\alpha_i,\theta_i) = \max_{\alpha_i\in[0,1]} \Psi_i^{sp}(\alpha_i,\theta_i),
\end{equation*}
which differ only in their constraint sets. The function $\Psi_i^{sp}(\alpha_i, \theta_i)$ is defined by \eqref{ABspsp}; it suffices to show $\alpha_i^{sp-\max}(\theta_i) \notin [0, \frac{1}{2})$. We proceed by contradiction. Suppose $\alpha_i^{sp-\max} \in [0, \frac{1}{2})$. Then there exists $\alpha_i' \in (\frac{1}{2}, 1]$ such that $\alpha_i^{sp-\max} + \alpha_i' = 1$. Thus, $(1-\alpha_i') \cdot \alpha_i' = (1-\alpha_i^{sp-\max}) \cdot \alpha_i^{sp-\max}$ and the second term of the objective function $\Psi_i(\cdot,\theta_i)$ satisfies:
\begin{equation*}
\frac{(1-\alpha_i') \cdot\alpha_i'}{1 - (1-\alpha_i') \cdot \alpha_i'} \cdot \frac{1-F_i(\theta_i)}{f_i(\theta_i)}\cdot \theta_i = \frac{(1-\alpha_i^{sp-\max}) \cdot\alpha_i^{sp-\max}}{1 - (1-\alpha_i^{sp-\max}) \cdot \alpha_i^{sp-\max}} \cdot \frac{1-F_i(\theta_i)}{f_i(\theta_i)}\cdot \theta_i.
\end{equation*}

For the first term, we denote $(1-\alpha_i') \cdot \alpha_i' = (1-\alpha_i^{sp-\max}) \cdot \alpha_i^{sp-\max} \triangleq K \in [0, \frac{1}{4})$ and show that the difference between the numerators satisfies:
\begin{equation*}
\begin{split}
& 2\alpha_i' \cdot (1 - K) - (\alpha_i')^4 - \left(2\alpha_i^{sp-\max} \cdot (1 - K) - (\alpha_i^{sp-\max})^4\right) \\
= & 2 (1-K) \cdot \left(\alpha_i' - \alpha_i^{sp-\max}\right) - \left((\alpha_i')^2 - (\alpha_i^{sp-\max})^2\right) \cdot \left((\alpha_i')^2 + (\alpha_i^{sp-\max})^2\right) \\
> & 2 (1-K) \cdot \left(\alpha_i' - \alpha_i^{sp-\max}\right) - \left((\alpha_i')^2 - (\alpha_i^{sp-\max})^2\right) \cdot \left(\alpha_i' + \alpha_i^{sp-\max}\right)^2 \\
= & \left(\alpha_i' - \alpha_i^{sp-\max}\right) \cdot \left(1-2K\right) > 0.
\end{split}
\end{equation*}

Therefore, $\alpha_i' \in (\frac{1}{2}, 1]$ yields a strictly higher objective value, contradicting the assumption that $\alpha_i^{sp-\max} \in [0, \frac{1}{2})$. Hence, we establish $\alpha_i^{sp-\max}(\theta_i) \geq \frac{1}{2}$ for all $\theta_i \in [\underline{\theta}_i, \overline{\theta}_i]$.

\paragraph{Non-increasing.} We analyze the monotonicity of $\alpha_i^{sp-\max}(\theta_i)$ separately for the cases of corner and interior solutions.

(1) Corner solution. We consider the equivalent optimization problem $\max_{\alpha_i \in [1/2,1]} \Psi_i^{sp}(\alpha_i,\theta_i)$:
\begin{equation*}
\max_{\alpha_i} \left\{ A^{sp}(\alpha_i) \cdot \theta_i^2  - B^{sp}(\alpha_i) \cdot \frac{1-F_i(\theta_i)}{f_i(\theta_i)}\cdot \theta_i \right\} \;\; \text{s.t. } \;\;  1 - \alpha_i \geq 0\;\;\text{ and } \;\; \alpha_i \geq \frac{1}{2}.
\end{equation*}

The Kuhn-Tucker conditions for the maximizer $\alpha_i^{sp-\max}(\theta_i)$ are:
\begin{equation}\label{KKT-sellerlead}
\begin{cases}
- \frac{(\alpha_i^{sp-\max})^3 + \alpha_i^{sp-\max} - 1}{(1 - \alpha_i^{sp-\max} + (\alpha_i^{sp-\max})^2)^3} \cdot \theta_i^2 + \frac{2\alpha_i^{sp-\max} - 1}{(1 - \alpha_i^{sp-\max} + (\alpha_i^{sp-\max})^2)^2} \cdot \frac{1 - F_i(\theta_i)}{f_i(\theta_i)} \cdot \theta_i - \lambda_1 + \lambda_2 = 0;\\
\lambda_1 \cdot (1 - \alpha_i^{sp-\max}) = 0,\;\; \lambda_2 \cdot (\alpha_i^{sp-\max} - \frac{1}{2}) = 0;\\
1 - \alpha_i^{sp-\max} \geq 0, \;\; \alpha_i^{sp-\max}- \frac{1}{2} \geq 0;\\
\lambda_1 \geq 0, \;\; \lambda_2 \geq 0.
\end{cases}
\end{equation}

Based on the slackness conditions: (1) if $\lambda_1 > 0$ and $\lambda_2 = 0$, then $\alpha_i^{sp-\max} = 1$, which implies $-\theta_i^2 + \frac{1 - F_i(\theta_i)}{f_i(\theta_i)} \cdot \theta_i - \lambda_1 = 0$, yielding $\frac{1 - F_i(\theta_i)}{\theta_i\cdot f_i(\theta_i)} \geq 1$; (2) if $\lambda_1 = 0$ and $\lambda_2 > 0$, then $\alpha_i^{sp-\max} = \frac{1}{2}$, which implies $\frac{8}{9}\theta_i^2 + \lambda_2 = 0$. This contradicts $\lambda_2 > 0$.

Therefore, for any $\theta_i \in [\underline{\theta}_i, \overline{\theta}_i]$ satisfying $\frac{1 - F_i(\theta_i)}{\theta_i \cdot f_i(\theta_i)} \geq 1$, the maximizer $\alpha_i^{sp-\max}(\theta_i) = 1$ constitutes a corner solution. Moreover, given that
\begin{equation*}
\frac{d}{d\theta_i} \left(\frac{1 - F_i(\theta_i)}{\theta_i \cdot f_i(\theta_i)}\right) = \frac{d}{d\theta_i} \left( \frac{1 - F_i(\theta_i)}{f_i(\theta_i)} \right) \cdot \frac{1}{\theta_i} - \left( \frac{1 - F_i(\theta_i)}{f_i(\theta_i)} \right) \cdot \frac{1}{\theta_i^2} < 0,
\end{equation*}
where the inequality follows from the fact that $\frac{1-F_i(\theta_i)}{f_i(\theta_i)}$ decreases in $\theta_i$ under log-concavity of $f_i(\cdot)$, we establish that $\frac{1 - F_i(\theta_i)}{\theta_i \cdot f_i(\theta_i)}$ is decreasing in $\theta_i$. Consequently, there exists a unique threshold $\theta_i^c$ such that:
\begin{equation}\label{ttthre}
\frac{1 - F_i(\theta_i^c)}{\theta_i^c \cdot f_i(\theta_i^c)} = 1.
\end{equation}
Thus, $\alpha_i^{sp-\max}(\theta_i) = 1$ for all $\theta_i \leq \theta_i^c$.

(2) Interior maximizer and decreasing nature. When the maximizer is an interior solution, i.e., $\alpha_i^{sp-\max} \in (\frac{1}{2},1)$, the stationary condition in \eqref{KKT-sellerlead} with $\lambda_1 = \lambda_2 = 0$ gives:
\begin{equation}\label{type2intFOC}
\frac{(\alpha_i^{sp-\max})^3 + \alpha_i^{sp-\max} - 1}{(1 - \alpha_i^{sp-\max} + (\alpha_i^{sp-\max})^2)\cdot (2\alpha_i^{sp-\max} - 1)} =  \frac{1 - F_i(\theta_i)}{\theta_i \cdot f_i(\theta_i)}.
\end{equation}

First, since $2\alpha_i^{sp-\max} - 1 > 0$ and $\frac{1 - F_i(\theta_i)}{\theta_i \cdot f_i(\theta_i)} \geq 0$, we have $(\alpha_i^{sp-\max})^3 + \alpha_i^{sp-\max} - 1 \geq 0 > (\frac{2}{3})^3 + \frac{2}{3} -1$, which implies:
\begin{equation}\label{larger23-lm4pf}
\alpha_i^{sp-\max} > \frac{2}{3}.
\end{equation}

Next, following the approach in \eqref{take-log-interiormax}, taking the logarithm and differentiating both sides with respect to $\theta_i$ yields:
\begin{equation*}
\left(
\begin{gathered}
\frac{3(\alpha_i^{sp-\max})^2 + 1}{(\alpha_i^{sp-\max})^3 + \alpha_i^{sp-\max} - 1} \\
- \frac{2 \alpha_i^{sp-\max} - 1}{(\alpha_i^{sp-\max})^2 - \alpha_i^{sp-\max} + 1}\\
- \frac{2}{2 \alpha_i^{sp-\max} - 1}
\end{gathered}
\right) \cdot \frac{d \alpha_i^{sp-\max}(\theta_i)}{d\theta_i} = \frac{d}{d\theta_i} \left( \frac{1 - F_i(\theta_i)}{f_i(\theta_i)} \right) \cdot \frac{f_i(\theta_i)}{1 - F_i(\theta_i)} - \frac{1}{\theta_i} < 0.
\end{equation*}
The negativity of the right-hand side follows from the same reasoning as in \eqref{CL-decrease-OptAlpha-01-01}. The terms in the bracket on the left-hand side derive as follows:
\begin{equation*}
\left(
\begin{gathered}
\underbrace{\frac{3(\alpha_i^{sp-\max})^2 - 2\alpha_i^{sp-\max}}{(\alpha_i^{sp-\max})^3 + \alpha_i^{sp-\max} - 1}}_{\text{term 1}}\\
+ \underbrace{\frac{2\alpha_i^{sp-\max} - 1}{(\alpha_i^{sp-\max})^3 + \alpha_i^{sp-\max} - 1} - \frac{2\alpha_i^{sp-\max} - 1}{(\alpha_i^{sp-\max})^2 - \alpha_i^{sp-\max} + 1}}_{\text{term 2}}\\
+ \underbrace{\frac{2}{(\alpha_i^{sp-\max})^3 + \alpha_i^{sp-\max} - 1} - \frac{2}{2\alpha_i^{sp-\max} - 1}}_{\text{term 3}}
\end{gathered}
\right) > 0,
\end{equation*}
where term 1 is positive by \eqref{larger23-lm4pf}. Term 2 is non-negative since the denominators satisfy $((\alpha_i^{sp-\max})^3 + \alpha_i^{sp-\max} - 1) - ((\alpha_i^{sp-\max})^2 - \alpha_i^{sp-\max} + 1) = (\alpha_i^{sp-\max} - 1)\cdot ((\alpha_i^{sp-\max})^2 + 2) \leq 0$. Term 3 is non-negative since $(\alpha_i^{sp-\max})^3 + \alpha_i^{sp-\max} - 1 \leq 2\alpha_i^{sp-\max} - 1$, as $(\alpha_i^{sp-\max})^3 \leq \alpha_i^{sp-\max} \leq 1$. Therefore, $\frac{d \alpha_i^{sp-\max}(\theta_i)}{d\theta_i } < 0$ when the maximizer is an interior solution.

\paragraph{Uniqueness.} First, for any $\theta_i \leq \theta_i^c$, the maximizer of \eqref{alphastarSL} is a corner solution, $\alpha_i^{sp-\max}(\theta_i) = 1$, where the threshold type is uniquely determined by \eqref{ttthre}. Second, for any $\theta_i \geq \theta_i^c$, we show that the interior maximizer is unique. From the objective function in \eqref{ABspsp}, we derive the second-order partial derivative with respect to $\alpha_i$:
\begin{equation}\label{vnvbvbvb}
\frac{\partial^2 \Psi_i^{sp}(\alpha_i, \theta_i)}{\partial \alpha_i^2} = \frac{\theta_i^2}{(1 - \alpha_i + \alpha_i^2)^4} \cdot
\underbrace{\left(
\begin{gathered}
(3\alpha_i^4 + 2\alpha_i^2 - 8\alpha_i + 2)\\
+ 6\alpha_i\cdot(1 - \alpha_i)\cdot(1 - \alpha_i + \alpha_i^2) \cdot \frac{1 - F_i(\theta_i)}{\theta_i \cdot f_i(\theta_i)}
\end{gathered}
\right)}_{\triangleq \Omega_i(\alpha_i,\theta_i)}.
\end{equation}

Since $\frac{1 - F_i(\theta_i)}{\theta_i \cdot f_i(\theta_i)} \leq 1$ for $\theta_i \geq \theta_i^c$, we obtain:
\begin{equation*}
\begin{split}
\frac{\partial^2 \Omega_i(\alpha_i,\theta_i)}{\partial \alpha_i^2} = &(36\alpha_i^2 + 4) - (72\alpha_i^2 - 72\alpha_i + 24) \cdot \frac{1 - F_i(\theta_i)}{\theta_i \cdot f_i(\theta_i)}\\
> & (36\alpha_i^2 + 4) - (72\alpha_i^2 - 72\alpha_i + 24) = 16 - 36(1-\alpha_i)^2 > 0.
\end{split}
\end{equation*}

Thus, $\Omega_i(\alpha_i,\theta_i)$ is convex in $\alpha_i$ on $[\frac{1}{2}, 1)$. Moreover, note that $\Omega_i(\frac{1}{2},\theta_i) = -\frac{21}{16} + \frac{9}{8}\frac{1 - F_i(\theta_i)}{\theta_i \cdot f_i(\theta_i)} < 0$ and $\Omega_i(1,\theta_i) = -1 <0$. Therefore, $\Omega_i(\alpha_i,\theta_i) < 0$ for all $\alpha_i \in [\frac{1}{2}, 1)$, which implies strict concavity of the objective function, i.e., $\frac{\partial^2 \Psi_i^{sp}(\alpha_i, \theta_i)}{\partial \alpha_i^2} < 0$ by \eqref{vnvbvbvb}. Consequently, for any $\theta_i \geq \theta_i^c$, the interior maximizer $\alpha_i^{sp-\max} \in [\frac{1}{2}, 1)$ is unique.

\section{Proof of Lemma \ref{increasing-of-surplus-SP} \label{increasing-of-surplus-SP-pf}}

When $\alpha_i^{sp-\max}(\theta_i)$ is an interior maximizer, according to \eqref{type2intFOC} and \eqref{larger23-lm4pf}, applying the envelope theorem and differentiating \eqref{phiphiphi-sp-pf} with respect to $\theta_i$ yields:
\begin{equation}\label{H4}
\begin{split}
\frac{d \phi_i^{sp}(\theta_i)}{d \theta_i} = & \frac{(\alpha_i^{sp-\max})^3 - \alpha_i^{sp-\max}\cdot(1-\alpha_i^{sp-\max})^4}{(1 - \alpha_i^{sp-\max} + (\alpha_i^{sp-\max})^2)\cdot (2\alpha_i^{sp-\max} - 1)} \cdot \theta_i\\
& - \frac{(1-\alpha_i^{sp-\max}) \cdot\alpha_i^{sp-\max}}{1 - (1-\alpha_i^{sp-\max}) \cdot \alpha_i^{sp-\max}} \cdot \frac{d}{d \theta_i} \left(\frac{1-F_i(\theta_i)}{f_i(\theta_i)}\right) \cdot \theta_i \\
> & \frac{(\alpha_i^{sp-\max})^3 + \alpha_i^{sp-\max} - 1}{(1 - \alpha_i^{sp-\max} + (\alpha_i^{sp-\max})^2)\cdot (2\alpha_i^{sp-\max} - 1)} \cdot \theta_i\\
& - \frac{(1-\alpha_i^{sp-\max}) \cdot\alpha_i^{sp-\max}}{1 - (1-\alpha_i^{sp-\max}) \cdot \alpha_i^{sp-\max}} \cdot \frac{d}{d \theta_i} \left(\frac{1-F_i(\theta_i)}{f_i(\theta_i)}\right) \cdot \theta_i \geq 0.
\end{split}
\end{equation}
This establishes that each bidder's virtual surplus is increasing in $\theta_i$. It remains to show that $\phi_i^{sp}(\theta_i) \geq 0$ for each $i\in\mathcal{N}$ and any $\theta_i\in\Theta_i$. Similar to \eqref{non-negateive-phiWL}, by the definition of virtual surplus in \eqref{phiphiphi-sp-pf}, the result follows from:
\begin{equation*}
\phi_i^{sp}(\theta_i) = \max_{\alpha_i \in [0,1]} \Psi_i^{sp}(\alpha_i,\theta_i) \geq \Psi_i^{sp}(1,\theta_i) = \frac{1}{2}\cdot \theta_i^2 \geq \frac{1}{2}\cdot \underline{\theta}_i^2 \geq 0,
\end{equation*}
which holds for any $\theta_i\in\Theta_i$, where the final inequality follows from $\underline{\theta}_i \geq 0$.

\section{Proof of Theorem \ref{ThmCompl1-typeSL-seller} \label{THMoptSL}}

Following the proof of Theorem \ref{ThmCompl1-type2} (see Appendix \ref{APPEEE}), we first verify the feasibility of the mechanism $(\bm{q}^{sp^*}, \bm{c}^{sp^*}, \bm{\pi}^{sp^*})$. From the expression for a bidder's expected payoff in \eqref{Uireporting-2}, we derive bidder $i$'s expected payoff under truthful reporting as follows:
\begin{equation*}
U_i^{sp}(\theta_i) = \int_{\bm{\Theta}_{-i}}\int_{\underline{\theta}_i}^{\theta_i} q_i^{sp^*}(\tau,\bm{\theta}_{-i}) \cdot \frac{(1-\alpha_i^{sp-\max}(\tau)) \cdot \alpha_i^{sp-\max}(\tau)}{1 - (1-\alpha_i^{sp-\max}(\tau)) \cdot \alpha_i^{sp-\max}(\tau)}\cdot \tau  d\tau\bm{f}_{-i}(\bm{\theta}_{-i}) d\bm{\theta}_{-i}.
\end{equation*}

To establish the bidder's truth-telling incentive, we show that, analogous to \eqref{GIC-2-type2}:
\begin{equation*}
\begin{split}
& U_i^{sp}(\theta_i) -  U_{i}^{sp}(\hat{\theta}_{i},\theta_{i})\\
= & \int_{\bm{\Theta}_{-i}}\int_{\hat{\theta}_i}^{\theta_i} q_i^{sp^*}(\tau,\bm{\theta}_{-i}) \cdot \frac{(1-\alpha_i^{sp-\max}(\tau)) \cdot \alpha_i^{sp-\max}(\tau)}{1 - (1-\alpha_i^{sp-\max}(\tau)) \cdot \alpha_i^{sp-\max}(\tau)}\cdot \tau d\tau\bm{f}_{-i}(\bm{\theta}_{-i}) d\bm{\theta}_{-i}\\
& - \int_{\bm{\Theta}_{-i}} q_i^{sp^*}(\hat{\theta}_i,\bm{\theta}_{-i}) \cdot \frac{(1-\alpha_i^{sp-\max}(\hat{\theta}_i)) \cdot \alpha_i^{sp-\max}(\hat{\theta}_i)}{1 - (1-\alpha_i^{sp-\max}(\hat{\theta}_i)) \cdot \alpha_i^{sp-\max}(\hat{\theta}_i)}\cdot \hat{\theta}_i \cdot \left(\theta_i - \hat{\theta}_i \right) \bm{f}_{-i}(\bm{\theta}_{-i}) d\bm{\theta}_{-i} \geq 0
\end{split}
\end{equation*}
holds for any $\theta_i, \hat{\theta}_i \in \Theta_i$ (without loss of generality, assume $\theta_i > \hat{\theta}_i$). This follows because $q_i^{sp^*}(\theta_i, \bm{\theta}_{-i})$ is increasing in $\theta_i$, since the virtual surplus is increasing in $\theta_i$ (see \eqref{H4}). Additionally,
\begin{equation*}
\begin{split}
& \frac{d}{d \theta_i}\left(\frac{(1-\alpha_i^{sp-\max}(\theta_i)) \cdot \alpha_i^{sp-\max}(\theta_i)}{1 - (1-\alpha_i^{sp-\max}(\theta_i)) \cdot \alpha_i^{sp-\max}(\theta_i)}\cdot\theta_i\right) \\
= & \frac{(\alpha_i^{sp-\max}(\theta_i))' \cdot (1-2\alpha_i^{sp-\max}(\theta_i))}{(1-(1-\alpha_i^{sp-\max}(\theta_i)) \cdot \alpha_i^{sp-\max}(\theta_i))^2}\cdot \theta_i + \frac{(1-\alpha_i^{sp-\max}(\theta_i)) \cdot \alpha_i^{sp-\max}(\theta_i)}{1-(1-\alpha_i^{sp-\max}(\theta_i)) \cdot \alpha_i^{sp-\max}(\theta_i)} \geq 0,
\end{split}
\end{equation*}
due to $\alpha_i^{sp-\max} \geq \frac{1}{2}$ and the fact that $\alpha_i^{sp-\max}(\cdot)$ is non-increasing (see Lemma \ref{lm-sellerlead-alphamax}).

Optimality follows from showing that $\text{Rev}^{sp}(\bm{q}^{sp^*}, \bm{c}^{sp^*}, \bm{\pi}^{sp^*}) = \overline{\text{Rev}}^{sp}(\bm{q}, \bm{c}, \bm{\pi})$, using the same approach as in Appendix \ref{APPEEE}.

\section{Proof of Proposition \ref{CompareAlpha} \label{Pf-CompareAlpha}}

First, according to Lemmas \ref{Lm-characterize-alphamax} and \ref{lm-sellerlead-alphamax}, both optimal deterministic value-sharing rules $\alpha_i^{wp^*}(\theta_i) = \alpha_i^{wp-\max}(\theta_i)$ and $\alpha_i^{sp^*}(\theta_i) =\allowbreak \alpha_i^{sp-\max}(\theta_i)$ are uniquely determined, and $\alpha_i^{sp-\max}(\theta_i) > \alpha_i^{wp-\max}(\theta_i)$ for any $\theta_i \leq \theta_i^c$ since $\alpha_i^{sp-\max}(\theta_i) = 1$ while $\alpha_i^{wp-\max}(\theta_i) < 1$ within this interval. Next, for $\theta_i > \theta_i^c$, we define:
\begin{equation}\label{def-alphac}
\alpha_c \triangleq \inf\{\alpha \in [0,1]: \alpha^3 > 1 - \alpha\},
\end{equation}
which satisfies $\alpha_c > \frac{1}{2}$. The stationary condition \eqref{type2intFOC} and inequality \eqref{larger23-lm4pf} imply that $(\alpha_i^{sp-\max}(\theta_i))^3 + \alpha_i^{sp-\max}(\theta_i) - 1 \geq 0$ for any $\theta_i > \theta_i^c$, thus:
\begin{equation}\label{mmmmmmmmmmmmmmmm0}
\alpha_i^{sp-\max}(\theta_i) \geq \alpha_c > \frac{1}{2},
\end{equation}
for any $\theta_i > \theta_i^c$. We further claim that:

\begin{claim}\label{claimalphawlsl}
$\alpha_i^{wp-\max}(\theta_i)< \alpha_c$ for any $\theta_i > \theta_i^c$.
\end{claim}
\begin{proof}[Proof of Claim \ref{claimalphawlsl}]
We proceed by contradiction. Suppose there exists $\theta_i^0 > \theta_i^c$ such that $\alpha_i^{wp-\max}(\theta_i^0) \geq \alpha_c$. Then $(\alpha_i^{wp-\max}(\theta_i^0))^3 \geq 1 - \alpha_i^{wp-\max}(\theta_i^0)$ by \eqref{def-alphac}. Rearranging yields:
\begin{equation}\label{mmmmmmmmmmmmm2}
\frac{\alpha_i^{wp-\max}(\theta_i^0)}{1 - \alpha_i^{wp-\max}(\theta_i^0)} \geq \left(1 - \alpha_i^{wp-\max}(\theta_i^0)\right)^2 + 3\alpha_i^{wp-\max}(\theta_i^0).
\end{equation}

According to the stationary condition in the winner-pivotal scenario in \eqref{CL-decrease-OptAlpha-00}, we have:
\begin{equation}\label{kkkkkkkkkkkks}
\frac{\left(\alpha_i^{wp-\max}(\theta_i^0)\right)^3 - 3\left(\alpha_i^{wp-\max}(\theta_i^0)\right)^2 + 4\alpha_i^{wp-\max}(\theta_i^0) - 1}{\left(1 - \alpha_i^{wp-\max}(\theta_i^0) + \left(\alpha_i^{wp-\max}(\theta_i^0)\right)^2\right)\cdot \left(1 - \left(\alpha_i^{wp-\max}(\theta_i^0)\right)^2\right)} = \frac{1-F_i(\theta_i^0)}{\theta_i^0 \cdot f_i(\theta_i^0)} < 1,
\end{equation}
where the inequality follows from $\frac{d}{d\theta_i} (\frac{1 - F_i(\theta_i)}{\theta_i \cdot f_i(\theta_i)} ) < 0$ and $\frac{1 - F_i(\theta_i^c)}{\theta_i^c \cdot f_i(\theta_i^c)} = 1$, given that $\theta_i^0 > \theta_i^c$. Inequality \eqref{kkkkkkkkkkkks} can be rewritten as:
\begin{equation}\label{mmmmmmmmmmmmm1}
\frac{\alpha_i^{wp-\max}(\theta_i^0)}{1 - \alpha_i^{wp-\max}(\theta_i^0)} < \left(1 - \alpha_i^{wp-\max}(\theta_i^0)\right)^2 + 1 + \left(\alpha_i^{wp-\max}(\theta_i^0)\right)^3.
\end{equation}
Combining \eqref{mmmmmmmmmmmmm2} and \eqref{mmmmmmmmmmmmm1} gives:
\begin{equation*}
\begin{split}
\left(1 - \alpha_i^{wp-\max}(\theta_i^0)\right)^2 + 3\alpha_i^{wp-\max}(\theta_i^0) \leq \frac{\alpha_i^{wp-\max}(\theta_i^0)}{1 - \alpha_i^{wp-\max}(\theta_i^0)}  < \left(1 - \alpha_i^{wp-\max}(\theta_i^0)\right)^2 + 1 + \left(\alpha_i^{wp-\max}(\theta_i^0)\right)^3,
\end{split}
\end{equation*}
implying $3\alpha_i^{wp-\max}(\theta_i^0) < 1 + (\alpha_i^{wp-\max}(\theta_i^0))^3$. This contradicts $3\alpha_i^{wp-\max}(\theta_i^0) >\allowbreak 1 + (\alpha_i^{wp-\max}(\theta_i^0))^3$, which holds since $\alpha_i^{wp-\max}(\theta_i^0) \geq \alpha_c$ and $\alpha_c > \frac{1}{2}$ by \eqref{mmmmmmmmmmmmmmmm0}, given that $\theta_i^0 > \theta_i^c$.
\end{proof}

Hence, combining Claim \ref{claimalphawlsl} and inequality \eqref{mmmmmmmmmmmmmmmm0}, we obtain:
\begin{equation*}
\alpha_i^{sp-\max}(\theta_i) \geq \alpha_c  > \alpha_i^{wp-\max}(\theta_i),
\end{equation*}
for any $\theta_i > \theta_i^c$. Combined with the result that $\alpha_i^{sp-\max}(\theta_i) > \alpha_i^{wp-\max}(\theta_i)$ for any $\theta_i \leq \theta_i^c$, this establishes the desired result.

\section{Proof of Proposition \ref{Prop-Effort} \label{pf-Prop-Effort}}

\paragraph{Case (1):} From the efforts chosen by both parties at optimum under winner-pivotal collaboration given in \eqref{Optefforttype2-opt}, we derive:
\begin{equation*}
\frac{e^{wp-b}_i(\theta_i)}{e^{wp-s}_i(\theta_i)} = \frac{\frac{\left(1-\alpha_i^{wp^*}(\theta_i)\right) \cdot \theta_i}{1 - \left(1-\alpha_i^{wp^*}(\theta_i)\right) \cdot \alpha_i^{wp^*}(\theta_i)}}{\frac{\left(1-\alpha_i^{wp^*}(\theta_i)\right) \cdot \alpha_i^{wp^*}(\theta_i)\cdot \theta_i}{1 - \left(1-\alpha_i^{wp^*}(\theta_i)\right) \cdot \alpha_i^{wp^*}(\theta_i)}} = \frac{1}{\alpha_i^{wp^*}(\theta_i)} > 1,
\end{equation*}
where the inequality follows from the non-negativity of both parties' efforts and $0 < \alpha_i^{wp^*}(\theta_i) <1$ (cf. Lemma \ref{Lm-characterize-alphamax}). Hence, $e^{wp-b}_i(\theta_i) > e^{wp-s}_i(\theta_i)$ for all $\theta_i \in [\underline{\theta}_i, \overline{\theta}_i]$. For the seller-pivotal scenario, using \eqref{Optefforttype2-2-2} and an analogous argument, we obtain:
\begin{equation*}
\frac{e^{sp-s}_i(\theta_i)}{e^{sp-b}_i(\theta_i)} = \frac{\frac{\alpha_i^{sp^*}(\theta_i) \cdot \theta_i}{1 - \left(1-\alpha_i^{sp^*}(\theta_i)\right) \cdot \alpha_i^{sp^*}(\theta_i)}}{\frac{\left(1-\alpha_i^{sp^*}(\theta_i)\right) \cdot \alpha_i^{sp^*}(\theta_i) \cdot \theta_i}{1 - \left(1-\alpha_i^{sp^*}(\theta_i)\right) \cdot \alpha_i^{sp^*}(\theta_i)}} = \frac{1}{1-\alpha_i^{sp^*}(\theta_i)} > 1,
\end{equation*}
by Lemma \ref{lm-sellerlead-alphamax}, which implies $e^{sp-s}_i(\theta_i) > e^{sp-b}_i(\theta_i)$ for all $\theta_i \in [\underline{\theta}_i, \overline{\theta}_i]$.

\paragraph{Case (2):} Comparing the seller's effort across scenarios, $e^{wp-s}_i(\theta_i)$ and $e^{sp-s}_i(\theta_i)$. For $\theta_i\in [\underline{\theta}_i, \theta_i^c]$, we have $\alpha_i^{sp^*}(\theta_i) = 1$, which implies:
\begin{equation*}
e^{sp-s}_i(\theta_i) = \theta_i > \frac{\left(1-\alpha_i^{wp^*}(\theta_i)\right) \cdot \alpha_i^{wp^*}(\theta_i)}{1 - \left(1-\alpha_i^{wp^*}(\theta_i)\right) \cdot \alpha_i^{wp^*}(\theta_i)}\cdot \theta_i = e^{wp-s}_i(\theta_i),
\end{equation*}
where the inequality follows from the properties of $\alpha_i^{wp^*}(\theta_i)$ established in Lemma \ref{Lm-characterize-alphamax}. For $\theta_i\in [\theta_i^c, \overline{\theta}_i]$, both $\alpha_i^{wp^*}(\theta_i),\alpha_i^{sp^*}(\theta_i) \in (0,1)$. We define:
\begin{equation*}
g(\alpha) \triangleq \frac{\alpha \cdot \theta_i}{1 - \left(1-\alpha\right) \cdot \alpha},
\end{equation*}
which is non-negative and strictly increasing since $g'(\alpha) = \frac{1-\alpha^2 }{\left(1 - \alpha + \alpha^2\right)^2}\cdot \theta_i > 0$. Therefore, from \eqref{Optefforttype2-opt} and \eqref{Optefforttype2-2-2}, we obtain:
\begin{equation*}
e^{sp-s}_i(\theta_i) = g(\alpha_i^{sp^*}) > g(\alpha_i^{wp^*}) > (1- \alpha_i^{wp^*}) \cdot g(\alpha_i^{wp^*}) = e^{wp-s}_i(\theta_i),
\end{equation*}
where the first inequality follows from Proposition \ref{CompareAlpha}.

Next, comparing the winner's effort across scenarios, $e^{wp-b}_i(\theta_i)$ and $e^{sp-b}_i(\theta_i)$. For $\theta_i\in [\underline{\theta}_i, \theta_i^c]$, we have $\alpha_i^{sp^*}(\theta_i) = 1$, which implies:
\begin{equation*}
e^{sp-b}_i(\theta_i) = 0 < \frac{\left(1-\alpha_i^{wp^*}(\theta_i)\right) \cdot \theta_i}{1 - \left(1-\alpha_i^{wp^*}(\theta_i)\right) \cdot \alpha_i^{wp^*}(\theta_i)} = e^{wp-b}_i(\theta_i).
\end{equation*}

For $\theta_i\in [\theta_i^c, \overline{\theta}_i]$, both $\alpha_i^{wp^*}(\theta_i),\alpha_i^{sp^*}(\theta_i) \in (0,1)$. Similarly, we define:
\begin{equation*}
h(\alpha) \triangleq \frac{(1-\alpha) \cdot \theta_i}{1 - \left(1-\alpha\right) \cdot \alpha},
\end{equation*}
which is non-negative and strictly decreasing since $h'(\alpha) = \frac{\alpha\cdot(\alpha-2)}{\left(1 - \alpha + \alpha^2\right)^2}\cdot \theta_i < 0$. Therefore, from \eqref{Optefforttype2-opt} and \eqref{Optefforttype2-2-2}, we obtain:
\begin{equation*}
e^{sp-b}_i(\theta_i) = \alpha_i^{sp^*} \cdot h(\alpha_i^{sp^*}) < h(\alpha_i^{sp^*}) < h(\alpha_i^{wp^*}) = e^{wp-b}_i(\theta_i),
\end{equation*}
where the second inequality follows from Proposition \ref{CompareAlpha}.

\section{Proof of Proposition \ref{ppproprev} \label{ppproprev-pf}}

According to Theorems \ref{ThmCompl1-type2} and \ref{ThmCompl1-typeSL-seller}, the seller's expected equilibrium payoffs under optimal linear mechanisms satisfy:
\begin{equation*}
\begin{split}
\text{Rev}^{wp}(\bm{q}^{wp^*}, \bm{c}^{wp^*}, \bm{\pi}^{wp^*}) =  & \int_{\bm{\Theta}} \max\{\phi_1^{wp}(\theta_1), \cdots, \phi_n^{wp}(\theta_n)\} \bm{f}(\bm{\theta}) d\bm{\theta},\\
\text{Rev}^{sp}(\bm{q}^{sp^*}, \bm{c}^{sp^*}, \bm{\pi}^{sp^*}) = & \int_{\bm{\Theta}} \max\{\phi_1^{sp}(\theta_1), \cdots, \phi_n^{sp}(\theta_n)\} \bm{f}(\bm{\theta}) d\bm{\theta},
\end{split}
\end{equation*}
that attain the upper bound values, respectively, as presented in Lemmas \ref{uppbdd} and \ref{S-lead-LM}.

To prove that seller-pivotal collaboration yields higher revenue than winner-pivotal collaboration, it suffices to establish pointwise dominance of the virtual surplus functions. Specifically, we show that for any bidder $i$ and type $\theta_i \in \Theta_i$:
\begin{equation*}
\phi_i^{sp}(\theta_i) > \phi_i^{wp}(\theta_i),
\end{equation*}
where the virtual surpluses are defined in \eqref{phiphiphi-type2-pf} and \eqref{phiphiphi-sp-pf}.


Algebraic computation shows that:
\begin{equation}\label{rev-sym-A}
A^{wp}(\alpha) = A^{sp}(1-\alpha),
\end{equation}
for any $\alpha\in [0,1]$. Next, we compute the difference:
\begin{equation*}\label{rev-diff-B}
B^{wp}(\alpha) - B^{sp}(1-\alpha) = \frac{(1-\alpha)^2 - \alpha\cdot(1-\alpha)}{D(\alpha)} = \frac{(1-\alpha)\cdot(1 - 2\alpha)}{D(\alpha)},
\end{equation*}
where $D(\alpha) \triangleq 1 - \alpha + \alpha^2$. Thus:
\begin{equation}\label{largeBBB}
B^{wp}(\alpha) \geq B^{sp}(1-\alpha) \iff \alpha \leq \frac{1}{2}.
\end{equation}

We examine the difference function $\Delta \Psi(\alpha) \triangleq \Psi_i^{sp}(\alpha, \theta_i) - \Psi_i^{wp}(\alpha, \theta_i)$, expressed as:
\begin{equation}\label{rev-diff-eq}
\Delta \Psi(\alpha) = \Delta A(\alpha) \cdot \theta_i \cdot \left( \theta_i - (1-\alpha)\cdot\frac{1-F_i(\theta_i)}{f_i(\theta_i)} \right),
\end{equation}
for any $\alpha\in[0,1]$, where $\Delta A(\alpha) = \frac{2\alpha-1}{1-\alpha+\alpha^2}$.

We then analyze the sign of $\Delta \Psi(\alpha)$ at the optimum. Specifically, let $\hat{\alpha} \triangleq \alpha_i^{wp-\max}(\theta_i)$; by definition, 
\begin{equation}\label{pfDelta-phi}
\phi_i^{wp}(\theta_i) = \Psi_i^{wp}(\hat{\alpha}, \theta_i) = A^{wp}(\hat{\alpha}) \cdot \theta_i^2 - B^{wp}(\hat{\alpha}) \cdot \frac{1-F_i(\theta_i)}{f_i(\theta_i)}\cdot \theta_i.
\end{equation}

\textbf{Case 1:} For $\hat{\alpha} < \frac{1}{2}$. Given that 
$\alpha_i^{sp-\max} \ge 1/2$ by Lemma \ref{lm-sellerlead-alphamax}, 
consider $\tilde{\alpha} = 1-\hat{\alpha} > 1/2$. Thus, 
$A^{sp}(\tilde{\alpha}) = A^{sp}(1-\hat{\alpha}) = A^{wp}(\hat{\alpha})$ 
by \eqref{rev-sym-A}. Since $\hat{\alpha} < \frac{1}{2}$, we have 
$B^{sp}(\tilde{\alpha}) < B^{wp}(\hat{\alpha})$ strictly by \eqref{largeBBB}. 
Since $\frac{1-F_i(\theta_i)}{f_i(\theta_i)}\cdot\theta_i > 0$ for 
$\theta_i > \underline{\theta}_i$, this gives:
\begin{equation*}
\Psi_i^{sp}(\tilde{\alpha}, \theta_i) = A^{wp}(\hat{\alpha})\cdot\theta_i^2 
- B^{sp}(\tilde{\alpha})\cdot\frac{1-F_i(\theta_i)}{f_i(\theta_i)}\cdot
\theta_i > \phi_i^{wp}(\theta_i),
\end{equation*}
from \eqref{pfDelta-phi}. Combined with $\phi_i^{sp}(\theta_i) \ge 
\Psi_i^{sp}(\tilde{\alpha}, \theta_i)$, we have 
$\phi_i^{sp}(\theta_i) > \phi_i^{wp}(\theta_i)$.

\textbf{Case 2:} For $\hat{\alpha} = \frac{1}{2}$. Then 
$\Delta A(\hat{\alpha}) = 0$, so $\Delta\Psi(\hat{\alpha}) = 0$ by 
\eqref{rev-diff-eq}. However, since $\alpha_i^{sp-\max}(\theta_i) \geq 
\alpha_c > \frac{1}{2}$ by \eqref{mmmmmmmmmmmmmmmm0} and Lemma 
\ref{lm-sellerlead-alphamax}, we have $\alpha_i^{sp-\max}(\theta_i) \neq 
\hat{\alpha}$, and:
\begin{equation*}
\phi_i^{sp}(\theta_i) = \Psi_i^{sp}(\alpha_i^{sp-\max}(\theta_i), \theta_i) 
> \Psi_i^{sp}(\hat{\alpha}, \theta_i) = \Psi_i^{wp}(\hat{\alpha}, \theta_i) 
= \phi_i^{wp}(\theta_i),
\end{equation*}
where the strict inequality follows from $\alpha_i^{sp-\max}$ being the 
unique maximizer of $\Psi_i^{sp}$.

\textbf{Case 3:} For $\hat{\alpha} > \frac{1}{2}$. Since $\Delta A(\hat{\alpha}) = \frac{2\hat{\alpha}-1}{1-\hat{\alpha}+\hat{\alpha}^2} > 0$, by \eqref{rev-diff-eq}, proving $\Psi_i^{sp}(\hat{\alpha},\theta_i) > \Psi_i^{wp}(\hat{\alpha},\theta_i)$ reduces to showing:
\begin{equation}\label{cond-pos}
\theta_i - (1-\hat{\alpha})\cdot\frac{1-F_i(\theta_i)}{f_i(\theta_i)} > 0.
\end{equation}

Since $\hat{\alpha} = \alpha_i^{wp\text{-max}}(\theta_i)$ is an interior maximizer of $\Psi_i^{wp}(\alpha,\theta_i)$, the FOC gives:
\begin{equation}\label{wp-foc-sub}
\frac{1-F_i(\theta_i)}{f_i(\theta_i)} = \frac{\left(A^{wp}(\hat{\alpha})\right)'}{\left(B^{wp}(\hat{\alpha})\right)'} \cdot \theta_i.
\end{equation}

Substituting \eqref{wp-foc-sub} into \eqref{cond-pos}, and using $\left(B^{wp}(\alpha)\right)' = -\frac{(1-\alpha)(1+\alpha)}{D(\alpha)^2} < 0$, condition \eqref{cond-pos} is equivalent to:
\begin{equation}\label{deriv-ineq}
\left|\left(A^{wp}(\hat{\alpha})\right)'\right| < \frac{(1+\hat{\alpha})\cdot D(\hat{\alpha})^2}{(1-\hat{\alpha})\cdot D(\hat{\alpha})^2} = \frac{1+\hat{\alpha}}{1-\hat{\alpha}}.
\end{equation}

Direct computation gives $\left|\left(A^{wp}(\alpha)\right)'\right| = \frac{\alpha^3 - 3\alpha^2 + 4\alpha - 1}{D(\alpha)^3}$, where the numerator is non-negative for $\hat{\alpha} > \frac{1}{2}$ by the stationary condition \eqref{CL-decrease-OptAlpha-00}. Using $(1+\alpha)\cdot D(\alpha) = 1 + \alpha^3$, condition \eqref{deriv-ineq} becomes:
\begin{equation*}
\frac{\hat{\alpha}^3 - 3\hat{\alpha}^2 + 4\hat{\alpha} - 1}{D(\hat{\alpha})^3} < \frac{1+\hat{\alpha}}{D(\hat{\alpha})^2}
\iff \hat{\alpha}^3 - 3\hat{\alpha}^2 + 4\hat{\alpha} - 1 < (1+\hat{\alpha})\cdot D(\hat{\alpha}) = 1 + \hat{\alpha}^3
\iff -3\hat{\alpha}^2 + 4\hat{\alpha} - 2 < 0,
\end{equation*}
which holds since $-3\hat{\alpha}^2 + 4\hat{\alpha} - 2 = -3\!\left(\hat{\alpha} - \tfrac{2}{3}\right)^2 - \tfrac{2}{3} < 0$ for all $\hat{\alpha}$.

Thus \eqref{cond-pos} is satisfied, and consequently:
\begin{equation*}
\phi_i^{sp}(\theta_i) \geq \Psi_i^{sp}(\hat{\alpha}, \theta_i) > \Psi_i^{wp}(\hat{\alpha}, \theta_i) = \phi_i^{wp}(\theta_i).
\end{equation*}

Combining both cases, we establish strict pointwise dominance $\phi_i^{sp}(\theta_i) > \phi_i^{wp}(\theta_i)$ for all $\theta_i$. Since the seller's expected revenue is the expectation of the maximum virtual surplus of the allocated bidder:
\begin{equation*}
\overline{\text{Rev}}^{sp} = \mathbb{E}_{\bm{\theta}} \left[ \max_i \phi_i^{sp}(\theta_i) \right] > \mathbb{E}_{\bm{\theta}} \left[ \max_i \phi_i^{wp}(\theta_i) \right] = \overline{\text{Rev}}^{wp}.
\end{equation*}
This completes the proof.

\section{Proof of Proposition \ref{ppproooo222} \label{winnerlead-IM}}

\paragraph{Verify equilibrium.} We first establish that each bidder's drop-out strategy,\linebreak $\widetilde{\beta}_i^{wp}(\theta_i| p_1, \cdots, p_k)$, is increasing in $\theta_i$ for any $k$, since $\phi_i^{wp}(\cdot)$ is an increasing function by Lemma \ref{increasing-of-surplus}.

Given the history of exit prices $p_1 \leq \cdots \leq p_k$, an effective deviation depends on the winner's identity. Using the winner's value-sharing payment $\widetilde{\alpha}_i^{wp}(\cdot)$ from \eqref{oryoryroy}, which is independent of other bidders' exit strategies, we compute bidder $i$'s ex-post payoff upon winning after exiting at price $p_n$. Based on the winner-pivotal structure from equation \eqref{Vstructure-1} and expression \eqref{expostpayoff-winner}, this payoff is given as follows:
\begin{equation*}
\begin{split}
\hat{U}_i^{wp}(p_n, \theta_i) = & \left(1 - \widetilde{\alpha}_i^{wp}(p_n)\right) \cdot \left(
\begin{gathered}
\left(\theta_i + \frac{(1-\widetilde{\alpha}_i^{wp}(p_n)) \cdot \widetilde{\alpha}_i^{wp}(p_n) \cdot \mathcal{G}(p_n)}{1 - (1-\widetilde{\alpha}_i^{wp}(p_n)) \cdot \widetilde{\alpha}_i^{wp}(p_n)}\right) \\
\cdot (1-\widetilde{\alpha}_i^{wp}(p_n))\cdot \left(\theta_i + \frac{(1-\widetilde{\alpha}_i^{wp}(p_n)) \cdot \widetilde{\alpha}_i^{wp}(p_n)\cdot \mathcal{G}(p_n)}{1 - (1-\widetilde{\alpha}_i^{wp}(p_n)) \cdot \widetilde{\alpha}_i^{wp}(p_n)}\right)
\end{gathered}
\right)\\
& - \frac{1}{2}\cdot \left((1-\widetilde{\alpha}_i^{wp}(p_n))\cdot \left(\theta_i + \frac{(1-\widetilde{\alpha}_i^{wp}(p_n)) \cdot \widetilde{\alpha}_i^{wp}(p_n)\cdot \mathcal{G}(p_n)}{1 - (1-\widetilde{\alpha}_i^{wp}(p_n)) \cdot \widetilde{\alpha}_i^{wp}(p_n)}\right)\right)^2\\
& - \left(
\begin{gathered}
\frac{\left(1-\widetilde{\alpha}_i^{wp}(p_n)\right)^2}{2}\cdot \left(\mathcal{G}(p_n) + \frac{(1-\widetilde{\alpha}_i^{wp}(p_n)) \cdot \widetilde{\alpha}_i^{wp}(p_n)\cdot \mathcal{G}(p_n)}{1 - (1-\widetilde{\alpha}_i^{wp}(p_n)) \cdot \widetilde{\alpha}_i^{wp}(p_n)}\right)^2\\
- \int_{\mathcal{G}(p_{n-1})}^{\mathcal{G}(p_n)} \left\{
\begin{gathered}
\left(1- \alpha_i^{wp-\max}(\tau)\right)^2 \\
\cdot \left(\tau + \frac{(1-\alpha_i^{wp-\max}(\tau)) \cdot \alpha_i^{wp-\max}(\tau)\cdot \tau}{1 - (1-\alpha_i^{wp-\max}(\tau)) \cdot \alpha_i^{wp-\max}(\tau)}\right)
\end{gathered}\right\}d\tau
\end{gathered}
\right),
\end{split}
\end{equation*}
given that the seller and winner choose post-auction efforts as specified in \eqref{opopop1111}, and the seller holds the posterior belief outlined in \eqref{xcxc}. Applying the envelope theorem yields winning bidder $i$'s ex-post equilibrium payoff as:
\begin{equation*}
\hat{U}_i^{wp}(\theta_i) = \int_{\mathcal{G}(p_{n-1})}^{\theta_i} \left(1- \alpha_i^{wp-\max}(\tau)\right)^2 \cdot \left(\tau + \frac{(1-\alpha_i^{wp-\max}(\tau)) \cdot \alpha_i^{wp-\max}(\tau)\cdot \tau}{1 - (1-\alpha_i^{wp-\max}(\tau)) \cdot \alpha_i^{wp-\max}(\tau)}\right) d\tau.
\end{equation*}

Then, for any deviating strategy, that is, exiting at $p_n = \phi_i^{wp}(\hat{\theta}_i) \neq \phi_i^{wp}(\theta_i)$, without loss of generality assume $\theta_i > \hat{\theta}_i$. We demonstrate that $\hat{U}_i^{wp}(\theta_i) - \hat{U}_i^{wp}(p_n, \theta_i) \geq 0$ using an approach analogous to that in \eqref{GIC-2-type2}, by setting $q_i^{wp^*}=1$ upon winning and applying Claim \ref{CL-derivative}. Hence, each bidder adopting the drop-out strategy $\widetilde{\beta}_i^{wp}(\theta_i|p_1, \cdots, p_k) = \phi_i^{wp}(\theta_i)$ constitutes a quasi-dominant strategy equilibrium.

\paragraph{Implementation.} In this equilibrium, the last active bidder wins, meaning the winning probability is $\mathds{1}_{\{i: p_i \geq \max \bm{p}_{-i}\}}$. This aligns with the optimal allocation rule $q^{wp^*}_i(\bm{\theta})$ specified in \eqref{tm11-11}, given that each bidder exits when the price reaches $\phi_i^{wp}(\theta_i)$, thereby selecting the bidder with the highest virtual surplus.

We now implement the optimal payment rule specified in \eqref{oppalphatype2} and \eqref{OptmechType2-cash}. The winner's value-sharing rule remains unchanged since $\widetilde{\alpha}_i^{wp} = \alpha_i^{wp-\max}$ according to \eqref{oryoryroy}. For the cash payments $(\nu_i^{wp^*-w}, \nu_i^{wp^*-l})$, based on the winner and losing bidders' cash payments in this ascending auction as specified in \eqref{tttttt}, we define each bidder $i$'s minimum type to guarantee winning as $\hat{\theta}_i(\bm{\theta}_{-i}) \triangleq (\phi_i^{wp})^{-1}(\max_{j \in \mathcal{N}\backslash \{i\}} \phi_j^{wp}(\theta_j)) = \mathcal{G}(p_{n-1})$ for any profile of opponent types $\bm{\theta}_{-i}$. We then derive each bidder $i$'s expected cash payment at equilibrium as follows:
\begin{equation}\label{pf-wl-im-cash}
\begin{split}
& \int_{\bm{\Theta}_{-i}} \left\{
\begin{gathered}
\frac{\left(1-\alpha_i^{wp-\max}(\theta_i)\right)^2}{2}\cdot \left(\theta_i + \frac{(1-\alpha_i^{wp-\max}(\theta_i)) \cdot \alpha_i^{wp-\max}(\theta_i)\cdot \theta_i}{1 - (1-\alpha_i^{wp-\max}(\theta_i)) \cdot \alpha_i^{wp-\max}(\theta_i)}\right)^2\\
- \int_{\hat{\theta}_i(\bm{\theta}_{-i})}^{\theta_i} \left(1- \alpha_i^{wp-\max}(\tau)\right)^2 \cdot \left(\tau + \frac{(1-\alpha_i^{wp-\max}(\tau)) \cdot \alpha_i^{wp-\max}(\tau)\cdot \tau}{1 - (1-\alpha_i^{wp-\max}(\tau)) \cdot \alpha_i^{wp-\max}(\tau)}\right) d\tau
\end{gathered}
\right\}\\
& \cdot \mathds{1}_{\{\theta_i: \theta_i\geq \hat{\theta}_i(\bm{\theta}_{-i})\}}\bm{f}_{-i}(\bm{\theta}_{-i}) d\bm{\theta}_{-i}\\
= & \int_{\bm{\Theta}_{-i}} \left\{
\begin{gathered}
\frac{\left(1-\alpha_i^{wp-\max}(\theta_i)\right)^2}{2}\cdot \left(\theta_i + \frac{(1-\alpha_i^{wp-\max}(\theta_i)) \cdot \alpha_i^{wp-\max}(\theta_i)\cdot \theta_i}{1 - (1-\alpha_i^{wp-\max}(\theta_i)) \cdot \alpha_i^{wp-\max}(\theta_i)}\right)^2 \\
\cdot\, q^{wp^*}_i(\theta_i,\bm{\theta}_{-i})\\
- \int_{\underline{\theta}_i}^{\theta_i} \left(1- \alpha_i^{wp-\max}(\tau)\right)^2 \cdot \left(\tau + \frac{(1-\alpha_i^{wp-\max}(\tau)) \cdot \alpha_i^{wp-\max}(\tau)\cdot \tau}{1 - (1-\alpha_i^{wp-\max}(\tau)) \cdot \alpha_i^{wp-\max}(\tau)}\right) \\
\cdot\, q^{wp^*}_i(\tau,\bm{\theta}_{-i})\, d\tau
\end{gathered}
\right\}\\
& \cdot \bm{f}_{-i}(\bm{\theta}_{-i}) d\bm{\theta}_{-i},
\end{split}
\end{equation}
which coincides with the optimal cash payment rule specified in \eqref{OptmechType2-cash}. Finally, the seller's Dirac-form posterior belief and the monotonicity of $\mathcal{G}(\cdot)$ ensured by Lemma \ref{increasing-of-surplus} align with the Dirac-measure optimal signal realization rule. Moreover, the post-auction effort choices of the seller and winner $\widetilde{\bm{e}}^{wp}(p_n,p_{n-1})$ correspond to the bilateral post-auction efforts under $(\bm{q}^{wp^*},\allowbreak \bm{c}^{wp^*}, \bm{\pi}^{wp^*})$ as specified in \eqref{Optefforttype2-opt}.

We have thus established that the equilibrium of this ascending auction implements the optimal direct linear mechanism.

\section{Proof of Proposition \ref{sllead-IM} \label{selllead-IM}}

Similar to winner-pivotal collaboration, consider bidder $i$'s ex-post payoff upon winning, $\hat{U}_i^{sp}(p_n, \theta_i)$. Following Appendix \ref{winnerlead-IM}, we show that $\hat{U}_i^{sp}(\theta_i) \geq \hat{U}_i^{sp}(p_n, \theta_i)$ for any deviation $p_n = \phi_i^{sp}(\hat{\theta}_i) \neq \phi_i^{sp}(\theta_i)$, establishing that no profitable deviation exists. Details are omitted.

The approximate implementation of the optimal mechanism follows the same approach as in the winner-pivotal case. Specifically, exiting at the bidder's virtual surplus implements the optimal allocation rule. The winner's value-sharing rule approximates $\alpha_i^{sp^*}(\cdot | \bm{\theta})$ as $\epsilon \rightarrow 0$. Additionally, the cash payments $(\widetilde{t}_i^{sp-w}, \widetilde{t}_i^{sp-l})$ implement the optimal cash payment rule in expected form as given in \eqref{sl-lead-op-cash} as $\epsilon \to 0$, analogous to \eqref{pf-wl-im-cash}.

\section{Seller's expected payoff upper bound \label{opt-mech-nesting-pf}}

\paragraph{Seller's expected payoff.}

We derive each party's equilibrium aftermarket strategies, $\varepsilon^s_i \in \Delta( E^s_i)$ and $\varepsilon^b_i \in \Delta( E^b_i)$, which also degenerates to a Dirac measure. For any signal realization $s \in S_i$ and seller's posterior belief $\mu$, we derive $(\varepsilon^{opt-s}_i, \varepsilon^{opt-b}_i) = (\delta_{\{e^{opt-s}_i\}}, \delta_{\{e^{opt-b}_i\}})$ from:
\begin{equation*}
\begin{split}
e^{opt-s}_i \in & \argmax_{e^s_i \in E^s_i}\left\{
\begin{gathered}
\alpha_i\cdot \mathbb{E}_{\widetilde{\theta}_i}\left[\left.\underbrace{\mathbb{E}_{\varepsilon^{opt-b}_i(\widetilde{\theta}_i)}\left[\left(\zeta \cdot \widetilde{\theta}_i + e^s_i\right)\cdot \left((1-\zeta) \cdot \widetilde{\theta}_i + e^b_i\right)\right]}_{= \left(\zeta \cdot \widetilde{\theta}_i + e^s_i\right)\cdot \left((1-\zeta) \cdot \widetilde{\theta}_i + e^{opt-b}_i(\widetilde{\theta}_i)\right)}\right|\mu\right] - \frac{1}{2}\cdot \left(e^s_i\right)^2
\end{gathered}
\right\} \\
& \xlongequal{F.O.C.} \alpha_i\cdot \mathbb{E}_{\widetilde{\theta}_i}\left[\left.\left((1-\zeta) \cdot \widetilde{\theta}_i + e^{opt-b}_i(\widetilde{\theta}_i)\right) \right|\mu\right],\\
e^{opt-b}_i \in & \argmax_{e^b_i \in E^b_i}\left\{
\begin{gathered}
\left(1-\alpha_i\right)\cdot \underbrace{\mathbb{E}_{\varepsilon^{opt-s}_i(\mu)}\left[\left(\zeta \cdot \theta_i + e^s_i\right)\cdot \left((1-\zeta) \cdot \theta_i + e^b_i\right)\right]}_{= \left(\zeta \cdot \theta_i + e^{opt-s}_i(\mu)\right)\cdot \left((1-\zeta) \cdot \theta_i + e^b_i\right)} - \frac{1}{2}\cdot \left(e^b_i\right)^2
\end{gathered}
\right\} \\
& \xlongequal{F.O.C.} \left(1-\alpha_i\right) \cdot \left(\zeta \cdot \theta_i + e^{opt-s}_i(\mu)\right).
\end{split}
\end{equation*}

Solving the above system of equations yields:
\begin{equation*}
\begin{split}
e^{opt-s}_i = & e^{opt-s}_i(\mu) = \frac{\alpha_i\cdot (1-\alpha_i\cdot \zeta) \cdot  \mathbb{E}\left[\left.\widetilde{\theta}_i\right|\mu\right]}{1 - (1-\alpha_i) \cdot \alpha_i},\\
e^{opt-b}_i = & e^{opt-b}_i(\theta_i,\mu) = (1-\alpha_i)\cdot \left(\zeta \cdot \theta_i + \frac{\alpha_i\cdot (1-\alpha_i\cdot \zeta) \cdot  \mathbb{E}\left[\left.\widetilde{\theta}_i\right|\mu\right]}{1 - (1-\alpha_i) \cdot \alpha_i}\right),
\end{split}
\end{equation*}
where $\mathbb{E}[\widetilde{\theta}_i | \mu] = \mathbb{E}[\widetilde{\theta}_i | \mu(\cdot |i, \bm{c}, s)]$ denotes the posterior mean of the winner's type. Therefore, for any direct linear mechanism $(\bm{q}, \bm{c}, \bm{\pi})$, bidder $i$'s expected payoff is given by:
\begin{equation*}
\begin{split}
& U_i^{opt}(\hat{\theta}_i, e^{opt-b}_i, e^{opt-s}_i, \theta_i)\\
= & \int_{\bm{\Theta}_{-i}} q_i(\hat{\theta}_i,\bm{\theta}_{-i}) \int_{\bm{\mathcal{T}}_{-i}^l}\int_{[0,1]\times \mathcal{T}_i^w} \\
& \cdot \int_{S_i} \left\{
\begin{gathered}
(1-\alpha_i) \cdot \left(\zeta \cdot \theta_i + \frac{\alpha_i\cdot (1-\alpha_i\cdot \zeta) \cdot \mathbb{E}\left[\left.\widetilde{\theta}_i\right|\mu(\cdot |i, \bm{c}, s)\right]}{1 - (1-\alpha_i) \cdot \alpha_i}\right) \\
\cdot \left((1-\zeta) \cdot \theta_i + (1-\alpha_i)\cdot \left(\zeta \cdot \theta_i + \frac{\alpha_i\cdot (1-\alpha_i\cdot \zeta) \cdot \mathbb{E}\left[\left.\widetilde{\theta}_i\right|\mu(\cdot |i, \bm{c}, s)\right]}{1 - (1-\alpha_i) \cdot \alpha_i}\right)\right)\\
- \frac{1}{2}\cdot\left((1-\alpha_i)\cdot \left(\zeta \cdot \theta_i + \frac{\alpha_i\cdot (1-\alpha_i\cdot \zeta) \cdot  \mathbb{E}\left[\left.\widetilde{\theta}_i\right|\mu(\cdot |i, \bm{c}, s)\right]}{1 - (1-\alpha_i) \cdot \alpha_i}\right)\right)^2 - t_i^w
\end{gathered}
\right\} \pi_i(d s|i,\bm{c},\hat{\theta}_i,\bm{\theta}_{-i}) \\
&\cdot \kappa_i(d\alpha_i,dt_i^w|\hat{\theta}_i,\bm{\theta}_{-i})\bm{\nu}_{-i}^l(d\bm{t}_{-i}^l|\hat{\theta}_i,\bm{\theta}_{-i})\bm{f}_{-i}(\bm{\theta}_{-i}) d\bm{\theta}_{-i}\\
& + \int_{\bm{\Theta}_{-i}} \left(1 - q_i(\hat{\theta}_i,\bm{\theta}_{-i})\right) \int_{\mathcal{T}_i^l} \left(-t_i^l\right)\nu_i^{l}(d t_i^l|\hat{\theta}_i,\bm{\theta}_{-i})\bm{f}_{-i}(\bm{\theta}_{-i}) d\bm{\theta}_{-i}.
\end{split}
\end{equation*}
where $\bm{\mathcal{T}}_{-i}^l \equiv \bigtimes_{j \in \mathcal{N}, j \neq i} \mathcal{T}_j^l$ and $\bm{\nu}_{-i}^l(d\bm{t}_{-i}^l|\hat{\theta}_i,\bm{\theta}_{-i}) \equiv \prod_{j \in \mathcal{N}, j \neq i} \nu_j^l(d t_j^l|\hat{\theta}_i,\bm{\theta}_{-i})$.

In a feasible direct linear mechanism, the incentive compatibility condition for each bidder requires that $U_i^{opt}(\theta_i) = \max_{\hat{\theta}_i \in \Theta_i} U_i^{opt}(\hat{\theta}_i, e^{opt-b}_i, e^{opt-s}_i, \theta_i)$. By applying the general envelope theorem (see \cite{Milgrom2002}), we find that an incentive-compatible direct mechanism necessarily satisfies:
\begin{equation}\label{Spayoff-opt}
U_i^{opt}(\theta_i) = U_i^{opt}(\underline{\theta}_i) + \int_{\underline{\theta}_i}^{\theta_i} W_i(\tau) d\tau,
\end{equation}
where $W_i(\tau)$ is defined as follows:
\begin{equation*}
\begin{split}
W_i(\tau) \triangleq & \int_{\bm{\Theta}_{-i}} q_i(\tau,\bm{\theta}_{-i}) \int_{\bm{\mathcal{T}}_{-i}^l}\int_{[0,1]\times \mathcal{T}_i^w} \\
& \cdot \int_{S_i} \left\{
\begin{gathered}
(1-\alpha_i) \cdot \left(\zeta \cdot \tau + \frac{\alpha_i\cdot (1-\alpha_i\cdot \zeta) \cdot \mathbb{E}\left[\left.\widetilde{\theta}_i\right|\mu(\cdot |i, \bm{c}, s)\right]}{1 - (1-\alpha_i) \cdot \alpha_i}\right)\\
\cdot \left((1-\zeta) + (1-\alpha_i) \cdot \zeta\right)
\end{gathered}
\right\} \pi_i(d s|i,\bm{c},\tau,\bm{\theta}_{-i}) \\
&\cdot \kappa_i(d\alpha_i,dt_i^w|\tau,\bm{\theta}_{-i})\bm{\nu}_{-i}^l(d\bm{t}_{-i}^l|\tau,\bm{\theta}_{-i})\bm{f}_{-i}(\bm{\theta}_{-i}) d\bm{\theta}_{-i}.
\end{split}
\end{equation*}

Meanwhile, according to the expression above, bidder $i$'s expected payoff from truthfully reporting is:
\begin{equation}\label{nnbnbnbn-opt}
\begin{split}
& U_i^{opt}(\theta_i) = \int_{\bm{\Theta}_{-i}} q_i(\theta_i,\bm{\theta}_{-i}) \int_{\bm{\mathcal{T}}_{-i}^l}\int_{[0,1]\times \mathcal{T}_i^w} \\
& \cdot \int_{S_i} \left\{
\begin{gathered}
(1-\alpha_i) \cdot \left(\zeta \cdot \theta_i + \frac{\alpha_i\cdot (1-\alpha_i\cdot \zeta) \cdot \mathbb{E}\left[\left.\widetilde{\theta}_i\right|\mu(\cdot |i, \bm{c}, s)\right]}{1 - (1-\alpha_i) \cdot \alpha_i}\right) \\
\cdot \left((1-\zeta) \cdot \theta_i + (1-\alpha_i)\cdot \left(\zeta \cdot \theta_i + \frac{\alpha_i\cdot (1-\alpha_i\cdot \zeta) \cdot \mathbb{E}\left[\left.\widetilde{\theta}_i\right|\mu(\cdot |i, \bm{c}, s)\right]}{1 - (1-\alpha_i) \cdot \alpha_i}\right)\right)\\
- \frac{1}{2}\cdot\left((1-\alpha_i)\cdot \left(\zeta \cdot \theta_i + \frac{\alpha_i\cdot (1-\alpha_i\cdot \zeta) \cdot  \mathbb{E}\left[\left.\widetilde{\theta}_i\right|\mu(\cdot |i, \bm{c}, s)\right]}{1 - (1-\alpha_i) \cdot \alpha_i}\right)\right)^2
\end{gathered}
\right\} \pi_i(d s|i,\bm{c},\theta_i,\bm{\theta}_{-i}) \\
&\cdot \kappa_i(d\alpha_i,dt_i^w|\theta_i,\bm{\theta}_{-i}) \bm{\nu}_{-i}^l(d\bm{t}_{-i}^l|\theta_i,\bm{\theta}_{-i}) \bm{f}_{-i}(\bm{\theta}_{-i}) d\bm{\theta}_{-i}\\
& - \underbrace{\int_{\bm{\Theta}_{-i}} \left\{
\begin{gathered}
q_i(\theta_i,\bm{\theta}_{-i}) \int_{\mathcal{T}_i^w} t_i^w\nu_i^{w}(d t_i^w|\theta_i,\bm{\theta}_{-i})\\
+ \left(1 - q_i(\theta_i,\bm{\theta}_{-i})\right)\int_{\mathcal{T}_i^l} t_i^l\nu_i^{l}(d t_i^l|\theta_i,\bm{\theta}_{-i})
\end{gathered}
\right\}\bm{f}_{-i}(\bm{\theta}_{-i}) d\bm{\theta}_{-i}}_{\triangleq P_i(\theta_i)},
\end{split}
\end{equation}
where the marginal probabilities, $\nu_i^{w}(\cdot|\bm{\theta}) \triangleq \int_0^1 \kappa_i(d\alpha_i,\cdot|\bm{\theta})$ and $\sigma_i(\cdot|\bm{\theta}) \triangleq \int_{\mathcal{T}_i^w} \kappa_i(\cdot,dt_i^w|\bm{\theta})$ are defined in the same way as previous two scenarios.

Combining \eqref{Spayoff-opt} and \eqref{nnbnbnbn-opt}, we derive the expression for $P_i(\theta_i)$ and obtain:
\begin{equation*}
\begin{split}
& \sum_{i\in\mathcal{N}}\int_{\Theta_i} P_i(\theta_i)f_i(\theta_i)d\theta_i =  \sum_{i\in\mathcal{N}} \int_{\bm{\Theta}} q_i(\bm{\theta}) \cdot \int_{\bm{\mathcal{T}}_{-i}^l} \int_{[0,1]\times \mathcal{T}_i^w} \\
&\cdot \int_{S_i} \left\{
\begin{gathered}
\frac{1}{2}\cdot (1-\alpha_i)^2 \cdot \left(\zeta \cdot \theta_i + \frac{\alpha_i\cdot (1-\alpha_i\cdot \zeta) \cdot \mathbb{E}\left[\left.\widetilde{\theta}_i\right|\mu(\cdot |i, \bm{c}, s)\right]}{1 - (1-\alpha_i) \cdot \alpha_i}\right)^2\\
\cdot \left((1-\zeta) + (1-\alpha_i) \cdot \zeta\right)^2\\
- (1-\alpha_i) \cdot \left(\zeta \cdot \theta_i + \frac{\alpha_i\cdot (1-\alpha_i\cdot \zeta) \cdot \mathbb{E}\left[\left.\widetilde{\theta}_i\right|\mu(\cdot |i, \bm{c}, s)\right]}{1 - (1-\alpha_i) \cdot \alpha_i}\right)\\
\cdot \left((1-\zeta) + (1-\alpha_i) \cdot \zeta\right) \cdot \frac{1 - F_i(\theta_i)}{f_i(\theta_i)}
\end{gathered}
\right\}  \pi_i(d s|i,\bm{c},\bm{\theta}) \\
&\cdot \kappa_i(d\alpha_i,dt_i^w|\bm{\theta})\bm{\nu}_{-i}^l(d\bm{t}_{-i}^l|\bm{\theta}) \bm{f}(\bm{\theta}) d \bm{\theta} - \sum_{i\in\mathcal{N}}U_i^{opt}(\underline{\theta}_i).
\end{split}
\end{equation*}

Then, the seller's expected payoff comprises the sum of bidders' cash payments and value shares, net of effort costs, expressed as:
\begin{equation}\label{REV-opt}
\begin{split}
& \text{Rev}^{opt}(\bm{q}, \bm{c}, \bm{\pi}) \\
= &\sum_{i\in\mathcal{N}}\int_{\Theta_i} P_i(\theta_i)f_i(\theta_i)d\theta_i + \sum_{i\in\mathcal{N}} \int_{\bm{\Theta}} q_i(\bm{\theta}) \cdot \int_{\bm{\mathcal{T}}_{-i}^l} \int_{[0,1]\times \mathcal{T}_i^w} \\
& \cdot\int_{S_i} \left\{
\begin{gathered}
\alpha_i \cdot \left(\zeta \cdot \theta_i +  \frac{\alpha_i\cdot (1-\alpha_i\cdot \zeta) \cdot \mathbb{E}\left[\left.\widetilde{\theta}_i\right|\mu(\cdot |i, \bm{c}, s)\right]}{1 - (1-\alpha_i) \cdot \alpha_i}\right) \\
\cdot \left((1-\zeta) \cdot \theta_i + (1-\alpha_i)\cdot \left(\zeta \cdot \theta_i + \frac{\alpha_i\cdot (1-\alpha_i\cdot \zeta) \cdot \mathbb{E}\left[\left.\widetilde{\theta}_i\right|\mu(\cdot |i, \bm{c}, s)\right]}{1 - (1-\alpha_i) \cdot \alpha_i}\right)\right)\\
- \frac{1}{2}\cdot\left(\frac{\alpha_i\cdot (1-\alpha_i\cdot \zeta) \cdot  \mathbb{E}\left[\left.\widetilde{\theta}_i\right|\mu(\cdot |i, \bm{c}, s)\right]}{1 - (1-\alpha_i) \cdot \alpha_i}\right)^2
\end{gathered}
\right\} \pi_i(d s|i,\bm{c},\bm{\theta}) \\
&\cdot \kappa_i(d\alpha_i,dt_i^w|\bm{\theta})\bm{\nu}_{-i}^l(d\bm{t}_{-i}^l|\bm{\theta}) \bm{f}(\bm{\theta}) d \bm{\theta}\\
= &\sum_{i\in\mathcal{N}} \int_{\bm{\Theta}} q_i(\bm{\theta}) \cdot \int_{\bm{\mathcal{T}}_{-i}^l}\int_{[0,1]\times \mathcal{T}_i^w} \\
& \cdot\int_{S_i} \left\{
\begin{gathered}
\left(\zeta - \frac{\zeta^2 \cdot  (1+\alpha_i^2)}{2}\right)\cdot \theta_i^2 \\
+ (1-\zeta \cdot \alpha_i^2) \cdot \theta_i \cdot \frac{\alpha_i\cdot (1-\alpha_i\cdot \zeta) \cdot  \mathbb{E}\left[\left.\widetilde{\theta}_i\right|\mu(\cdot |i, \bm{c}, s)\right]}{1 - (1-\alpha_i) \cdot \alpha_i}\\
- \frac{\alpha_i^2}{2}\cdot\left(\frac{\alpha_i\cdot (1-\alpha_i\cdot \zeta) \cdot  \mathbb{E}\left[\left.\widetilde{\theta}_i\right|\mu(\cdot |i, \bm{c}, s)\right]}{1 - (1-\alpha_i) \cdot \alpha_i}\right)^2 \\
- (1-\alpha_i) \cdot \zeta \cdot (2 - \zeta \cdot (1+\alpha_i)) \cdot \theta_i \cdot \frac{1 - F_i(\theta_i)}{f_i(\theta_i)} \\
- (1-\alpha_i) \cdot (1-\alpha_i \cdot \zeta) \cdot \frac{\alpha_i\cdot (1-\alpha_i\cdot \zeta) \cdot  \mathbb{E}\left[\left.\widetilde{\theta}_i\right|\mu(\cdot |i, \bm{c}, s)\right]}{1 - (1-\alpha_i) \cdot \alpha_i} \cdot \frac{1 - F_i(\theta_i)}{f_i(\theta_i)}
\end{gathered}
\right\} \pi_i(d s|i,\bm{c},\bm{\theta}) \\
&\cdot \kappa_i(d\alpha_i,dt_i^w|\bm{\theta}) \bm{\nu}_{-i}^l(d\bm{t}_{-i}^l|\bm{\theta}) \bm{f}(\bm{\theta}) d \bm{\theta} - \sum_{i\in\mathcal{N}}U_i^{opt}(\underline{\theta}_i).
\end{split}
\end{equation}

\paragraph{Determine the feasible upper bound.} SImilar to the winner- and seller-pivotal cases, we denote $\xi(i, d\bm{c}, ds)$ as the probability measure associated with bidder $i$ winning, where the direct linear mechanism generates payment profile $\bm{c} = (c_i, \bm{t}^l) = (\alpha_i, t_i^w, (t_j^l)_{j\neq i})$ and public signal $s\in S_i$. Conditional on this, from the seller's perspective, bidders' types are drawn according to probability measure $\varrho(\cdot|i, \bm{c}, s)\in \Delta (\bm{\Theta} )$, which implies $\int_{\bm{\Theta}_{-i}}\varrho(d\bm{\theta}|i, \bm{c}, s) = \mu(d\theta_{i}|i, \bm{c}, s)$. The seller's expected equilibrium payoff in \eqref{REV-opt} can then be expressed as:
\begin{equation*}
\begin{split}
& \text{Rev}^{opt}(\bm{q}, \bm{c}, \bm{\pi})\\
= & \sum_{i\in\mathcal{N}} \int_{[0,1]\times \mathcal{T}_i^w \times \bm{\mathcal{T}}_{-i}^l \times S_i} \left\{
\begin{gathered}
\left(\zeta - \frac{\zeta^2 \cdot  (1+\alpha_i^2)}{2}\right)\cdot \mathbb{E}\left[\left.\widetilde{\theta}_i^2\right|\mu(\cdot |i, \bm{c}, s)\right] \\
+ \frac{(1-\zeta \cdot \alpha_i^2) \cdot \alpha_i\cdot (1-\alpha_i\cdot \zeta) \cdot  \mathbb{E}^2\left[\left.\widetilde{\theta}_i\right|\mu(\cdot |i, \bm{c}, s)\right]}{1 - (1-\alpha_i) \cdot \alpha_i}\\
- \frac{\alpha_i^2}{2}\cdot\left(\frac{\alpha_i\cdot (1-\alpha_i\cdot \zeta) \cdot  \mathbb{E}\left[\left.\widetilde{\theta}_i\right|\mu(\cdot |i, \bm{c}, s)\right]}{1 - (1-\alpha_i) \cdot \alpha_i}\right)^2 \\
- (1-\alpha_i) \cdot \zeta \cdot (2 - \zeta \cdot (1+\alpha_i)) \\
\cdot \mathbb{E}\left[\left.\widetilde{\theta}_i\right|\mu(\cdot |i, \bm{c}, s)\right] \cdot \mathbb{E}\left[\left.\frac{1 - F_i(\widetilde{\theta}_i)}{f_i(\widetilde{\theta}_i)}\right|\mu(\cdot |i, \bm{c}, s)\right] \\
- \frac{(1-\alpha_i) \cdot \alpha_i\cdot (1-\alpha_i\cdot \zeta)^2}{1 - (1-\alpha_i) \cdot \alpha_i} \\
\cdot  \mathbb{E}\left[\left.\widetilde{\theta}_i\right|\mu(\cdot |i, \bm{c}, s)\right] \cdot \mathbb{E}\left[\left.\frac{1 - F_i(\widetilde{\theta}_i)}{f_i(\widetilde{\theta}_i)}\right|\mu(\cdot |i, \bm{c}, s)\right]
\end{gathered}
\right\} \xi(i, d\bm{c}, ds)\\
& - \sum_{i\in\mathcal{N}}U_i^{opt}(\underline{\theta}_i).
\end{split}
\end{equation*}

Using the following two inequalities: $\mathbb{E}^2[\widetilde{\theta}_i|\mu(\widetilde{\theta}_i|i, \bm{c}, s)] \leq \mathbb{E}[\widetilde{\theta}_i^2|\mu(\widetilde{\theta}_i|i, \bm{c}, s)]$ by Jensen's inequality, and
\begin{equation*}
\mathbb{E}\left[\left.\frac{F_i(\widetilde{\theta}_i)-1}{f_i(\widetilde{\theta}_i)}\right|\mu(\widetilde{\theta}_i|i, \bm{c}, s)\right] \cdot \mathbb{E}\left[\left.\widetilde{\theta}_i\right|\mu(\widetilde{\theta}_i|i, \bm{c}, s)\right] \leq \mathbb{E}\left[\left.\frac{F_i(\widetilde{\theta}_i)-1}{f_i(\widetilde{\theta}_i)} \cdot \widetilde{\theta}_i\right|\mu(\widetilde{\theta}_i|i, \bm{c}, s)\right].
\end{equation*}
Moreover, individual rationality implies $\sum_{i \in \mathcal{N}} U_i^{opt}(\underline{\theta}_i) \geq 0$ since $U_i^{opt}(\theta_i) \geq 0$ for all $\theta_i \in \Theta_i$. We therefore obtain the following inequality:
\begin{equation*}
\begin{split}
& \text{Rev}^{opt}(\bm{q}, \bm{c}, \bm{\pi}) \\
\leq & \sum_{i\in\mathcal{N}} \int_{[0,1]\times \mathcal{T}_i^w \times \bm{\mathcal{T}}_{-i}^l \times S_i}\left\{
\begin{gathered}
\underbrace{\left(\begin{gathered}\zeta - \frac{\zeta^2 \cdot  (1+\alpha_i^2)}{2} + \frac{\alpha_i \cdot (1-\alpha_i \cdot \zeta)\cdot(1-\zeta\cdot\alpha_i^2)}{1 - (1-\alpha_i) \cdot \alpha_i}\\
- \frac{\alpha_i^4\cdot(1-\alpha_i\cdot\zeta)^2}{2(1 - (1-\alpha_i) \cdot \alpha_i)^2}\end{gathered}\right)}_{\text{non-negative}} \\
\cdot \mathbb{E}\left[\left.\widetilde{\theta}_i^2\right|\mu(\cdot|i, \bm{c}, s)\right] \\
- \underbrace{(1-\alpha_i)\cdot\left(
\zeta\cdot(2 - \zeta\cdot(1+\alpha_i)) + \frac{\alpha_i\cdot(1-\alpha_i\cdot\zeta)^2}{1 - (1-\alpha_i) \cdot \alpha_i}
\right)}_{\text{non-negative}} \\
\cdot \mathbb{E}\left[\left.\frac{1 - F_i(\widetilde{\theta}_i)}{f_i(\widetilde{\theta}_i)}\cdot \widetilde{\theta}_i \right|\mu(\cdot|i, \bm{c}, s)\right]
\end{gathered}
\right\} \xi(i, d\bm{c}, ds)\\
= & \sum_{i\in\mathcal{N}} \cdot \int_{[0,1]\times \mathcal{T}_i^w \times \bm{\mathcal{T}}_{-i}^l \times S_i} \\
& \cdot \int_{\bm{\Theta}} \left\{
\begin{gathered}
\left(\begin{gathered}
\zeta - \frac{\zeta^2 \cdot (1+\alpha_i^2)}{2} + \frac{\alpha_i \cdot (1-\alpha_i \cdot \zeta)\cdot(1-\zeta\cdot\alpha_i^2)}{1 - (1-\alpha_i) \cdot \alpha_i} \\
- \frac{\alpha_i^4\cdot(1-\alpha_i\cdot\zeta)^2}{2(1 - (1-\alpha_i) \cdot \alpha_i)^2}
\end{gathered}\right) \cdot \theta_i^2 \\
- (1-\alpha_i)\cdot\left(
\zeta\cdot(2 - \zeta\cdot(1+\alpha_i))+ \frac{\alpha_i\cdot(1-\alpha_i\cdot\zeta)^2}{1 - (1-\alpha_i) \cdot \alpha_i}
\right) \cdot \frac{1-F_i(\theta_i)}{f_i(\theta_i)}\cdot \theta_i
\end{gathered}
\right\} \varrho(d\bm{\theta}|i, \bm{c}, s)\xi(i, d\bm{c}, ds)\\
= & \sum_{i\in\mathcal{N}}\int_{\bm{\Theta}} q_i(\bm{\theta}) \cdot \int_{\bm{\mathcal{T}}_{-i}^l} \int_{[0,1]\times \mathcal{T}_i^w} \\
& \cdot\int_{S_i} \left\{
\begin{gathered}
\left(\begin{gathered}
\zeta - \frac{\zeta^2 \cdot (1+\alpha_i^2)}{2} + \frac{\alpha_i \cdot (1-\alpha_i \cdot \zeta)\cdot(1-\zeta\cdot\alpha_i^2)}{1 - (1-\alpha_i) \cdot \alpha_i} \\
- \frac{\alpha_i^4\cdot(1-\alpha_i\cdot\zeta)^2}{2(1 - (1-\alpha_i) \cdot \alpha_i)^2}
\end{gathered}\right) \cdot \theta_i^2 \\
- (1-\alpha_i)\cdot\left(
\zeta\cdot(2 - \zeta\cdot(1+\alpha_i)) + \frac{\alpha_i\cdot(1-\alpha_i\cdot\zeta)^2}{1 - (1-\alpha_i) \cdot \alpha_i}
\right) \cdot \frac{1-F_i(\theta_i)}{f_i(\theta_i)}\cdot \theta_i
\end{gathered}
\right\} \pi_i(d s|i,\bm{c},\bm{\theta}) \\
&\cdot \kappa_i(d\alpha_i,dt_i^w|\bm{\theta}) \bm{\nu}_{-i}^l(d\bm{t}_{-i}^l|\bm{\theta}) \bm{f}(\bm{\theta}) d \bm{\theta}\\
= & \sum_{i\in\mathcal{N}}\int_{\bm{\Theta}} q_i(\bm{\theta}) \cdot \int_{0}^1 \left\{
\begin{gathered}
\left(
\zeta - \frac{\zeta^2 \cdot (1+\alpha_i^2)}{2} + \frac{\alpha_i \cdot (1-\alpha_i \cdot \zeta)\cdot(1-\zeta\cdot\alpha_i^2)}{1 - (1-\alpha_i) \cdot \alpha_i}
- \frac{\alpha_i^4\cdot(1-\alpha_i\cdot\zeta)^2}{2(1 - (1-\alpha_i) \cdot \alpha_i)^2}
\right) \cdot \theta_i^2 \\
- (1-\alpha_i)\cdot\left(
\zeta\cdot(2 - \zeta\cdot(1+\alpha_i)) + \frac{\alpha_i\cdot(1-\alpha_i\cdot\zeta)^2}{1 - (1-\alpha_i) \cdot \alpha_i}
\right) \cdot \frac{1-F_i(\theta_i)}{f_i(\theta_i)}\cdot \theta_i
\end{gathered}
\right\}\\
& \cdot \sigma_i(d\alpha_i|\bm{\theta}) \bm{f}(\bm{\theta}) d\bm{\theta},
\end{split}
\end{equation*}
where the equalities follow from a change-of-measure and the marginal probability measure $\sigma_i(\cdot|\bm{\theta})$. We next maximize over $\sigma_i(\cdot|\bm{\theta})$. For any profile of bidders' types $\bm{\theta}\in \bm{\Theta}$, we solve:
\begin{equation*}
\max_{\sigma_i(d\alpha_i|\bm{\theta})}\int_{0}^1 \underbrace{\left\{
\begin{gathered}
\left(
\zeta - \frac{\zeta^2 \cdot (1+\alpha_i^2)}{2} + \frac{\alpha_i \cdot (1-\alpha_i \cdot \zeta)\cdot(1-\zeta\cdot\alpha_i^2)}{1 - (1-\alpha_i) \cdot \alpha_i} - \frac{\alpha_i^4\cdot(1-\alpha_i\cdot\zeta)^2}{2(1 - (1-\alpha_i) \cdot \alpha_i)^2}
\right) \cdot \theta_i^2\\
- (1-\alpha_i)\cdot\left(
\zeta\cdot(2 - \zeta\cdot(1+\alpha_i)) + \frac{\alpha_i\cdot(1-\alpha_i\cdot\zeta)^2}{1 - (1-\alpha_i) \cdot \alpha_i}
\right) \cdot \frac{1-F_i(\theta_i)}{f_i(\theta_i)}\cdot \theta_i
\end{gathered}
\right\}}_{\text{defined as }\Psi_i^{opt}(\alpha_i,\theta_i,\zeta)} \sigma_i(d\alpha_i|\bm{\theta}).
\end{equation*}

The maximizer, denoted $\sigma_i^{opt-\max}(\cdot|\bm{\theta}) \in \Delta([0,1])$, is a Dirac measure that places unit mass on $\alpha_i^{opt-\max} = \alpha_i^{opt-\max}(\theta_i,\zeta)$, where $\alpha_i^{opt-\max}(\theta_i,\zeta) \in \argmax_{\alpha_i\in[0,1]} \Psi_i^{opt}(\alpha_i,\theta_i,\zeta)$. We define the corresponding maximal value as bidder $i$'s virtual surplus:
\begin{equation}\label{mAxxx_opt}
\phi_i^{opt}(\theta_i,\zeta) \triangleq \int_{0}^1 \Psi_i^{opt}(\alpha_i,\theta_i,\zeta)\sigma_i^{opt-\max}(d\alpha_i|\bm{\theta}) = \Psi_i^{opt}(\alpha_i^{opt-\max},\theta_i,\zeta).
\end{equation}

Hence, for any feasible direct linear mechanism $(\bm{q}, \bm{c}, \bm{\pi})$, the seller's expected payoff is therefore bounded above by $\overline{\text{Rev}}^{opt}(\bm{q}, \bm{c}, \bm{\pi})$, characterized by:
\begin{equation}\label{uuuuuperbddOPT}
\begin{split}
\overline{\text{Rev}}^{opt}(\bm{q}, \bm{c}, \bm{\pi}) \triangleq & \max_{\substack{0\leq q_i(\bm{\theta}) \leq 1;\\ \sum_{i\in\mathcal{N}} q_i(\bm{\theta}) \leq 1}} \sum_{i\in\mathcal{N}}\int_{\bm{\Theta}} q_i(\bm{\theta}) \cdot \phi_i^{opt}(\theta_i,\zeta) \bm{f}(\bm{\theta}) d\bm{\theta}\\
= & \int_{\bm{\Theta}} \max\left\{\phi_1^{opt}(\theta_1,\zeta), \cdots, \phi_i^{opt}(\theta_i,\zeta),\cdots,\phi_n^{opt}(\theta_n,\zeta)\right\} \bm{f}(\bm{\theta}) d\bm{\theta}.
\end{split}
\end{equation}

\section{Proof of Proposition \ref{ppppp07} \label{ppppp07-pff}}




We show the existence of an optimal mechanism by providing a mechanism which is feasible and attaining the upper bound payoff. Specifically,
\begin{proposition}\label{ThmCompl1-opt}
Under assumptions in Proposition \ref{ppppp07}, a feasible optimal mechanism exists, characterized as:
\begin{enumerate}
\item[(i)] The allocation rule is given by:
\begin{equation}\label{tm11-11-opt}
q_i^{opt^*}(\bm{\theta}) = \begin{cases}
1, &\text{if } \phi^{opt}_i(\theta_i,\zeta)> \max_{j\neq i}\phi^{opt}_j(\theta_j,\zeta)\\
0, & \text{otherwise}
\end{cases}
\end{equation}
for all $i \in \mathcal{N}$, where $\phi_i^{opt}(\theta_i,\zeta)$ is each bidder's virtual surplus. Ties are broken uniformly at random.

\item[(ii)] The payment rule $\bm{c}^{opt^*}$ induces a deterministic value-sharing rule for the winner:
\begin{equation*}
\alpha_i^{opt^*}(\theta_i,\zeta) = \alpha_i^{opt-\max}(\theta_i,\zeta).
\end{equation*}
The cash payments upon winning and losing are determined by $\nu_i^{opt^*-w}(\cdot | \bm{\theta})$ and $\nu_i^{opt^*-l}(\cdot | \bm{\theta})$ that satisfy:
\begin{equation*}
\begin{split}
& \int_{\bm{\Theta}_{-i}} \left\{
\begin{gathered}
q_i^{opt^*}(\theta_i,\bm{\theta}_{-i}) \cdot \int_{\mathcal{T}_i^w} t_i^w\nu_i^{opt^*-w}(d t_i^w|\theta_i,\bm{\theta}_{-i})  \\
+ \left(1 - q_i^{opt^*}(\theta_i,\bm{\theta}_{-i})\right)\int_{\mathcal{T}_i^l} t_i^l\nu_i^{opt^*-l}(d t_i^l|\theta_i,\bm{\theta}_{-i})
\end{gathered}
\right\} \bm{f}_{-i}(\bm{\theta}_{-i}) d\bm{\theta}_{-i} \\
= & \int_{\bm{\Theta}_{-i}} \left\{
\begin{gathered}
q_i^{opt^*}(\theta_i,\bm{\theta}_{-i}) \cdot \frac{(1-\alpha_i^{opt-\max}(\theta_i,\zeta))^2}{2} \\
\cdot \left(\zeta \cdot \theta_i + \frac{\alpha_i^{opt-\max}(\theta_i,\zeta) \cdot (1-\alpha_i^{opt-\max}(\theta_i,\zeta)\cdot \zeta) \cdot \theta_i}{1 - (1-\alpha_i^{opt-\max}(\theta_i,\zeta)) \cdot \alpha_i^{opt-\max}(\theta_i,\zeta)}\right)^2 \\
-\int_{\underline{\theta}_i}^{\theta_i} q_i^{opt^*}(\tau,\bm{\theta}_{-i}) \cdot (1-\alpha_i^{opt-\max}(\tau,\zeta))^2\\
\cdot \left(\zeta \cdot \tau + \frac{\alpha_i^{opt-\max}(\tau,\zeta) \cdot (1-\alpha_i^{opt-\max}(\tau,\zeta)\cdot \zeta) \cdot \tau}{1 - (1-\alpha_i^{opt-\max}(\tau,\zeta)) \cdot \alpha_i^{opt-\max}(\tau,\zeta)}\right) d\tau
\end{gathered}
\right\} \bm{f}_{-i}(\bm{\theta}_{-i}) d\bm{\theta}_{-i}.
\end{split}
\end{equation*}
This expression (integration form) shows that the cash payments for both the winner and the losers can be constructed in multiple ways.

\item[(iii)] The signal realization rule $\pi_i^{opt^*}(\cdot|i,\bm{c},\bm{\theta}) \in \Delta(S_i)$, conditional on the profile of bidders' reported types, winner's identity, and bidders' payments, is a Dirac measure:
\begin{equation*}
\pi_i^{opt^*}(\cdot|i,\bm{c},\bm{\theta}) = \delta_{\{\theta_i\}},
\end{equation*}
which assigns probability mass one to the winner's truthfully reported type $\theta_i$.
\end{enumerate}
\end{proposition}

\begin{proof}
we show Proposition \ref{ThmCompl1-opt} through two steps: feasibility and optimality.

\paragraph{Feasibility.} Note that bidder $i$'s expected payoff from truthfully reporting $\theta_i$ under the mechanism $(\bm{q}^{opt^*}, \bm{c}^{opt^*}, \bm{\pi}^{opt^*})$:
\begin{equation*}
U_i^{opt}(\theta_i,\zeta) = \int_{\bm{\Theta}_{-i}}\int_{\underline{\theta}_i}^{\theta_i} \left\{
\begin{gathered}
q_i^{opt^*}(\tau,\bm{\theta}_{-i}) \cdot (1-\alpha_i^{opt-\max}(\tau,\zeta))^2 \\
\cdot \left(\zeta \cdot \tau + \frac{\alpha_i^{opt-\max}(\tau,\zeta) \cdot (1-\alpha_i^{opt-\max}(\tau,\zeta)\cdot \zeta) \cdot \tau}{1 - (1-\alpha_i^{opt-\max}(\tau,\zeta)) \cdot \alpha_i^{opt-\max}(\tau,\zeta)}\right)
\end{gathered}
\right\}d\tau \, \bm{f}_{-i}(\bm{\theta}_{-i}) d\bm{\theta}_{-i},
\end{equation*}
which is increasing in $\theta_i$ and satisfies individual rationality, thus, the demanding feasibility requires that:
\begin{equation}\label{GIC-opt-target}
\begin{split}
& U_i^{opt}(\theta_i,\zeta) - U_{i}^{opt}(\hat{\theta}_{i},\theta_{i},\zeta)\\
= & \int_{\bm{\Theta}_{-i}} \underbrace{\left\{
\begin{gathered}
\int_{\hat{\theta}_i}^{\theta_i} \left\{
\begin{gathered}
q_i^{opt^*}(\tau,\bm{\theta}_{-i}) \cdot (1-\alpha_i^{opt-\max}(\tau,\zeta))^2 \\
\cdot \left(\zeta \cdot \tau + \frac{\alpha_i^{opt-\max}(\tau,\zeta) \cdot (1-\alpha_i^{opt-\max}(\tau,\zeta)\cdot \zeta) \cdot \tau}{1 - (1-\alpha_i^{opt-\max}(\tau,\zeta)) \cdot \alpha_i^{opt-\max}(\tau,\zeta)}\right)
\end{gathered}
\right\} d\tau \\
- q_i^{opt^*}(\hat{\theta}_i,\bm{\theta}_{-i}) \cdot \frac{(1-\alpha_i^{opt-\max}(\hat{\theta}_i,\zeta))^2}{2} \\
\cdot \left(
\begin{gathered}
\theta_i + \hat{\theta}_i + 2\zeta\cdot(\theta_i + \hat{\theta}_i) \\
+ 2\cdot\frac{\alpha_i^{opt-\max}(\hat{\theta}_i,\zeta) \cdot (1-\alpha_i^{opt-\max}(\hat{\theta}_i,\zeta)\cdot \zeta) \cdot \hat{\theta}_i}{1 - (1-\alpha_i^{opt-\max}(\hat{\theta}_i,\zeta)) \cdot \alpha_i^{opt-\max}(\hat{\theta}_i,\zeta)}
\end{gathered}
\right) \cdot \left(\theta_i - \hat{\theta}_i\right)
\end{gathered}
\right\}}_{\triangleq \mathcal{I}^{opt}(\theta_i,\zeta)} \bm{f}_{-i}(\bm{\theta}_{-i}) d\bm{\theta}_{-i} \geq  0,
\end{split}
\end{equation}
for any $\theta_i, \hat{\theta}_i\in\Theta_i$ (without loss of generality, assume $\theta_i > \hat{\theta}_i$).

Similar to the proof of feasibility in winner- and seller-pivotal cases, the important intermediate step is:
\begin{equation}\label{vvvvvopt}
\frac{d}{d \theta_i} \left((1-\alpha_i^{opt-\max}(\theta_i,\zeta))^2 \cdot \left(\zeta \cdot \theta_i + \frac{\alpha_i^{opt-\max}(\theta_i,\zeta) \cdot (1-\alpha_i^{opt-\max}(\theta_i,\zeta)\cdot \zeta) \cdot \theta_i}{1 - (1-\alpha_i^{opt-\max}(\theta_i,\zeta)) \cdot \alpha_i^{opt-\max}(\theta_i,\zeta)}\right)\right) \geq 0,
\end{equation}
which critically rely on $\frac{d\alpha_i^{opt-\max}(\theta_i,\zeta)}{d\theta_i} \leq 0$.

\begin{claim}\label{optCLLLLL1}
$\alpha_i^{opt-\max}(\theta_i,\zeta)$ is non-increasing in $\theta_i$, i.e.,
$\frac{d\alpha_i^{opt-\max}(\theta_i,\zeta)}{d\theta_i} \leq 0$
\end{claim}
\begin{proof}
We first derive:
\begin{equation*}
\begin{split}
\left.\frac{\partial\Psi_i^{opt}}{\partial\alpha_i}\right|_{\alpha_i=0} = & A'(0,\zeta) \theta_i^2 - B'(0,\zeta) \frac{1-F_i(\theta_i)}{f_i(\theta_i)} \theta_i\\
= & \theta_i^2 + 2\zeta(1-\zeta) \frac{1-F_i(\theta_i)}{f_i(\theta_i)} \theta_i \geq \theta_i^2 \geq \underline{\theta}_i^2 > 0.
\end{split}
\end{equation*}

Given that $\frac{\partial\Psi_i^{opt}}{\partial\alpha_i}$ is continuous in $\alpha_i$ and strictly positive at $\alpha_i = 0$ with a uniform lower bound $\underline{\theta}_i^2$, there exists a uniform $\delta > 0$ (independent of $\theta_i, \zeta$) such that:
\begin{equation*}
\frac{\partial\Psi_i^{opt}}{\partial\alpha_i} > \frac{\underline{\theta}_i^2}{2} > 0 \quad \text{for all } \alpha_i \in [0, \delta],
\end{equation*}
which implies the maximizer $\alpha_i^{opt-\max}(\theta_i,\zeta) > \delta$ for all $(\theta_i,\zeta) \in [\underline{\theta}_i,\overline{\theta}_i] \times [0,1]$. Therefore the maximizer $\alpha_i^{opt-\max}$ can either be interior or at the corner $\alpha_i^{opt-\max} = 1$, which coincide with the two cases. When $\alpha_i^{opt-\max}$ is a corner maximizer, it is non-increasing in $\theta_i$; for an interior maximizer, $\alpha_i^* = \alpha_i^{opt-\max}(\theta_i,\zeta)$, by first-order condition (FOC):
\begin{equation}\label{FOC-interior-opt}
\frac{\partial \Psi_i^{opt}}{\partial \alpha_i}\bigg|_{\alpha_i = \alpha_i^*} = A'(\alpha_i^*,\zeta) \cdot \theta_i^2 - B'(\alpha_i^*,\zeta) \cdot H(\theta_i) \cdot \theta_i = 0,
\end{equation}
where primes denote derivatives with respect to $\alpha_i$ and $H(\theta_i) = \frac{1-F_i(\theta_i)}{f_i(\theta_i)}$. Rearranged as:
\begin{equation}\label{FOC-ratio}
\frac{A'(\alpha_i^*,\zeta)}{B'(\alpha_i^*,\zeta)} = \frac{H(\theta_i)}{\theta_i},
\end{equation}
which implies $A'(\alpha_i^*,\zeta)$ and $B'(\alpha_i^*,\zeta)$ have the same sign at the optimum since $\frac{H(\theta_i)}{\theta_i} > 0$. Meanwhile, the following second-order condition (SOC) holds for the interior maximizer:
\begin{equation}\label{SOC-main}
\frac{\partial^2 \Psi_i^{opt}}{\partial \alpha_i^2}\bigg|_{\alpha_i = \alpha_i^*} = A''(\alpha_i^*,\zeta) \cdot \theta_i^2 - B''(\alpha_i^*,\zeta) \cdot H(\theta_i) \cdot \theta_i < 0.
\end{equation}
Substituting \eqref{FOC-ratio} into the SOC \eqref{SOC-main}:
\begin{equation*}
\begin{split}
\text{SOC} &= A''(\alpha_i^*,\zeta) \cdot \theta_i \cdot \left[\frac{B'(\alpha_i^*,\zeta)}{A'(\alpha_i^*,\zeta)} \cdot H(\theta_i)\right] - B''(\alpha_i^*,\zeta) \cdot H(\theta_i) \cdot \theta_i\\
&= H(\theta_i) \cdot \theta_i \cdot \left[ \frac{A''(\alpha_i^*,\zeta) \cdot B'(\alpha_i^*,\zeta)}{A'(\alpha_i^*,\zeta)} - B''(\alpha_i^*,\zeta) \right] < 0,
\end{split}
\end{equation*}
which further implies:
\begin{equation}\label{SOC-simplified}
\frac{A''(\alpha_i^*,\zeta) \cdot B'(\alpha_i^*,\zeta)}{A'(\alpha_i^*,\zeta)} - B''(\alpha_i^*,\zeta) < 0,
\end{equation}
because $H(\theta_i) > 0$ and $\theta_i > 0$. By assumptions in Proposition \ref{ppppp07}, we have $A'(\alpha_i^*,\zeta) < 0$ and $B'(\alpha_i^*,\zeta) <0$. Thus, \eqref{SOC-simplified} yields:
\begin{equation}\label{Delta-inequality}
\frac{A''(\alpha_i^*,\zeta)}{A'(\alpha_i^*,\zeta)} - \frac{B''(\alpha_i^*,\zeta)}{B'(\alpha_i^*,\zeta)} > 0.
\end{equation}

Differentiating the FOC \eqref{FOC-interior-opt} with respect to $\theta_i$ using the implicit function theorem:
\begin{equation*}
\begin{split}
&\frac{d}{d\theta_i}\left(A'(\alpha_i^*,\zeta) \cdot \theta_i^2 - B'(\alpha_i^*,\zeta) \cdot H(\theta_i) \cdot \theta_i\right) = 0 \nonumber\\
\Rightarrow\quad & A''(\alpha_i^*,\zeta) \cdot \frac{d\alpha_i^*}{d\theta_i} \cdot \theta_i^2 + A'(\alpha_i^*,\zeta) \cdot 2\theta_i \nonumber\\
&\quad - B''(\alpha_i^*,\zeta) \cdot \frac{d\alpha_i^*}{d\theta_i} \cdot H(\theta_i) \cdot \theta_i - B'(\alpha_i^*,\zeta) \cdot \left(H'(\theta_i) \cdot \theta_i + H(\theta_i)\right) = 0\\
\Rightarrow\quad & \left(A''(\alpha_i^*,\zeta) \cdot \theta_i^2 - B''(\alpha_i^*,\zeta) \cdot H(\theta_i) \cdot \theta_i\right) \cdot \frac{d\alpha_i^*}{d\theta_i} = B'(\alpha_i^*,\zeta) \cdot \left(H'(\theta_i) \cdot \theta_i + H(\theta_i)\right) - A'(\alpha_i^*,\zeta) \cdot 2\theta_i.
\end{split}
\end{equation*}
Dividing both sides by $\theta_i^2$ and use the FOC \eqref{FOC-ratio}, we obtain:
\begin{equation*}
\begin{split}
& \frac{d\alpha_i^*}{d\theta_i} \left(A''(\alpha_i^*,\zeta) - B''(\alpha_i^*,\zeta) \cdot \frac{H(\theta_i)}{\theta_i}\right) = B'(\alpha_i^*,\zeta) \cdot \frac{H'(\theta_i) \cdot \theta_i + H(\theta_i)}{\theta_i^2} - A'(\alpha_i^*,\zeta) \cdot \frac{2}{\theta_i}\\
\Rightarrow &\quad \left(\frac{A''(\alpha_i^*,\zeta)}{A'(\alpha_i^*,\zeta)} - \frac{B''(\alpha_i^*,\zeta)}{B'(\alpha_i^*,\zeta)}\right) \cdot \frac{d\alpha_i^*}{d\theta_i} = \frac{B'(\alpha_i^*,\zeta)}{A'(\alpha_i^*,\zeta)} \cdot \frac{H'(\theta_i) \cdot \theta_i + H(\theta_i)}{\theta_i^2} - \frac{2}{\theta_i}\\
\Rightarrow &\quad  \underbrace{\left(\frac{A''(\alpha_i^*,\zeta)}{A'(\alpha_i^*,\zeta)} - \frac{B''(\alpha_i^*,\zeta)}{B'(\alpha_i^*,\zeta)}\right)}_{\text{$>0$ by \eqref{Delta-inequality}}} \cdot \frac{d\alpha_i^*}{d\theta_i} = \underbrace{\frac{H'(\theta_i)}{H(\theta_i)} - \frac{1}{\theta_i},}_{\text{$<0$ by log-concavity of $F$}}
\end{split}
\end{equation*}
which implies $\frac{d\alpha_i^*}{d\theta_i} < 0$. Thus, $\alpha_i^{opt-\max}(\theta_i,\zeta)$ is non-increasing in $\theta_i$.
\end{proof}

Therefore, we show the desired result in \eqref{vvvvvopt} through:
\begin{equation*}
\begin{split}
& \frac{d}{d \theta_i} \left((1-\alpha_i^{opt-\max}(\theta_i,\zeta))^2 \cdot \left(\zeta \cdot \theta_i + \frac{\alpha_i^{opt-\max}(\theta_i,\zeta) \cdot (1-\alpha_i^{opt-\max}(\theta_i,\zeta)\cdot \zeta) \cdot \theta_i}{1 - (1-\alpha_i^{opt-\max}(\theta_i,\zeta)) \cdot \alpha_i^{opt-\max}(\theta_i,\zeta)}\right)\right)\\
= & \underbrace{(1-\alpha_i^{opt-\max}(\theta_i,\zeta))^2 \cdot \left(\zeta  + \frac{\alpha_i^{opt-\max}(\theta_i,\zeta) \cdot (1-\alpha_i^{opt-\max}(\theta_i,\zeta)\cdot \zeta)}{1 - (1-\alpha_i^{opt-\max}(\theta_i,\zeta)) \cdot \alpha_i^{opt-\max}(\theta_i,\zeta)}\right)}_{\text{$= H(\alpha_i^{opt-\max},\zeta) \geq 0$}}\\
& + \theta_i \cdot \underbrace{\left((1-\alpha_i^{opt-\max}(\theta_i,\zeta))^2 \cdot \left(\zeta  + \frac{\alpha_i^{opt-\max}(\theta_i,\zeta) \cdot (1-\alpha_i^{opt-\max}(\theta_i,\zeta)\cdot \zeta)}{1 - (1-\alpha_i^{opt-\max}(\theta_i,\zeta)) \cdot \alpha_i^{opt-\max}(\theta_i,\zeta)}\right)\right)'}_{\text{$= H'(\alpha_i^{opt-\max},\zeta) < 0$ by assumption in Proposition \ref{ppppp07}}} \cdot \underbrace{\frac{d\alpha_i^{opt-\max}(\theta_i,\zeta)}{d\theta_i}}_{\text{$\leq 0$ by Claim \ref{optCLLLLL1}}}\\
\geq & 0.
\end{split}
\end{equation*}

Furthermore, we show that:

\begin{claim}\label{claimoptincrease}
The virtual surplus $\phi_i^{opt}(\theta_i, \zeta)$ defined in \eqref{mAxxx_opt} is non-negative and non-decreasing.
\end{claim}
\begin{proof}
We first show the non-negative virtual surplus according to:
\begin{equation}\label{nonnegaGen}
\phi_i^{opt}(\theta_i, \zeta) = \max_{\alpha_i \in [0,1]} \Psi_i^{opt}(\alpha_i, \theta_i, \zeta) \geq \Psi_i^{opt}(1, \theta_i, \zeta) = \frac{1 - \zeta^2}{2} \cdot \theta_i^2 \geq 0,
\end{equation}
for $\zeta \in [0,1]$ and $\theta_i \geq 0$.

For non-decreasing, using the Envelope Theorem, differentiating with respect to $\theta_i$ yields:
\begin{equation*}
\frac{\partial \phi_i^{opt}(\theta_i, \zeta)}{\partial \theta_i} = 2 A \cdot \theta_i - B \cdot \left( \frac{1-F_i(\theta_i)}{f_i(\theta_i)} + \theta_i \cdot \left(\frac{1-F_i(\theta_i)}{f_i(\theta_i)}\right)'\right).
\end{equation*}
According to the non-negativity in \eqref{nonnegaGen}:
\begin{equation*}
A \cdot \theta_i^2 \geq B \cdot \frac{1-F_i(\theta_i)}{f_i(\theta_i)} \cdot \theta_i \implies 2 A \cdot \theta_i \geq 2 B \cdot \frac{1-F_i(\theta_i)}{f_i(\theta_i)},
\end{equation*}
which implies:
\begin{equation*}
\begin{split}
\frac{\partial \phi_i^{opt}(\theta_i, \zeta)}{\partial \theta_i} \geq & 2 B \cdot \frac{1-F_i(\theta_i)}{f_i(\theta_i)} - B \cdot \left( \frac{1-F_i(\theta_i)}{f_i(\theta_i)} + \theta_i \cdot \left(\frac{1-F_i(\theta_i)}{f_i(\theta_i)}\right)'\right)\\
= & B \cdot \left( \frac{1-F_i(\theta_i)}{f_i(\theta_i)} - \theta_i \cdot \left(\frac{1-F_i(\theta_i)}{f_i(\theta_i)}\right)'\right) \geq 0,
\end{split}
\end{equation*}
since $\frac{1-F_i(\theta_i)}{f_i(\theta_i)}$ is decreasing in $\theta_i$ under the log-concavity of $f_i(\cdot)$. Thus, $\frac{\partial \phi_i^{opt}(\theta_i, \zeta)}{\partial \theta_i} \geq 0$.
\end{proof}

Back to the demanding inequality \eqref{GIC-opt-target}, we denote $\Upsilon_i(\theta_i,\zeta) = \theta_i \cdot H(\alpha_i^{opt-\max}(\theta_i,\zeta),\zeta)$, then for $\theta_i > \hat{\theta}_i$, we obtain:
\begin{equation*}
\begin{split}
U_i^{opt}(\theta_i,\zeta) - U_{i}^{opt}(\hat{\theta}_{i},\theta_{i},\zeta) \geq & q_i^{opt^*}(\hat{\theta}_i,\bm{\theta}_{-i}) \cdot \int_{\hat{\theta}_i}^{\theta_i} \Upsilon_i(\tau,\zeta) \;d\tau \\
& - q_i^{opt^*}(\hat{\theta}_i,\bm{\theta}_{-i}) \cdot (1-\alpha_i^{opt-\max}(\hat{\theta}_i,\zeta))^2 \cdot \left(\theta_i - \hat{\theta}_i\right)\\
& \quad \cdot \left(\zeta\cdot\theta_i + \frac{\alpha_i^{opt-\max}(\hat{\theta}_i,\zeta) \cdot (1-\alpha_i^{opt-\max}(\hat{\theta}_i,\zeta)\cdot \zeta) \cdot \hat{\theta}_i}{1 - (1-\alpha_i^{opt-\max}(\hat{\theta}_i,\zeta)) \cdot \alpha_i^{opt-\max}(\hat{\theta}_i,\zeta)}\right) \\
> & q_i^{opt^*}(\hat{\theta}_i,\bm{\theta}_{-i}) \cdot \left[ \int_{\hat{\theta}_i}^{\theta_i} \Upsilon_{i}(\hat{\theta}_i, \zeta) d\tau - \Upsilon_{i}(\hat{\theta}_i, \zeta) \cdot (\theta_i - \hat{\theta}_i) \right] = 0,
\end{split}
\end{equation*}
in which the first inequality follows because $q_i^{opt*}(\theta_i, \bm{\theta}_{-i})$ is increasing in $\theta_i$ (since virtual surplus $\phi_i^{opt}$ is increasing), and the second inequality follows from \eqref{vvvvvopt}, that implies $\Upsilon_{i}(\tau, \zeta) > \Upsilon_{i}(\hat{\theta}_i, \zeta)$ for $\tau > \hat{\theta}_i$.

\paragraph{Optimality.} We demonstrate that this mechanism yields an expected payoff for the seller that attains the upper bound $\overline{\text{Rev}}^{opt}(\bm{q}, \bm{c}, \bm{\pi})$ given in \eqref{uuuuuperbddOPT}. Under the feasible linear mechanism $(\bm{q}^{opt^*}, \bm{c}^{opt^*}, \bm{\pi}^{opt^*})$, using the definitions from \eqref{REV-opt} and substituting the optimal values $\alpha_i^{opt^*} = \alpha_i^{opt-\max}$, we have:
\begin{equation*}
\begin{split}
& \text{Rev}^{opt}(\bm{q}^{opt^*}, \bm{c}^{opt^*}, \bm{\pi}^{opt^*}) \\
= & \sum_{i\in\mathcal{N}} \int_{\bm{\Theta}} q_i^{opt^*}(\bm{\theta}) \cdot \left\{
\begin{gathered}
A(\alpha_i^{opt-\max}(\theta_i, \zeta), \zeta) \cdot \theta_i^2 \\
- B(\alpha_i^{opt-\max}(\theta_i, \zeta), \zeta) \cdot \frac{1-F_i(\theta_i)}{f_i(\theta_i)} \cdot \theta_i
\end{gathered}
\right\} \bm{f}(\bm{\theta}) d\bm{\theta} \\
= & \sum_{i\in\mathcal{N}} \int_{\bm{\Theta}} q_i^{opt^*}(\bm{\theta}) \cdot \Psi_i^{opt}(\alpha_i^{opt-\max}(\theta_i,\zeta),\theta_i,\zeta) \bm{f}(\bm{\theta}) d\bm{\theta}\\
= & \sum_{i\in\mathcal{N}} \int_{\bm{\Theta}} q_i^{opt^*}(\bm{\theta}) \cdot \phi_i^{opt}(\theta_i,\zeta) \bm{f}(\bm{\theta}) d\bm{\theta}.
\end{split}
\end{equation*}

Since the allocation rule $q_i^{opt^*}(\bm{\theta})$ is defined in \eqref{tm11-11-opt} to select the bidder with the highest virtual surplus $\phi_i^{opt}(\theta_i,\zeta)$ (provided it is non-negative and non-decreasing by Claim \ref{claimoptincrease}), the integral is maximized pointwise. Thus:
\begin{equation*}
\begin{split}
\text{Rev}^{opt}(\bm{q}^{opt^*}, \bm{c}^{opt^*}, \bm{\pi}^{opt^*}) & = \int_{\bm{\Theta}} \max\left\{0, \phi_1^{opt}(\theta_1,\zeta), \cdots, \phi_n^{opt}(\theta_n,\zeta)\right\} \bm{f}(\bm{\theta}) d\bm{\theta}\\
& = \overline{\text{Rev}}^{opt}(\bm{q}, \bm{c}, \bm{\pi}).
\end{split}
\end{equation*}
This confirms that the proposed mechanism achieves the upper bound of the seller's expected payoff and is therefore optimal.
\end{proof}

\section{Proof of Proposition \ref{pppp07} \label{pppp07-pf}}

We prove the result through the following four steps.

\paragraph{Step 1: Establish $A(\alpha_i,\zeta) = A(1-\alpha_i,1-\zeta)$.} We denote $D(\alpha_i) = 1 - \alpha_i + \alpha_i^2$ and we have $D(\alpha_i) = D(1-\alpha_i)$. Then, we derive:
\begin{equation}\label{alpand1-alp}
\begin{split}
2 D^2(\alpha_i) \cdot A(\alpha_i,\zeta)
= & \underbrace{2\alpha_i - 2\alpha_i^2 + 2\alpha_i^3 - \alpha_i^4}_{\triangleq C_0(\alpha_i)} + \underbrace{\left(2 -4\alpha_i + 4\alpha_i^2 - 4\alpha_i^3 + 2\alpha_i^4\right)}_{\triangleq C_1(\alpha_i)}\cdot \zeta\\
& + \underbrace{\left(-1 + 2\alpha_i -4\alpha_i^2 + 4\alpha_i^3 -2 \alpha_i^4\right)}_{\triangleq C_2(\alpha_i)}\cdot \zeta^2,
\end{split}
\end{equation}
where the coefficients satisfy:
\begin{equation}\label{alpand1-alp-1}
\begin{cases}
C_2(\alpha_i) = C_2(1-\alpha_i)\\
C_1(\alpha_i) = - C_1(1-\alpha_i) - 2C_2(1-\alpha_i)\\
C_0(\alpha_i) = C_0(1-\alpha_i) + C_1(1-\alpha_i) + C_2(1-\alpha_i)
\end{cases}
\end{equation}

Therefore, \eqref{alpand1-alp} implies that:
\begin{equation*}
\begin{split}
& 2 D^2(1-\alpha_i) \cdot A(1-\alpha_i,1-\zeta)\\
= & C_0(1-\alpha_i) + C_1(1-\alpha_i)\cdot (1-\zeta) + C_2(1-\alpha_i)\cdot (1-\zeta)^2\\
= & \underbrace{C_0(1-\alpha_i) + C_1(1-\alpha_i) + C_2(1-\alpha_i)}_{= C_0(\alpha_i)} + \underbrace{\left(- C_1(1-\alpha_i) - 2C_2(1-\alpha_i)\right)}_{= C_1(\alpha_i)}\cdot \zeta + \underbrace{C_2(1-\alpha_i)}_{= C_2(\alpha_i)} \cdot \zeta^2\\
= & 2 D^2(\alpha_i) \cdot A(\alpha_i,\zeta),
\end{split}
\end{equation*}
where the last equality follows from \eqref{alpand1-alp-1}. Combined with $D(\alpha_i) = D(1-\alpha_i)$, we obtain $A(\alpha_i,\zeta) = A(1-\alpha_i,1-\zeta)$ for any $\alpha_i\in[0,1]$ and $\zeta\in[0,1]$.

\paragraph{Step 2: Establish $\max_{\alpha_i \in [0,1]} A(\alpha_i, \zeta) = \max_{\alpha_i \in [0,1]} A(\alpha_i, 1-\zeta)$.} We denote $S^*$ and $S^{**}$ as the sets of all achievable values of $A$ given parameters $\zeta$ and $1-\zeta$, respectively, i.e.,
\begin{equation*}
\begin{split}
S^* \triangleq & \left\{ v \in \mathbb{R} | \exists \alpha_i \in [0,1] \text{ such that } v = A(\alpha_i, \zeta) \right\},\\
S^{**} \triangleq & \left\{ v \in \mathbb{R} | \exists \alpha_i \in [0,1] \text{ such that } v = A(\alpha_i, 1-\zeta) \right\}.
\end{split}
\end{equation*}

For any $\hat{v} \in S^*$, there exists some $\hat{\alpha}_i \in [0,1]$ such that $\hat{v} = A(\hat{\alpha}_i, \zeta)$, which implies $\hat{v} = A(1-\hat{\alpha}_i, 1-\zeta)$ by Step 1. Let $\hat{\alpha}_i' = 1 - \hat{\alpha}_i \in [0,1]$. Thus, $\hat{v} = A(\hat{\alpha}_i', 1-\zeta)$ for some $\hat{\alpha}_i'\in [0,1]$, implying $\hat{v} \in S^{**}$. Therefore, we obtain $S^{*} \subseteq S^{**}$. Applying similar arguments, we establish the reverse inclusion $S^{**} \subseteq S^{*}$. Thus, we have $S^{*} = S^{**}$, which implies $\max(S^*) = \max(S^{**})$. This yields the desired result.

\paragraph{Step 3: Comparing $B(\alpha_i,\zeta)$ and $B(1-\alpha_i,1-\zeta)$.} We define $E(\zeta) \triangleq \zeta \cdot (2-\zeta)$. We then simplify $B(\alpha_i,\zeta)$ as:
\begin{equation*}
B(\alpha_i,\zeta) = \frac{1-\alpha_i}{D(\alpha_i)} \cdot \left((1-\alpha_i) \cdot E(\zeta) + \alpha_i \cdot E(1-\zeta)\right).
\end{equation*}

Next, we compute the difference $\text{Diff} \triangleq B(\alpha_i,\zeta) - B(1-\alpha_i,1-\zeta)$:
\begin{equation*}
\begin{split}
D(\alpha_i) \cdot \text{Diff} = & (1-\alpha_i) \cdot \left((1-\alpha_i) \cdot E(\zeta) + \alpha_i \cdot E(1-\zeta)\right) - \alpha_i \cdot \left(\alpha_i  \cdot E(1-\zeta) + (1-\alpha_i) \cdot E(\zeta)\right) \\
= & (1-2\alpha_i) \cdot \left((1-\alpha_i) \cdot E(\zeta) + \alpha_i \cdot E(1-\zeta)\right).
\end{split}
\end{equation*}

Note that the term in brackets $(1-\alpha_i) \cdot E(\zeta) + \alpha_i \cdot E(1-\zeta)$ is strictly positive for $\alpha_i \in [0,1]$ and $\zeta \in (0,1)$, and $D(\alpha_i)$ is strictly positive as well. Hence, the sign of the difference $B(\alpha_i,\zeta) - B(1-\alpha_i,1-\zeta)$ is determined entirely by the sign of $1-2\alpha_i$. That is, $B(\alpha_i, \zeta) > B(1-\alpha_i, 1-\zeta)$ if $\alpha_i < \frac{1}{2}$, and $B(\alpha_i, \zeta) < B(1-\alpha_i, 1-\zeta)$ if $\alpha_i > \frac{1}{2}$.

\paragraph{Step 4: Optimal $\zeta^* \leq \frac{1}{2}$.} We prove by contradiction that for any $\zeta' > \frac{1}{2}$, $\zeta^* = 1- \zeta'$ yields strictly higher revenue for the seller. We first claim that $\alpha_i^{\max}(\theta_i, \zeta) < \frac{1}{2}$ for $\zeta > \frac{1}{2}$. Differentiating both sides of \eqref{alpand1-alp} with respect to $\alpha$,
\begin{equation*}
2 D^2(\alpha) \cdot A(\alpha,\zeta) = C_0(\alpha) + C_1(\alpha)\cdot \zeta + C_2(\alpha)\cdot \zeta^2,
\end{equation*}
yields:
\begin{equation*}
2 \left(2 D(\alpha) \cdot D'(\alpha) \cdot A + D^2(\alpha) \cdot \frac{\partial A}{\partial \alpha} \right) = C_0'(\alpha) + C_1'(\alpha) \cdot \zeta + C_2'(\alpha) \cdot \zeta^2.
\end{equation*}

Evaluating at the midpoint $\alpha=\frac{1}{2}$, we obtain:
\begin{equation*}
2 \left( 0 + \left(\frac{3}{4}\right)^2 \cdot\left.\frac{\partial A}{\partial \alpha}\right|_{\alpha=1/2} \right) = \frac{9}{8} \left.\frac{\partial A}{\partial \alpha}\right|_{\alpha=1/2} = 1 - 2\zeta + 0 \cdot \zeta^2 = 1 - 2\zeta,
\end{equation*}
which is negative for $\zeta > \frac{1}{2}$. Therefore, a strictly negative slope at the midpoint $\alpha = \frac{1}{2}$ implies that the peak must occur strictly to the left of the midpoint, i.e., $\alpha_i^{opt-\max}(\theta_i, \zeta) < \frac{1}{2}$.\footnote{Intuitively, from the value creation function in \eqref{Vgeneral}, $V_i = (\zeta \cdot \theta_i + e_i^s)\cdot((1-\zeta)\cdot\theta_i + e_i^b)$. Note that $\zeta \cdot \theta_i$ is strictly larger than $(1-\zeta)\cdot\theta_i$ for $\zeta > \frac{1}{2}$, which implies the marginal productivity of the winner's effort, $\zeta \theta_i + e_i^s$, is higher than that of the seller's effort, $(1-\zeta)\theta_i + e_i^b$. Hence, surplus maximization requires a higher incentive for the bidder than the seller, i.e., $1-\alpha_i > \alpha_i$, which implies $\alpha_i < \frac{1}{2}$.}

Therefore, for any $\zeta' > \frac{1}{2}$, we have the optimal $\alpha'_i = \alpha_i^{opt-\max}(\theta_i, \zeta') < \frac{1}{2}$. Then, according to Steps 1 and 3, there exists $\zeta^* = 1- \zeta'$ and $\alpha_i^* = 1- \alpha_i'$ such that:
\begin{equation*}
\begin{split}
A(\alpha_i^*,\zeta^*) = A(\alpha_i',\zeta'), \\
B(\alpha_i^*,\zeta^*) < B(\alpha_i',\zeta').
\end{split}
\end{equation*}

Therefore, $\phi_i(\theta_i, \zeta^*) > \phi_i(\theta_i, \zeta')$ holds pointwise for all $\theta_i$. Consequently, the expected revenue upper bound $\overline{\text{Rev}}(\bm{q}, \bm{c}, \bm{\pi}|\zeta^*)$ is strictly greater than $\overline{\text{Rev}}(\bm{q}, \bm{c}, \bm{\pi}|\zeta')$ for any $\zeta' > \frac{1}{2}$. Furthermore, if such $\zeta^* \leq \frac{1}{2}$ satisfies the conditions in Proposition \ref{ppppp07}, then the seller will choose this $\zeta^*$ such that $\text{Rev}^*(\zeta^*) = \overline{\text{Rev}}(\bm{q}, \bm{c}, \bm{\pi}|\zeta^*)$.

\end{appendices}


\bigskip
\singlespacing
\bibliographystyle{wang}
\bibliography{Bibliography_202604}

\end{document}